\pdfoutput=1
\documentclass[final,1p,times,authoryear]{elsarticle}
\usepackage[shortlabels]{enumitem}
\setlist{topsep=0.25\baselineskip,partopsep=0pt,itemsep=1pt,parsep=0pt}

\usepackage{amsmath,amsfonts,amssymb,amsthm,thmtools}
\declaretheorem[style=plain,parent=section]{definition}
\declaretheorem[sibling=definition]{theorem}
\declaretheorem[sibling=definition]{corollary}
\declaretheorem[sibling=definition]{proposition}
\declaretheorem[sibling=definition]{lemma}
\declaretheorem[style=remark,sibling=definition,qed={\qedsymbol}]{remark}

\declaretheoremstyle[headpunct={},notebraces={\textbf{--}~}{}]{algorithm}
\declaretheorem[style=algorithm]{problem}
\declaretheorem[style=algorithm]{algorithm}

\defcitealias{GiJeVi03}{ibid.}
\defcitealias{JeNeScVi16}{ibid.}
\defcitealias{ZhoLab12}{ibid.}
\defcitealias{JeNeScVi17}{ibid.}

\usepackage{mdframed}

\usepackage{tikz}
\usetikzlibrary{arrows,positioning} 
\tikzset{
    >=stealth',
}

\usepackage[linkcolor=black,colorlinks=true,citecolor=blue,urlcolor=blue]{hyperref} 
\usepackage[capitalise]{cleveref}
\crefname{subsec}{Subsection}{Subsections}
\Crefname{subsec}{Subsection}{Subsections}
\crefname{problem}{Problem}{Problems}
\Crefname{problem}{Problem}{Problems}

\newcommand{\bigO}[1]{O(#1)} 
\newcommand{\bigOPar}[1]{O\!\left(#1\right)} 
\newcommand{\softO}[1]{O\tilde{~}(#1)} 
\newcommand{\polmultime}[1]{\mathsf{M}(#1)} 
\newcommand{\polmatmultime}[1]{\mathsf{MM}(#1)} 
\newcommand{\appbastime}[1]{\mathsf{MM'}(#1)} 
\newcommand{\expmatmul}{\omega} 
\newcommand{\hypsbal}{\mathcal{H}_{\shifts,\mathrm{bal}}} 
\newcommand{\hypobal}{\mathcal{H}_{\orders}} 
\newcommand{\hypsmin}{\mathcal{H}_{\shifts,\mathrm{min}}} 
\newcommand{\hypsmax}{\mathcal{H}_{\shifts,\mathrm{max}}} 
\newcommand{\hypslmm}{\mathcal{H}_{\mathsf{MM}}} 
\newcommand{\hyppolmul}{\mathcal{H}_{\mathsf{M}}} 

\newcommand{\algoname}[1]{{\normalfont\textsc{#1}}}
\newcommand{\algoword}[1]{\textsf{#1}}
\newcommand{\assign}{\leftarrow}
\newcommand{\comment}[1]{\texttt{\small/* #1 */}}
\newcommand{\eolcomment}[1]{\hfill\texttt{\small// #1}}

\newcommand{\ZZ}{\mathbb{Z}} 
\newcommand{\NN}{\mathbb{Z}_{\ge 0}} 
\newcommand{\ZZp}{\mathbb{Z}_{> 0}} 
\newcommand{\tuple}[1]{\mathbf{#1}}  
\newcommand{\tuplegrk}[1]{\boldsymbol{#1}}  
\newcommand{\subTuple}[2]{{#1}_{#2}} 

\newcommand{\var}{X} 
\newcommand{\field}{\mathbb{K}} 
\newcommand{\polRing}{\field[\var]} 
\newcommand{\module}[1][M]{\mathcal{#1}} 
\newcommand{\rdim}{m} 
\newcommand{\cdim}{n} 
\newcommand{\genRank}{r} 
\newcommand{\storeArg}{} 
\newcommand{\matSpace}[1][\rdim]{\renewcommand\storeArg{#1}\matSpaceAux} 
\newcommand{\polMatSpace}[1][\rdim]{\renewcommand\storeArg{#1}\polMatSpaceAux} 
\newcommand{\matSpaceAux}[1][\storeArg]{\field^{\storeArg \times #1}} 
\newcommand{\polMatSpaceAux}[1][\storeArg]{\polRing^{\storeArg \times #1}} 

\newcommand{\row}[1]{\mathbf{\MakeLowercase{#1}}} 
\newcommand{\rowgrk}[1]{\boldsymbol{#1}} 
\newcommand{\col}[1]{\mathbf{\MakeLowercase{#1}}} 
\newcommand{\mat}[1]{\mathbf{\MakeUppercase{#1}}} 
\newcommand{\matt}[1]{\mathbf{\hat{\MakeUppercase{#1}}}} 
\newcommand{\matz}{\mat{0}} 
\newcommand{\sumVec}[1]{|#1|} 
\newcommand{\card}[1]{\mathrm{Card}(#1)} 

\newcommand{\any}{\ast}   
\newcommand{\trsp}[1]{{#1}^\mathsf{T}} 
\newcommand{\matrow}[2]{{#1}_{#2,\any}} 
\newcommand{\matcol}[2]{{#1}_{\any,#2}} 
\newcommand{\matrows}[2]{{#1}_{#2,\any}} 
\newcommand{\matcols}[2]{{#1}_{\any,#2}} 
\newcommand{\matsub}[3]{{#1}_{#2,#3}} 
\newcommand{\diag}[1]{\mathrm{diag}(#1)}  
\newcommand{\idMat}[1][\rdim]{\mat{I}_{#1}} 
\newcommand{\perm}{\pi} 
\newcommand{\permMat}{\boldsymbol\perm} 
\newcommand{\sddots}{\raisebox{3pt}{$\scalebox{.75}{$\ddots$}$}} 

\newcommand{\rdeg}[2][]{\mathrm{rdeg}_{{#1}}(#2)} 
\newcommand{\cdeg}[2][]{\mathrm{cdeg}_{{#1}}(#2)} 
\newcommand{\leadingMat}[2][\unishift]{\mathrm{lm}_{#1}(#2)} 
\newcommand{\xDiag}[1]{\mat{\var}^{#1\,}} 
\newcommand{\shiftSpace}[1][\rdim]{\ZZ^{#1}} 
\newcommand{\unishift}{\mathbf{0}} 
\newcommand{\shift}[2][s]{#1_{#2}} 
\newcommand{\shifts}[1][s]{\mathbf{#1}} 

\newcommand{\popov}{\mat{P}} 
\newcommand{\reduced}{\mat{R}} 

\newcommand{\vsdim}{\sigma} 
\newcommand{\order}{d} 
\newcommand{\orders}{\tuple{\order}} 
\newcommand{\orderSpace}[1][\nbeq]{\ZZp^{#1}} 
\newcommand{\minDeg}{\delta}  
\newcommand{\minDegs}{\boldsymbol{\delta}}  

\newcommand{\degExp}{ \delta }  
\newcommand{\sdiff}{t}  
\newcommand{\expandMat}{\mat{C}}  
\newcommand{\quoExp}{\alpha}  
\newcommand{\remExp}{\beta}  
\newcommand{\expand}[1]{\overline{#1}}  
\newcommand{\sumDistMax}{\zeta}  

\newcommand{\nbeq}{\cdim} 
\newcommand{\nbun}{\rdim} 
\newcommand{\sys}{\mat{F}} 
\newcommand{\sysSpace}{\polMatSpace[\nbun][\nbeq]} 
\newcommand{\res}{\mat{G}} 
\newcommand{\app}{\row{p}} 
\newcommand{\appSpace}{\polMatSpace[1][\nbun]} 
\newcommand{\appbas}{\mat{P}} 
\newcommand{\appbasSpace}{\polMatSpace[\nbun]} 
\newcommand{\mods}[1][\orders]{\mat{\var}^{#1\,}} 
\newcommand{\ovlpLin}[2]{\mathcal{L}_{#1}(#2)} 
\newcommand{\ovlpLinOrd}[2]{\mathcal{L}_{#1}(#2)} 
\newcommand{\emat}{\mat{E}} 

\newcommand{\sysR}{\mat{\hat{F}}}
\newcommand{\ordersR}{\tuple{\hat\order}}
\newcommand{\minDegsR}{\boldsymbol{\hat\minDeg}}
\newcommand{\shiftsR}{\tuple{\hat{s}}}
\newcommand{\appbasR}{\mat{\hat{P}}}

\newcommand{\sysL}{\mat{\check{F}}}
\newcommand{\ordersL}{\tuple{\check\order}}
\newcommand{\modAppL}{{\mathcal A}_{\ordersL}(\sysL)}
\newcommand{\shiftsL}{\tuple{\check{s}}}
\newcommand{\appL}{\row{\check{p}}}
\newcommand{\appbasL}{\mat{\check{P}}}
\newcommand{\idMatEven}[1]{\mat{J}_{#1}} 
\newcommand{\idMatOdd}[1]{\mat{J}^{\mathrm{c}}_{#1}} 
\newcommand{\selCol}{\mat{S}} 
\newcommand{\selColComp}{\mat{S}^\mathrm{c}} 
\newcommand{\okRows}{I}  
\newcommand{\okRowsNb}{i}  
\newcommand{\okRowsComp}{I^\mathrm{c}}  
\newcommand{\remCols}{J}  
\newcommand{\sumRdeg}{\xi} 
\newcommand{\costZhoLabMin}[1]{\mathcal{C}(#1)}  

\newcommand{\modAppCustom}[2]{{\mathcal A}_{#1}(#2)}
\newcommand{\modApp}{{\mathcal A}_{\orders}(\sys)}
\newcommand{\modAppOvlpLin}[1][]{\modAppCustom{\ovlpLinOrd{#1\degExp}{\orders}}{\ovlpLin{\orders,#1\degExp}{\sys}}}
\newcommand{\mulshift}{\mat{Z}} 

\newcommand{\argfig}[1]{\begin{figure}[#1]} 

\newenvironment{algobox}[1][htbp]{
  \newcommand{\algoInfo}[3]{
    \begin{algorithm}[{name=[\algoname{##2}:~##1]\algoname{##2}}]
    \label{##3}
    ~ \hfill
    {\small\emph{(##1)}}
    \smallskip

  }
  \newcommand{\dataInfos}[2]{
    \algoword{##1:}
      \begin{itemize}[topsep=0pt]
          ##2
      \end{itemize}
    \smallskip
  }
  \newcommand{\dataInfo}[2]{
    \algoword{##1:} ##2 
  }
  \newcommand{\algoSteps}[1]{
    \addtolength{\leftmargini}{-0.35cm}
    \begin{enumerate}[{\bf 1.}]
        ##1
    \end{enumerate}
    \smallskip
  }

  \expandafter\argfig\expandafter{#1}
    \centering
    \begin{minipage}{0.99\textwidth}
    \begin{mdframed}
      \smallskip
    }
    {
    \end{algorithm}
    \end{mdframed}
    \end{minipage}
  \end{figure}
}

\newenvironment{problembox}[1][htbp]{
  \newcommand{\problemInfo}[3]{
    \begin{problem}[{name=[\emph{##2}\ifx&##1&\else##1\fi]\emph{##2}}]
    \label{##3}
    ~\smallskip

  }
  \newcommand{\dataInfos}[2]{
    \emph{##1:}
    \begin{itemize}[topsep=0pt]
      ##2
    \end{itemize}
    \smallskip
  }
  \newcommand{\dataInfo}[2]{
    \emph{##1:}  ##2
  }

  \expandafter\argfig\expandafter{#1}
    \centering
    \begin{minipage}{0.75\textwidth}
    \begin{mdframed}
    }
    {
    \end{problem}
    \end{mdframed}
    \end{minipage}
  \end{figure}
}

\begin{document}

\begin{frontmatter}

\title{Fast computation of approximant bases in canonical form}

\author{Claude-Pierre Jeannerod}
\address{Univ Lyon, Inria, CNRS,  ENS de Lyon, Universit\'e Claude Bernard Lyon 1, LIP UMR 5668, F-69007 Lyon, France}

\author{Vincent Neiger}
\address{Univ.~Limoges, CNRS, XLIM, UMR 7252, F-87000 Limoges, France}

\author{Gilles Villard}
\address{Univ Lyon, CNRS,  ENS de Lyon, Inria, Universit\'e Claude Bernard Lyon 1, LIP UMR 5668, F-69007 Lyon, France}

\begin{abstract}
  In this article, we design fast algorithms for the computation of approximant
  bases in shifted Popov normal form.  We first recall the algorithm known as
  \algoname{PM-Basis}, which will be our second fundamental engine after
  polynomial matrix multiplication: most other fast approximant basis
  algorithms basically aim at efficiently reducing the input instance to
  instances for which \algoname{PM-Basis} is fast.  Such reductions usually
  involve partial linearization techniques due to Storjohann, which have the
  effect of balancing the degrees and dimensions in the manipulated matrices.

  Following these ideas, Zhou and Labahn gave two algorithms which are faster
  than \algoname{PM-Basis} for important cases including Hermite-Pad\'e
  approximation, yet only for shifts whose values are concentrated around the
  minimum or the maximum value. The three mentioned algorithms were designed
  for balanced orders and compute approximant bases that are generally not
  normalized. Here, we show how they can be modified to return the shifted
  Popov basis without impact on their cost bound; besides, we extend Zhou and
  Labahn's algorithms to arbitrary orders.  

  Furthermore, we give an algorithm which handles arbitrary shifts with one
  extra logarithmic factor in the cost bound compared to the above algorithms.
  To the best of our knowledge, this improves upon previously known algorithms
  for arbitrary shifts, including for particular cases such as Hermite-Pad\'e
  approximation. This algorithm is based on a recent divide and conquer
  approach which reduces the general case to the case where information on the
  output degree is available. As outlined above, we solve the latter case via
  partial linearizations and \algoname{PM-Basis}.
\end{abstract}


\end{frontmatter}

\section{Introduction}
\label{sec:intro}

Let $\orders= (\order_1,\ldots,\order_\nbeq) \in \orderSpace$, and let $\sys
\in \sysSpace$ be a matrix of univariate polynomials over a field $\field$,
which represents a matrix of formal power series with the $j$th column
truncated at order $\order_j$. We consider a matrix-type generalization of
Hermite-Pad\'e approximation, which consists in computing polynomial row
vectors $\app \in \appSpace$ such that
\begin{equation}
  \label{eqn:main}
  \app \sys =0 \bmod \mods,
  \quad \text{where}\;\;
  \mods = \diag{\var^{\order_1},\ldots,\var^{\order_\nbeq}}.
\end{equation} 
Here, $\app \sys =0 \bmod \mods$ means that $\app \sys = \col{q} \mods$ for
some $\col{q} \in \polMatSpace[1][\nbeq]$. The set of all such approximants
forms a free $\polRing$-module of rank $\nbun$ denoted by $\modApp$; its bases
are represented as the rows of nonsingular matrices in $\appbasSpace$. One is
usually interested in bases having minimal row degrees with respect to a
\emph{shift} $\shifts \in \shiftSpace$, used as column weights.

In this paper, we improve complexity bounds for the computation of such
\emph{$\shifts$-minimal approximant bases}. In addition, our algorithms return
a canonical $\shifts$-minimal basis of $\modApp$, called the
\emph{$\shifts$-Popov} basis \citep{Popov72,BeLaVi99} and defined in
\cref{subsec:forms}. The properties of this basis allow us to compute it faster
than $\shifts$-minimal bases in general \citep[for more insight,
see][]{JeNeScVi16} and also, once obtained, to efficiently perform operations
with this basis \citep[see for example][Thm.\,12]{RosSto16}.

Our problem is stated in \cref{pbm:pab}; $\cdeg{\sys}$ denotes the tuple of the
$\nbeq$ column degrees of the matrix $\sys$. Here and hereafter, tuples of
integers are always compared componentwise. The assumption that $\cdeg{\sys} <
\orders$ is harmless: truncating the column $j$ of $\sys$ modulo
$\var^{\order_j}$ does not affect the module of approximants.

\begin{problembox}[h]
  \problemInfo
  {}
  {Approximant basis in shifted Popov form}
  {pbm:pab}

  \dataInfos{Input}{
    \item approximation order $\orders \in \orderSpace$,
    \item matrix $\sys$ in $\sysSpace$ with $\cdeg{\sys} < \orders$ componentwise,
    \item shift $\shifts \in \shiftSpace$.
  }

  \dataInfo{Output}{
    the $\shifts$-Popov basis $\appbas \in \appbasSpace$ of the
      $\polRing$-module \vspace{-0.5\baselineskip}
      \[
        \modApp =
        \big\{ \app \in \appSpace \mid \app \sys = 0 \bmod \mods \big\}.
      \]
  }
\end{problembox}

For estimating the tightness of the cost bounds below, we consider the number
of field elements used to represent the input and output of the problem.
Representing polynomials in the standard monomial basis, the matrix $\sys$ is
represented by $\nbun \vsdim$ coefficients from $\field$, where
\[
  \vsdim = \order_1+\cdots+\order_\nbeq=\sumVec{\orders};
\]
here, $\sumVec{\cdot}$ denotes the sum of a tuple of nonnegative integers. By
definition of \(\shifts\)-Popov forms, the output basis can be written $\appbas =
\xDiag{\minDegs} + \mat{A}$ for a matrix $\mat{A}$ such that
$\cdeg{\mat{A}}<\minDegs=\cdeg{\appbas}$. Importantly, we have
$\sumVec{\minDegs} \le \vsdim$ (see \cref{lem:degdet_appbasis}). Thus,
$\appbas$ can be represented by the degrees $\minDegs$ together with
$\nbun\sumVec{\minDegs} \le \nbun\vsdim$ coefficients from $\field$ for the
columns of \(\mat{A}\) (not counting those corresponding to identity columns in
\(\appbas\)). The tuple $\minDegs$, called the \emph{$\shifts$-minimal degree
of $\modApp$}, plays a central role in our algorithms; knowing $\minDegs$
amounts to knowing the degrees of the columns of the sought basis.

Our cost model estimates the number of arithmetic operations in $\field$ on an
algebraic RAM. We consider an exponent $\expmatmul$ for matrix multiplication:
two matrices in $\matSpace[\nbun]$ can be multiplied in
$\bigO{\nbun^\expmatmul}$ operations in $\field$. \emph{In this paper, all cost
bounds are given for $\expmatmul>2$}; additional logarithmic factors may appear if
$\expmatmul=2$. \citep{CopWin90,LeGall14} show that one can take $\expmatmul <
2.373$. We also use a cost function $\polmatmultime{\cdot,\cdot}$ for the
multiplication of polynomial matrices, defined as follows: for two real numbers
$\nbun,d>0$, $\polmatmultime{\nbun,d}$ is such that two matrices of degree at
most $d$ in $\polMatSpace[\bar\nbun]$ with $\bar\nbun \le \nbun$ can be
multiplied using $\polmatmultime{\nbun,d}$ operations in $\field$.
Furthermore, we will use $\appbastime{\nbun,d} = \sum_{0\le i\le \log({d})} 2^i
\polmatmultime{{\nbun},2^{-i}{d}}$ from \citep{Storjohann03,GiJeVi03}, which is
typically related to divide-and-conquer computations. 


We will always give cost bounds in function of $\polmatmultime{\nbun,d}$ and
$\appbastime{\nbun,d}$; the current best known upper bounds on the former
quantity can be found in \citep{CanKal91,BosSch05,HaHoLe17}. The first of these
references proves
\[
  \polmatmultime{\nbun,d} \in \bigO{\nbun^\expmatmul d \log(d) + \nbun^2 d \log(d) \log(\log(d)) + \nbun^\expmatmul}
\]
for an arbitrary field $\field$, while the last two show better bounds in the
case of fields that are either finite or of characteristic zero. For the sake
of presentation, we will also give simplified cost bounds for our main results,
relying on the following assumption:
\begin{align*}
  \label{eqn:polmatmultime_superlinear}
  \hypslmm: & \;\; \polmatmultime{\nbun,d} + \polmatmultime{\nbun,d'} \le
  \polmatmultime{\nbun,d+d'} \text{ for } \nbun,d,d'>0
  & \textit{(super-linearity)}.
\end{align*}
We remark that $\hypslmm$ implies $\appbastime{\nbun,d} \in
\bigO{\polmatmultime{\nbun,d} \log(d)}$.

It is customary to assume $\polmatmultime{\nbun,d} \in \bigO{\nbun^\expmatmul
\polmultime{d}}$ for a cost function $\polmultime{\cdot}$ such that two
polynomials in $\polRing$ of degree at most $d$ can be multiplied in
$\polmultime{d}$ operations in $\field$.  However this does not always reflect
well the actual cost of polynomial matrix multiplication, which tends to have a
term in $\nbun^2 d$ with several (sub)logarithmic factors, and a term in
$\nbun^\expmatmul d$ with at most one logarithmic factor.  In fact, even the
above general bound on $\polmatmultime{\nbun,d}$ is asymptotically better than
$\bigO{\nbun^\expmatmul \polmultime{d}}$ if we replace $\polmultime{d}$ by the
best known bound.

As a consequence, and since we will be discussing cost bound improvements on
the level of logarithmic factors, we will not follow this custom.  Instead, and
as in \citep{Storjohann03} for example, we will prefer to write our cost bounds
with general expressions involving $\polmatmultime{\nbun,d}$ and
$\appbastime{\nbun,d}$, which one can then always replace with
context-dependent upper bounds.

\paragraph{Main result}

We give an efficient solution to \cref{pbm:pab} for arbitrary orders and
shifts.

\begin{theorem}
  \label{thm:fast_pab}
  Let $\orders \in \orderSpace$, let $\sys \in \sysSpace$ with
  $\cdeg{\sys}<\orders$, and let $\shifts \in \shiftSpace$. Then, writing
  $\vsdim=\sumVec{\orders}$ for the sum of the entries of $\orders$ and
  assuming $\nbun \in \bigO{\vsdim}$, \cref{pbm:pab} can be solved in 
  \[
    \bigOPar{\left(\sum_{k=0}^{\lceil\log_2(\vsdim/\nbun)\rceil} 2^k
    \appbastime{\nbun,2^{-k}\vsdim/\nbun}\right) + \nbun^{\expmatmul-1} \vsdim
  \log(\nbun)}
  \]
  operations in $\field$. Assuming $\hypslmm$, this is in
  $\bigO{\polmatmultime{\nbun,\vsdim/\nbun} \log(\vsdim/\nbun)^2 +
  \nbun^{\expmatmul-1}\vsdim \log(\nbun)}$.
\end{theorem}

\noindent Hiding logarithmic factors, this cost bound
is $\softO{\nbun^{\expmatmul-1}\vsdim}$, the same as for the multiplication of
two $\nbun\times \nbun$ matrices of degree $\vsdim/\nbun$. As mentioned above,
the output basis has average column degree at most $\vsdim/\nbun$, which is
reached generically. Furthermore, there are instances of \cref{pbm:pab} whose
resolution does require at least as many field operations as the multiplication of
two matrices in $\appbasSpace$ of degree about $\vsdim/\nbun$ (see
\cref{subsec:reductions}).

In the case $\vsdim \in \bigO{\nbun}$, less common in applications, the current
fastest known algorithm for solving \cref{pbm:pab} uses $\bigO{\nbun
\vsdim^{\expmatmul-1} + \vsdim^\expmatmul\log(\max(\orders))}$ operations
\citep[Prop.\,7.1]{JeNeScVi17}.

The overall design of our main algorithm is based on
\citep[Algo.\,1]{JeNeScVi16}; we refer to \citepalias[Sec.\,1.2]{JeNeScVi16}
for an overview of this approach. In short, we use a divide and conquer
strategy which splits the order $\orders$ into two parts whose sums are about
$\vsdim/2$. Two corresponding shifted Popov bases are found recursively and
yield the $\shifts$-minimal degree $\minDegs$, which then helps us to
efficiently compute the $\shifts$-Popov approximant basis.

In fact, \citepalias[Algo.\,1]{JeNeScVi16} solves a more general problem; we refer to
\citep{BarBul92,Beckermann92,BecLab97} for details about and earlier solutions
to \emph{matrix rational interpolation} problems. \cref{eqn:main} is indeed a
particular case of 
\begin{equation} \label{eqn:maingen}
  \app \sys =0 \bmod (\mat{\var}-\mathbf{x})^{\orders}, 
  \quad \text{where}\;\; 
  (\mat{\var}-\mathbf{x})^{\orders} = \diag{[(\var-x_{j})^{d_j}]_{1\le j\le\nbeq}},
\end{equation}
where these diagonal entries are given by their roots $\mathbf{x}$ and
multiplicities $\orders$.

For such equations, \citep[Algo.\,FFFG]{BecLab00} returns the $\shifts$-Popov
basis of solutions in $\bigO{\nbun \vsdim^2}$ operations
\cite[Sec.\,6.4]{Neiger16}. At each step of this iterative algorithm, one
normalizes the computed basis to better control its degrees, and thus achieve
better efficiency. Indeed, similar algorithms without normalization, such as
the one in \citep{BarBul92}, have a cost of $\bigO{\nbun^2\vsdim^2}$ operations
in general.

The algorithm of~\citep{JeNeScVi16} also addresses \cref{eqn:maingen}. Here, we
obtain a faster algorithm in the case $\mathbf{x}=\tuple{0}$ by improving one
of its core components: solving \cref{pbm:pab} when the $\shifts$-minimal
degree $\minDegs$ is known a priori. Explicitly, the gain here compared to the
cost bound in \citepalias[Thm.\,1.3]{JeNeScVi16} is in $\Omega(\log(\vsdim))$.

This extra logarithmic factor in \citepalias{JeNeScVi16} has two independent
sources. First, it originates from the computation of \emph{residuals}, which
are matrix remainders of the form $\appbas \sys \bmod
(\mat{\var}-\mathbf{x})^{\orders}$; here, with $\mathbf{x}=\tuple{0}$, these
are simply truncated products. Second, it also comes from the strategy for
handling unbalanced output degrees, by relying on \citep[Algo.\,2]{JeNeScVi17}
which uses unbalanced polynomial matrix products and changes of shifts. Here we
rather make use of the overlapping linearization from
\citep[Sec.\,2]{Storjohann06}, allowing us to reduce more directly to cases
solved by \citep[Algo.\,\algoname{PM-Basis}]{GiJeVi03} using balanced
polynomial matrix products.

\paragraph{Balanced orders: obtaining the canonical basis via
\algoname{PM-Basis}}

Let us now consider the case where the $\nbeq$ entries of the order $\orders$
are roughly the same. More precisely, we assume that
\begin{align*}
  \hypobal: & \;\; \max(\orders) \in \bigO{\vsdim/\nbeq} & \textit{(balanced order)},
\end{align*}
and we let $\order = \max(\orders)$. We note that any algorithm designed for a
uniform order $(\order,\ldots,\order)$ can straightforwardly be used to deal
with any order $\orders$ (see \cref{rmk:pmbasis_unbalancedorders}); yet, this
might lead to poor performance if the latter order is not balanced.

Under $\hypobal$, the divide and conquer algorithm of~\citep{BecLab94},
improved as in \citep[Algo.\,\algoname{PM-Basis}]{GiJeVi03}, computes an
$\shifts$-minimal approximant basis using $\bigO{(1 +
\nbeq/\nbun)\appbastime{\nbun,\vsdim/\nbeq}}$ operations.  This is achieved for
arbitrary shifts, despite the existence of $\shifts$-minimal bases with
arbitrarily large degree: \algoname{PM-Basis} always returns a basis of degree
$\le\order$. It is particularly efficient in the case $\nbeq=\Theta(\nbun)$,
the cost bound being then in $\softO{\nbun^{\expmatmul-1} \vsdim}$.

Here, we slightly modify \algoname{PM-Basis} so that its output basis reveals
the $\shifts$-minimal degree $\minDegs$.  For this, we ensure that, in addition
to being $\shifts$-minimal, this basis exhibits a so-called \emph{pivot entry}
on each row; it is then said to be in $\shifts$-weak Popov form
\citep{MulSto03}.  Computing bases in this form to obtain $\minDegs$ will be a
common thread in all the algorithms we present.

Then, we show that the canonical basis can be obtained by using essentially two
successive calls to \algoname{PM-Basis}: the first one to find $\minDegs$, and
the second one to find the basis by using $-\minDegs$ in place of the shift.
The correctness of this approach is detailed in \cref{lem:mindeg_shift}.

\begin{theorem}
  \label{thm:pm_basis}
  Let $\orders \in \orderSpace$, let $\sys \in \sysSpace$ with
  $\cdeg{\sys}<\orders$, and let $\shifts \in \shiftSpace$.  Then,
  \begin{itemize}
    \item \cref{pbm:pab} can be solved in $\bigO{(1 + \nbeq/\nbun)
      \appbastime{\nbun,\order}}$ operations in $\field$, where
      $\order=\max(\orders)$; assuming $\hypslmm$, this is in $\bigO{(1 +
      \nbeq/\nbun) \polmatmultime{\nbun,\order} \log(\order)}$.
    \item If \(\nbeq>\nbun\) (hence also \(\vsdim>\nbun\)),
      \cref{pbm:pab} can be solved in
      $\bigO{\appbastime{\nbun,\vsdim/\nbun} +
        \appbastime{\nbun,\order}}$ operations in $\field$; assuming
        $\hypslmm$, this is in
        $\bigO{\polmatmultime{\nbun,\vsdim/\nbun}\log(\vsdim/\nbun)
        + \polmatmultime{\nbun,\order} \log(\order)}$.
  \end{itemize}
\end{theorem}

When $\nbeq>\nbun$, the cost bound in the second item improves upon that in the
first item for some unbalanced orders.  Take for example
$\orders=(\vsdim/2,1,\ldots,1)$ with $\nbeq=\vsdim/2+1 \ge \nbun$: then,
$\order=\vsdim/2$ and the first bound is $\bigO{\frac{\vsdim}{\nbun}
\appbastime{\nbun,\vsdim}}$ whereas the second bound is only
$\bigO{\appbastime{\nbun,\vsdim}}$.  This is obtained via an algorithm which
reduces the column dimension to $\nbeq<\nbun$ (first term in the cost) and then
applies \algoname{PM-Basis} on the remaining instance (second term in the
cost).  The first step is itself done by applying \algoname{PM-Basis} a
logarithmic number of times to process all columns whose corresponding order is
less than $\vsdim/\nbun$; there are at least $\nbeq-\nbun$ such columns by
definition of $\vsdim$.

To illustrate the involved logarithmic factors, let us consider
$\nbun=\nbeq+1=2$.  The cost bounds in the last theorem become
$\bigO{\polmultime{\vsdim} \log(\vsdim)}$, the same as for the related half-gcd
algorithm in $\polRing$ of \citet{Knu70,Sch71,Moe73}. Besides, the bound
$\bigO{\polmultime{\vsdim} \log(\vsdim)^3}$ from \citep{JeNeScVi16} is replaced
by $\bigO{\polmultime{\vsdim}\log(\vsdim)^2}$ in \cref{thm:fast_pab}. We will
see that this remaining extra logarithmic factor compared to the half-gcd comes
from two layers of recursion: at each node of the global divide and conquer
scheme, there is a call to \algoname{PM-Basis}, which itself is a divide and
conquer algorithm performing a polynomial matrix product at each node. To avoid
this factor for the general approximation problem considered here is an open
question.

\paragraph{Weakly unbalanced shifts, around their minimal or maximum value}

In this paragraph, we report cost bounds from \citep{ZhoLab12} which are
proved under the following assumptions:
{ \setlength{\jot}{1pt} 
\begin{align*}
  \hyppolmul: & \;\; \polmatmultime{\nbun,d} \in \Theta(\nbun^\expmatmul
                    \polmultime{d}), \polmultime{k d} \in \bigO{k^{\expmatmul-1}\polmultime{d}}, \\
              & \;\; \text{and } \polmultime{d} + \polmultime{d'} \le \polmultime{d+d'} \text{ for } \nbun,d,d'>0 \text{ and } k\ge 1.
\end{align*} }%
Note that $\hyppolmul$ implies $\hypslmm$. Hereafter, for an integer $t$ and a
shift $\shifts = (\shift{1},\ldots,\shift{\nbun})$, we denote by $\shifts+t$
the shift $(\shift{1}+t,\ldots,\shift{\nbun}+t)$, and notation such as the inequality $\shifts
\le t$ stands for $\max(\shifts)\le t$.

The algorithm \algoname{PM-Basis} discussed above is efficient for
$\nbeq\in\Omega(\nbun)$ and assuming $\hypobal$.  Yet, when $\nbeq$ is small
compared to $\nbun$, this assumption \(\hypobal\) becomes weaker and so does
the bound $\order=\max(\orders)$ controlling the output degree.  In the extreme
case $\nbeq=1$, $\hypobal$ is void since $\order \le \vsdim = \sumVec{\orders}$
always holds; then, \algoname{PM-Basis} manipulates bases of degree up to
$\order=\vsdim$, and its cost bound is $\softO{\nbun^\expmatmul \vsdim}$.
Focusing on the case $\nbeq<\nbun$, \citet{ZhoLab12} noted that both the
assumption
\begin{align*}
  \hypsbal: & \;\; \max(\shifts) -\min(\shifts) \in \bigO{\vsdim/\nbun} & \qquad \qquad \textit{(balanced shift)}
\end{align*}
and the weaker assumption
\begin{align*}
  \hypsmin: & \;\; \sumVec{\shifts-\min(\shifts)} \in \bigO{\vsdim}  & \textit{(weakly unbalanced shift, around $\min$)}
\end{align*}
imply that the average row degree of any $\shifts$-minimal approximant basis is
in $\bigO{\vsdim/\nbun}$. Then, using the \emph{overlapping linearization}
technique from \citep[Sec.\,2]{Storjohann06} at most $\log(\nbun/\nbeq)$ times,
they reduced to the case $\nbeq=\Theta(\nbun)$ and obtained the cost bound
$\bigO{\nbun^\expmatmul \polmultime{\vsdim/\nbun} \log(\vsdim/\nbeq)} \subseteq
\softO{\nbun^{\expmatmul-1}\vsdim}$ \citep[Sec.\,3~to~5]{ZhoLab12}, under
$\hyppolmul$, $\hypobal$, and $\hypsmin$.  The partial linearizations are done
at a degree $\degExp$ which is doubled at each iteration, each of them allowing
to recover the rows of degree $\le \degExp$ of the sought basis.  There are
many such rows since the average row degree is small by assumption: after the
$k$th iteration, only $\bigO{\nbun/2^k}$ rows remain to be found.  An essential
property for efficiency is that the found rows can be discarded in the further
iterations; this yields a dimension decrease which compensates for the
increase of the degree $\degExp$.

On the other hand, assuming
\begin{align*}
  \hypsmax: & \;\; \sumVec{\!\max(\shifts)-\shifts} \in \bigO{\vsdim} & \textit{(weakly unbalanced shift, around $\max$)}
\end{align*}
implies
roughly that the sought basis has average row degree in $\bigO{\vsdim/\nbun}$
up to a small number of columns whose degree is large, and that the shift can
be used to guess locations for these columns.  Then,
\citet[Sec.\,6]{ZhoLab12} use $\log(\nbun)$ calls to the \emph{output column
linearization} from \citep[Sec.\,3]{Storjohann06} in degree $\degExp$.  At each
call, this transformation reduces to the case $\hypsbal$ and allows one to
uncover rows of the sought basis whose degree is at a distance at most
$\degExp$ from the expected one.  Again, there must be many such rows under
$\hypsmax$, and since the remaining rows have degrees which do not agree well
with the shift, they must contain large blocks of zeroes; this leads to
decreasing the dimensions while $\degExp$ is doubled.  This approach has the
same asymptotic cost as above, still under $\hyppolmul$ and $\hypobal$; we
summarize this in \cref{fig:linearizations} (top).

Most often, the approximant bases returned by the algorithms in
\citep{ZhoLab12} are not normalized.  Here, we show how to modify these
algorithms to obtain the $\shifts$-Popov basis without impacting the cost
bound.  Furthermore, we generalize them to arbitrary orders; in other words, we
remove the assumptions $\nbeq < \nbun$ and $\hypobal$.  Instead of making
assumptions on $\shifts$ such as $\hypsmin$ and $\hypsmax$, we extend the
algorithms to arbitrary shifts and give cost bounds parametrized by the
quantities $\sumVec{\shifts-\min(\shifts)}$ and
$\sumVec{\!\max(\shifts)-\shifts}$ which appear in the latter assumptions and
are inherent to the approach.  Then, the obtained cost bounds range from
$\softO{\nbun^{\expmatmul-1} \vsdim}$ under $\hypsmin$ or $\hypsmax$, thus
matching \cref{thm:fast_pab} up to logarithmic factors, to
$\softO{\nbun^{\expmatmul} \order}$ when the quantities above exceed some
threshold, thus matching \cref{thm:pm_basis}; in the latter case, the
algorithms essentially boil down to a single call to \algoname{PM-Basis}.
Precisely, we obtain the next result.

\begin{theorem}
  \label{thm:zhou_labahn}
  Let $\orders \in \orderSpace$, let $\sys \in \sysSpace$ with
  $\cdeg{\sys}<\orders$, and let $\shifts \in \shiftSpace$.  Consider the
  parameters $\vsdim = \sumVec{\orders}$, $\order=\max(\orders)$, $\sumRdeg =
  \vsdim + \sumVec{\shifts-\min(\shifts)}$, and $\sumDistMax = \vsdim +
  \sumVec{\!\max(\shifts)-\shifts}$.  Then,
  \begin{itemize}
    \item If $\sumRdeg \le \nbun\order$, \cref{pbm:pab} can be solved in
      $\bigO{\costZhoLabMin{\sumRdeg,\nbun,\order}}$ operations in $\field$,
      where 
      \begin{equation}
        \label{eqn:cost_zholab_min}
        \costZhoLabMin{\sumRdeg,\nbun,\order} = \sum_{k=0}^{\lceil
        \log_2(\order/\lceil\sumRdeg/\nbun\rceil)\rceil}
          \appbastime{2^{-k}\nbun,2^k\lceil\sumRdeg/\nbun\rceil} +
          2^k \polmatmultime{2^{-k}\nbun,2^k\lceil\sumRdeg/\nbun\rceil}.
      \end{equation}
      Assuming $\hyppolmul$, the latter quantity is in $\bigO{\nbun^\expmatmul
      \polmultime{\lceil\sumRdeg/\nbun\rceil} \log(\order)}$.
    \item If $\sumDistMax \le \nbun \order$, \cref{pbm:pab} can be solved in
      \[
        \bigOPar{
          \appbastime{\mu,\lceil\vsdim/\mu\rceil}
          + \appbastime{\mu,\order}
          + \sum_{k=0}^{\lfloor\log_2(\nbun\order/\sumDistMax)\rfloor}
          \costZhoLabMin{\sumDistMax,2^{-k}\nbun,\order} 
          }
      \]
      operations in $\field$, for some integer $\mu \in \ZZp$ such that
      $\mu\le\nbun$ and $\mu\order<\sumDistMax$.  Assuming $\hyppolmul$, this
      cost bound is in $\bigO{\nbun^{\expmatmul}
        \polmultime{\lceil\sumDistMax/\nbun\rceil} \log(\order) +
      \mu^\expmatmul \polmultime{\lceil\vsdim/\mu\rceil}
    \log(\lceil\vsdim/\mu\rceil)}$.
  \end{itemize}
\end{theorem}

As above, consider these cost bounds for $\vsdim\ge\nbun$. They can be
written $\softO{\nbun^{\expmatmul-1} \sumRdeg}$ and
$\softO{\nbun^{\expmatmul-1} \sumDistMax}$ and they improve upon those in
\cref{thm:pm_basis} when $\sumRdeg \in o(\nbun\order)$ and when $\sumDistMax
\in o(\nbun\order)$, respectively.  Note that $\hypsmin$ and $\hypsmax$ are
equivalent to $\sumRdeg \in \bigO{\vsdim}$ and $\sumDistMax \in \bigO{\vsdim}$,
respectively; under either of these two assumptions, the corresponding cost
bound in the above theorem improves upon that in \cref{thm:fast_pab} at the
level of logarithmic factors, assuming $\hyppolmul$.

An important example of a shift which satisfies neither $\sumRdeg \le
\nbun\order$ nor $\sumDistMax \le \nbun \order$ is the one which yields the
approximant basis in Hermite form; namely, $\shifts =
(\vsdim,2\vsdim,\ldots,\nbun\vsdim)$ for which we have $\sumRdeg = \sumDistMax
= \frac{\nbun(\nbun-1)}{2} \vsdim \ge \frac{\nbun-1}{2} \nbun \order$.  Then,
only the cost in \cref{thm:fast_pab} meets the target
$\softO{\nbun^{\expmatmul-1}\vsdim}$ in general: \cref{thm:zhou_labahn} is void
with such $\sumRdeg$ and $\sumDistMax$, while the cost
$\softO{\nbun^{\expmatmul-1}\vsdim+\nbun^\expmatmul \order}$ in
\cref{thm:pm_basis} has an extra factor $\nbun \order / \vsdim$ which can be as
large as $\nbun$.

The cost bounds in \cref{thm:zhou_labahn} refine those in \citep[Thm.\,5.3
and\,6.14]{ZhoLab12}.  \citet{JeNeScVi17} gave an algorithm achieving a cost
similar to that in the first item above, in the more general context of
\cref{eqn:maingen} and thus covering the case of arbitrary orders as well; the
cost bound above improves upon that given in \citepalias[Thm.\,1.5]{JeNeScVi17}
by a logarithmic factor.

\begin{figure}[tb]
  \centering
  \fbox{
    \begin{tikzpicture}[node distance=0.5cm, auto,]
      \node[draw,rectangle,rounded corners] (smax) {
        [$\nbeq<\nbun$, $\hypobal$, and] $\hypsmax$
      };
      \node[draw,rectangle,rounded corners,below=of smax] (sbal) {
        [$\nbeq<\nbun$, $\hypobal$, and] $\hypsbal$
      };
      \node[draw,rectangle,rounded corners,below=of sbal] (sbalcdimred) {
        [$\hypobal$, and] $\nbeq<\nbun$, $\hypsbal$
      };
      \node[left=1.2cm of sbalcdimred] (empty) {\strut };
      \node[draw,rectangle,rounded corners,left=1.2cm of empty] (smincdimred) {
        [$\hypobal$ and] $\nbeq<\nbun$, $\hypsmin$
      };
      \node[draw,rectangle,rounded corners,above=of smincdimred] (smin) {
        [$\nbeq<\nbun$, $\hypobal$ and] $\hypsmin$
      };
      \node[draw,rectangle,rounded corners,below=of empty] (balanced) {
         $\nbeq\in\Theta(\nbun)$ and $\hypobal$
      };
      \node[right=0.8cm of balanced] (pmbasis) {
        \small \emph{fast solution using} \algoname{PM-Basis}
      };
      \draw[->] (smax) -- node[left] {\small\emph{output column linearization}} (sbal);
      \draw[->] (sbal) -- node[left] {\small\emph{\cref{algo:reduce_coldim}, based on \algoname{PM-Basis}}\phantom{spaces}} (sbalcdimred);
      \draw[->,shorten >=2pt,bend left=5] (sbalcdimred) to node[left,yshift=0.05cm] {\small\emph{overlapping linearization}} (balanced);
      \draw[->,shorten >=2pt,bend right=5] (smincdimred) to (balanced);
      \draw[->] (smin) -- (smincdimred);
      \draw[->] (balanced) -- (pmbasis);
    \end{tikzpicture}
  }

  \smallskip

  \fbox{
    \begin{tikzpicture}[node distance=0.5cm, auto,]
      \node[draw,rectangle,rounded corners] (input) {
        known $\minDegs = \cdeg{\appbas}$ with $\sumVec{\minDegs}\le\vsdim$
      };
      \node[draw,rectangle,rounded corners,below=of input] (smalldeg) {
        known $\minDegs = \cdeg{\appbas}$ with $\max(\minDegs)\in \bigO{\vsdim/\nbun}$
      };
      \node[draw,rectangle,rounded corners,below=of smalldeg] (smalldegtall) {
        $\nbeq<\nbun$, known $\minDegs = \cdeg{\appbas}$ with $\max(\minDegs)\in \bigO{\vsdim/\nbun}$
      };
      \node[draw,rectangle,rounded corners,below=of smalldegtall] (balanced) {
        $\nbeq\in\Theta(\nbun)$ and $\hypobal$
      };
      \node[right=0.5cm of balanced] (pmbasis) {
        \small \emph{fast solution using} \algoname{PM-Basis}
      };
      \draw[->] (input) -- node[right] {\small\emph{output column linearization}} (smalldeg);
      \draw[->] (smalldeg) -- node[right] {\small\emph{\cref{algo:reduce_coldim}, based on \algoname{PM-Basis}}} (smalldegtall);
      \draw[->] (smalldegtall) -- node[right] {\small\emph{overlapping linearization}} (balanced);
      \draw[->] (balanced) -- (pmbasis);
    \end{tikzpicture}
  }
  \caption{\emph{(Top)} Fast algorithm from \citep{ZhoLab12} assuming either
    $\hypsmin$ or $\hypsmax$, via a logarithmic number of partial
    linearizations from \citep{Storjohann06} and calls to \algoname{PM-Basis}.
    In brackets, assumptions that we have removed in our modified algorithm; we
    have also inserted the column dimension reduction
    (\cref{algo:reduce_coldim}) which is not necessary in \citep{ZhoLab12}
    where $\nbeq<\nbun$ is assumed.  \emph{(Bottom)} Fast algorithm when the
    shifted minimal degree is known, using two partial linearizations from
    \citep{Storjohann06} and calls to
  \citep[Algo.\,\algoname{PM-Basis}]{GiJeVi03}.}
  \label{fig:linearizations}
\end{figure}

\paragraph{Known minimal degree}

The main new ingredient behind \cref{thm:fast_pab} is an efficient algorithm
for \cref{pbm:pab} when the $\shifts$-minimal degree $\minDegs$ of $\modApp$ is
known.

As noted above, knowing $\minDegs$ leads us to consider the shift $-\minDegs$
instead of $\shifts$. This new shift is weakly unbalanced around its maximum
value, since $\sumVec{\minDegs} \le \vsdim$.  Inspired by the efficient
algorithms of \citep{ZhoLab12} for such shifts, we consider the same overall
strategy while exploiting the additional information given by $\minDegs$ to
design a simpler and more efficient algorithm.

To handle the unbalancedness of the output column degrees,
\citepalias{ZhoLab12} uses a logarithmic number of output column
linearizations, each of them leading to find some rows of the sought basis.
Thanks to the knowledge of $\minDegs$, we are able to use the same
linearization only once, with parameters which directly yield the full basis
(\cref{algo:knowndeg_pab}, Step~\textbf{1}).  This transformation builds a new
instance for which the new shifted minimal degree $\minDegs$ is known and
balanced: $\max(\minDegs) \in \bigO{\vsdim/\nbun}$.

Then, we use \algoname{PM-Basis} to efficiently reduce to the case $\nbeq <
\nbun$ (\cref{algo:knowndeg_pab}, Step~\textbf{2}).  This is not done in
\citepalias{ZhoLab12} since $\nbeq < \nbun$ holds by assumption in this
reference (yet, we do resort to column dimension reduction in our generalized
version of this algorithm, see \cref{algo:zhou_labahn_min}, Step~\textbf{1}).

Now, to handle balanced shifts such as the new $-\minDegs$,
\citepalias{ZhoLab12} uses a logarithmic number of overlapping linearizations.
Each of these transformations gives an instance satisfying
$\nbeq\in\Theta(\nbun)$ and $\hypobal$, which can thus be solved efficiently
via \algoname{PM-Basis}, thereby uncovering some rows of the output basis.
Here, since the output degree is $\max(\minDegs)\in\bigO{\vsdim/\nbun}$, a
single call to overlapping linearization (\cref{algo:knowndeg_pab},
Step~\textbf{3}) yields a new instance which directly gives the full basis; as
above, it satisfies $\nbeq\in\Theta(\nbun)$ and $\hypobal$ and thus can be
solved efficiently via \algoname{PM-Basis}.

We summarize our approach in \cref{fig:linearizations} (bottom diagram). We
note that similar ideas were already used in \citep[Sec.~3]{GupSto11}, in the
context of Hermite form computation when the degrees of the diagonal entries
are known.

To summarize, we obtain the cost bound $\bigO{\appbastime{\nbun,\vsdim/\nbun}}$
for solving \cref{pbm:pab} when $\minDegs$ is known (see
\cref{prop:algo:knowndeg_pab}), without any further assumption.  This improves
over the algorithm in \citep[Sec.\,4]{JeNeScVi16}, designed for the same
purpose but in the more general context of \cref{eqn:maingen}, in which it is
unclear to us how to generalize the overlapping linearization.

\paragraph{Outline of the paper}

In \cref{sec:preliminaries}, we present preliminary definitions and properties.
Then, in \cref{sec:balanced_order}, we describe the algorithm
\algoname{PM-Basis} and prove the first item of \cref{thm:pm_basis}. We use
this algorithm in \cref{sec:reduce_nbeq} to show how to reduce to $\nbeq<\nbun$
efficiently; this implies the second item of \cref{thm:pm_basis}.  Together
with partial linearizations that we recall, this allows us to solve
\cref{pbm:pab} when the $\shifts$-minimal degree is known
(\cref{sec:knowndeg}).  Then, in \cref{sec:fastpab}, we give our main algorithm
and the proof of \cref{thm:fast_pab}.  Finally, we present generalizations of
the algorithms of \citep{ZhoLab12} and we prove \cref{thm:zhou_labahn} in
\cref{sec:weakly_unbalanced_shift}.

\section{Preliminaries}
\label{sec:preliminaries}

\subsection{Minimal bases, Popov bases, and minimal degree}
\label{subsec:forms}

For a shift $\shifts = (\shift{j})_j \in \shiftSpace[\rdim]$, the
\emph{$\shifts$-degree} of $\row{p} = [p_j]_j \in \polMatSpace[1][\rdim]$ is
$\max_{j} (\deg(p_j) + \shift{j})$, with the convention $\deg(0) = -\infty$.
If \(\row{p}\) is nonzero, its \emph{\(\shifts\)-pivot} is its rightmost entry
\(p_i\) such that \(\deg(p_i) + \shift{i} = \rdeg{\row{p}}\); then, \(i\) and
\(\deg(p_i)\) are called the \(\shifts\)-pivot \emph{index} and the
\(\shifts\)-pivot \emph{degree} of \(\row{p}\), respectively. The
\emph{$\shifts$-row degree} of a matrix $\mat{P} \in \polMatSpace[k][\rdim]$ is
$\rdeg[\shifts]{\mat{P}} = (r_1,\ldots,r_k)$ where $r_i$ is the
$\shifts$-degree of the $i$th row of $\mat{P}$, and the \emph{$\shifts$-leading
matrix} of $\mat{P} = [p_{ij}]_{ij}$ is the matrix
$\leadingMat[\shifts]{\mat{P}} \in \matSpace[k][\rdim]$ whose entry $(i,j)$ is
the coefficient of degree $r_i - \shift{j}$ of $p_{ij}$. Furthermore, if
\(\mat{P}\) has no zero row, its \(\shifts\)-pivot index (resp.~degree) is the
tuple of the \(\shifts\)-pivot indices (resp.~degrees) of its rows. The column
degree of $\appbas$ is $\cdeg{\appbas} = \rdeg[\unishift]{\trsp{\appbas}}$,
where $\trsp{\appbas}$ is the transpose of $\appbas$. We use the following
definitions from \citep{Kailath80,BeLaVi99,MulSto03}.

\begin{definition}
  \label{dfn:spopov}
  For $\shifts \in \shiftSpace$, a nonsingular matrix $\appbas \in
  \appbasSpace$ is said to be in
  \begin{itemize}
    \item \emph{$\shifts$-reduced form} if $\leadingMat[\shifts]{\appbas}$ is
      invertible;
    \item \emph{$\shifts$-ordered weak Popov form} if
      $\leadingMat[\shifts]{\appbas}$ is invertible and lower triangular;
    \item \emph{$\shifts$-weak Popov form} if it is in $\shifts$-ordered weak
      Popov form up to row permutation;
    \item \emph{$\shifts$-Popov form} if $\leadingMat[\shifts]{\appbas}$
      is unit lower triangular and $\leadingMat[\unishift]{\trsp{\appbas}}$
      is the identity matrix.
  \end{itemize}
\end{definition}

\noindent In particular, the $\shifts$-pivot degree of a matrix $\appbas$ in
$\shifts$-ordered weak Popov form is the tuple $\minDegs \in \NN^\nbun$ of the
degrees of its diagonal entries, and for $\appbas$ in $\shifts$-Popov form we
have $\minDegs = \cdeg{\appbas}$.

For $\orders\in\orderSpace$ and $\sys\in\sysSpace$, a basis of $\modApp$ in
$\shifts$-reduced form is said to be an \emph{$\shifts$-minimal} basis of
$\modApp$. We further call \emph{$\shifts$-minimal degree of $\modApp$} the
$\shifts$-pivot degree of the $\shifts$-Popov basis of $\modApp$, and in fact
of any $\shifts$-ordered weak Popov basis of $\modApp$
\citep[Lem.\,3.3]{JeNeScVi16}. The importance of these degrees is highlighted
by the next two lemmas.

The first one allows us to control the degrees in the computed bases and can be
found in \citep[Thm.\,4.1]{BarBul92} in a more general context. The second one
follows from \citep[Lem.\,15~and~17]{SarSto11} and shows that when the
$\shifts$-minimal degree $\minDegs$ is known, the computations may be performed
with the shift $-\minDegs$.

\begin{lemma}
  \label{lem:degdet_appbasis}
  Let $\orders \in \orderSpace$, let $\vsdim=\sumVec{\orders}$, and let $\sys
  \in \sysSpace$ with $\cdeg{\sys} < \orders$. Then, for any basis $\appbas \in
  \appbasSpace$ of $\modApp$, we have $\deg(\det(\appbas)) \le \vsdim$.
  Furthermore, for $\shifts \in \shiftSpace$, the $\shifts$-minimal degree
  $\minDegs \in \NN^\nbun$ of $\modApp$ satisfies $\sumVec{\minDegs} \le
  \vsdim$ and $\max(\minDegs) \le \max(\orders)$.
\end{lemma}
\begin{proof}
  Let $\appbas$ be the $\shifts$-Popov basis of $\modApp$. Then, $\appbas$ is
  in particular $\unishift$-column reduced, hence $\deg(\det(\appbas)) =
  \sumVec{\cdeg{\appbas}} = \sumVec{\minDegs}$ \citep[Sec.\,6.3.2]{Kailath80};
  and since any basis of $\modApp$ has determinant $\lambda\det(\appbas)$ for
  some nonzero $\lambda\in\field$, it is enough to prove that
  $\sumVec{\minDegs} \le \vsdim$.

  Since $\appbas$ has column degree $(\minDeg_1,\ldots,\minDeg_\nbun)$,
  according to \citep[Thm.\,6.3.15]{Kailath80} the quotient $\appSpace/\modApp$
  is isomorphic to $\polRing/(\var^{\minDeg_1}) \times \cdots \times
  \polRing/(\var^{\minDeg_\nbun})$ as a $\field$-vector space, and thus has
  dimension $\sumVec{\minDegs}$. Now, this dimension is at most $\vsdim$, since
  $\modApp$ is the kernel of the morphism
  $\app \in \appSpace \mapsto \app \sys \bmod \mods \in
  \polRing/(\var^{\order_1}) \times \cdots \times
  \polRing/(\var^{\order_\nbeq})$, whose codomain has dimension
  $\sumVec{\orders}=\vsdim$ as a $\field$-vector space.

  The matrix $\var^{\max(\orders)}\idMat[\nbun]$ is a left-multiple of
  $\appbas$ since $\var^{\max(\orders)}\idMat[\nbun] \sys = 0 \bmod \mods$;
  thus the inequality $\max(\minDegs) \le \max(\orders)$ follows from the
  predictable degree property \citep{Forney75}.
\end{proof}

\begin{lemma}[{\citet[Lem.\,4.1]{JeNeScVi16}}]
  \label{lem:mindeg_shift}
  Let $\shifts\in\shiftSpace$ and let $\appbas\in\appbasSpace$ be in
  $\shifts$-Popov form with column degree $\minDegs \in \NN^\nbun$. Then
  $\appbas$ is also in $-\minDegs$-Popov form, and we have
  $\rdeg[-\minDegs]{\appbas} = \unishift$.  In particular, for any matrix
  $\reduced \in \appbasSpace$ which is unimodularly equivalent to $\appbas$ and
  $-\minDegs$-reduced, $\reduced$ has column degree $\minDegs$, and $\appbas =
  \leadingMat[-\minDegs]{\reduced}^{-1} \reduced$.
\end{lemma}

Let $\minDegs$ be the $\shifts$-minimal degree of $\modApp$. This result states
that, up to a \emph{constant} transformation, the $\shifts$-Popov basis of
$\modApp$ is equal to any of its $-\minDegs$-minimal bases $\reduced$.
Furthermore, $\cdeg{\reduced} = \minDegs$ implies that $\reduced$ has average
column degree $\sumVec{\minDegs}/\nbun \le \vsdim/\nbun$. We have no such
control on the column degree of $\shifts$-minimal bases when $\shifts$ is not
linked to $\minDegs$, even under assumptions on the shift such as $\hypsmax$,
$\hypsmin$, or $\hypsbal$.

\subsection{Recursive computation of approximant bases}
\label{subsec:recursiveness}

Here, we state the correctness of the approach which consists in computing a
first basis from the input, then a residual instance, then a second basis from
the residual, and finally combining both bases by multiplication to obtain the
output basis. This scheme is followed for example by the iterative algorithms
in \citep{BarBul91,BecLab00} and by the divide and conquer algorithms in
\citep{BecLab94,GiJeVi03}.

In the next lemma, the first and second items focus on minimal bases and extend
\citep[Sec.\,5.1]{BecLab97}; the third item gives a similar result for
ordered weak Popov bases. The fourth item, from \citep[Sec.\,3]{JeNeScVi16},
shows how to retrieve the $\shifts$-minimal degree from two bases in normal
form without computing their product.

\begin{lemma}
  \label{lem:recursiveness}
  Let $\module \subseteq \module_1$ be two $\polRing$-submodules of
  $\polRing^\rdim$ of rank $\rdim$, and let $\popov_1 \in
  \polMatSpace[\rdim]$ be a basis of $\module_1$. Let further $\shifts \in
  \shiftSpace$ and $\shifts[t] = \rdeg[\shifts]{\popov_1}$. Then,  
  \begin{enumerate}[(i)]
    \item The rank of the module $\module_2 = \{ \rowgrk{\lambda} \in
      \polMatSpace[1][\rdim] \mid \rowgrk{\lambda} \popov_1 \in \module \}$
      is $\rdim$, and for any basis $\popov_2 \in \polMatSpace[\rdim]$ of
      $\module_2$, the product $\popov_2 \popov_1$ is a basis of
      $\module$.
    \item If $\popov_1$ is $\shifts$-reduced and $\popov_2$ is
      $\shifts[t]$-reduced,
      then $\popov_2 \popov_1$ is $\shifts$-reduced.
    \item If $\popov_1$ is in $\shifts$-ordered weak Popov form and
      $\popov_2$ is in $\shifts[t]$-ordered weak Popov form, then
      $\popov_2 \popov_1$ is in $\shifts$-ordered weak Popov form.
    \item If $\minDegs_1$ is the $\shifts$-minimal degree of $\module_1$ and
      $\minDegs_2$ is the $\shifts[t]$-minimal degree of $\module_2$, then the
      $\shifts$-minimal degree of $\module$ is $\minDegs_1 + \minDegs_2$.
  \end{enumerate}
\end{lemma}
\begin{proof}
  $(i)$ Let $\mat{A} \in \polMatSpace[\rdim]$ denote the adjugate of
  $\popov_1$. Then, we have $\mat{A} \popov_1 = \det(\popov_1) \idMat[\rdim]$.
  Thus, $\row{p} \mat{A} \popov_1 = \det(\popov_1) \row{p} \in \module$ for all
  $\row{p} \in \module$, and therefore $\module \mat{A} \subseteq \module_2$.
  Now, the nonsingularity of $\mat{A}$ ensures that $\module \mat{A}$ has rank
  $\rdim$; from \citep[Sec.\,12.1, Thm.\,4]{DumFoo04}, this implies that
  $\module_2$ has rank $\rdim$ as well.  The matrix $\popov_2 \popov_1$ is
  nonsingular since $\det(\popov_2\popov_1)\neq0$. Now let $\row{p} \in
  \module$; we want to prove that $\row{p}$ is a $\polRing$-linear combination
  of the rows of $\popov_2 \popov_1$. First, $\row{p} \in \module_1$, so there
  exists $\rowgrk{\lambda} \in \polMatSpace[1][\rdim]$ such that $\row{p} =
  \rowgrk{\lambda} \popov_1$. But then $\rowgrk{\lambda} \in \module_2$, and
  thus there exists $\rowgrk{\mu} \in \polMatSpace[1][\rdim]$ such that
  $\rowgrk{\lambda} = \rowgrk{\mu} \popov_2$. This yields the combination
  $\row{p} = \rowgrk{\mu} \popov_2 \popov_1$.

  $(ii)$ Let $\shifts[d] = \rdeg[{\shifts[t]}]{\popov_2}$; we have $\shifts[d]
  = \rdeg[\shifts]{\popov_2 \popov_1}$ by the predictable degree property.
  Using $\xDiag{-\shifts[d]} \popov_2 \popov_1
  \xDiag{\shifts} = \xDiag{-\shifts[d]} \popov_2 \xDiag{\shifts[t]}
  \xDiag{-\shifts[t]} \popov_1 \xDiag{\shifts}$, we obtain that
  $\leadingMat[\shifts]{\popov_2 \popov_1} =
  \leadingMat[{\shifts[t]}]{\popov_2} \leadingMat[\shifts]{\popov_1}$.  By
  assumption, $\leadingMat[{\shifts[t]}]{\popov_2}$ and
  $\leadingMat[\shifts]{\popov_1}$ are invertible, hence
  $\leadingMat[\shifts]{ \popov_2 \popov_1 }$ is invertible as well; thus
  $\popov_2 \popov_1$ is $\shifts$-reduced.

  $(iii)$ The matrix $\leadingMat[\shifts]{ \popov_2 \popov_1 } =
  \leadingMat[{\shifts[t]}]{\popov_2} \leadingMat[\shifts]{\popov_1}$ is lower
  triangular and invertible.

  $(iv)$ Let $\popov_1$ be the $\shifts$-Popov basis of $\module_1$ and
  $\popov_2$ be the $\shifts[t]$-Popov basis of $\module_2$. Then, by the items
  $(i)$ and $(iii)$ above, $\popov_2 \popov_1$ is a $\shifts$-ordered weak
  Popov basis of $\module$. Thus, from \citep[Lem.\,3.3]{JeNeScVi16}, it is
  enough to show that the $\shifts$-pivot degree of $\popov_2 \popov_1$ is
  $\minDegs_1 + \minDegs_2$, that is, $\rdeg[\shifts]{\popov_2 \popov_1} =
  \shifts + \minDegs_1 + \minDegs_2$. This follows from the predictable degree
  property, since $\rdeg[\shifts]{\popov_2 \popov_1} =
  \rdeg[{\shifts[t]}]{\popov_2} = \shifts[t] + \minDegs_2 =
  \rdeg[\shifts]{\popov_1} + \minDegs_2 = \shifts + \minDegs_1 + \minDegs_2$.
\end{proof}

Now, consider the case where the basis $\popov_1$ of $\module_1$ already has
some rows in $\module$: we show that we may directly store these rows in the
basis of \(\module\) being computed, and that $\popov_2$ can be obtained by
focusing only on the rows of $\popov_1$ not in $\module$. In the next lemma, we
use standard notation for submatrices and subtuples: \(\matrows{\popov}{I}\),
\(\matcols{\popov}{J}\), \(\matsub{\popov}{I}{J}\), \(\subTuple{\shifts}{I}\),
where \(I\) and \(J\) are subsets of \(\{1,\ldots,\rdim\}\).

\begin{lemma}
  \label{lem:ok_rows}
  (Using notation from \cref{lem:recursiveness}.) Let \(I\) be a subset of
  $\{1,\ldots,\rdim\}$ of cardinality \(k\in\{0,\ldots,\rdim\}\) and such that
  all rows of \(\popov_1\) with index in \(I\) are in \(\module\). Let also
  \(I^c = \{1,\ldots,\rdim\}\setminus I\) be the complement of \(I\). Then, the
  module $\module_3 = \{\rowgrk{\mu} \in \polMatSpace[1][(\rdim-k)] \mid
  \rowgrk{\mu} \matrows{(\popov_1)}{I^c} \in \module\}$ has rank $\rdim-k$, and
  for any basis $\popov_3$ of $\module_3$, the matrix $\popov_2 \in
  \appbasSpace$ defined by its submatrices
  \[
    \begin{bmatrix}
      \matsub{(\popov_2)}{I}{I} & \matsub{(\popov_2)}{I}{I^c} \\
      \matsub{(\popov_2)}{I^c}{I} & \matsub{(\popov_2)}{I^c}{I^c} 
    \end{bmatrix}
    =
    \begin{bmatrix}
      \idMat[k] & \matz \\
      \matz & \popov_3
    \end{bmatrix}
  \]
  is a basis of $\module_2$. Furthermore, if $\popov_1$ and $\popov_3$ are in
  $\shifts$- and $\subTuple{\shifts[t]}{I^c}$-ordered weak Popov form, then
  $\popov_2\popov_1$ is an $\shifts$-ordered weak Popov basis of $\module$.
\end{lemma}
\begin{proof}
  Let $\rowgrk{\lambda} \in \polMatSpace[1][\rdim]$, and consider
  \(\rowgrk{\mu} = \matcols{\rowgrk{\lambda}}{I^c} \in
  \polMatSpace[1][(\rdim-k)]\). Then, we have the equivalence $\rowgrk{\lambda}
  \in \module_2 \Leftrightarrow \rowgrk{\mu} \matrows{(\popov_1)}{I^c} \in
  \module$ since the rows of $\matrows{(\popov_1)}{I}$ are already in
  \(\module\).  Hence $\rowgrk{\lambda} \in \module_2 \Leftrightarrow
  \rowgrk{\mu} \in \module_3$, by definition of \(\module_3\). This shows that
  $\module_3$ has rank $\rdim-k$, and since $\popov_3$ is a basis of
  \(\module_3\), we also deduce that $\popov_2$ is a basis of \(\module_2\).

  It is easily verified that if $\popov_3$ is in
  $\subTuple{\shifts[t]}{I^c}$-ordered weak Popov form, then $\popov_2$ is in
  $\shifts[t]$-ordered weak Popov form.  Hence the conclusion, by the first and
  third items of \cref{lem:recursiveness}.
\end{proof}
We remark that the left-multiplication by \(\popov_2\) amounts to simply
copying the submatrix \(\matrows{(\popov_1)}{I}\), and left-multiplying the
submatrix \(\matrows{(\popov_1)}{I^c}\) by \(\popov_3\).

\subsection{Computing residuals}
\label{subsec:residuals}

Approximant basis algorithms commonly make use of \emph{residuals}, which are
truncated matrix products $\appbas \sys \bmod \mods$. Here, we discuss their
efficient computation in two cases: when we control $\deg(\appbas)$, and when
we control the average column degree of $\appbas$.

\begin{lemma}
  \label{lem:compute_residual}
  Let $\appbas \in \appbasSpace$ and $\sys \in \sysSpace$. Then,
  \begin{itemize}
    \item for $\order,\vsdim\in\NN$ such that $\deg(\appbas)\le\order$ and
      $\sumVec{\cdeg{\sys}}\le\vsdim$, one can compute $\appbas \sys$ using
      $\bigOPar{ \left\lceil \frac{\nbeq+\vsdim/(\order+1)}{\nbun} \right\rceil
      \polmatmultime{\nbun,\order}}$ operations in $\field$ if \(\order>0\) and
      $\bigOPar{ \left\lceil \frac{\nbeq+\vsdim}{\nbun} \right\rceil
      \nbun^\expmatmul}$ operations if \(\order=0\);
    \item for $\orders\in\orderSpace$ and $\vsdim \ge \nbun$ such that
      $\sumVec{\orders}\le\vsdim$ and $\sumVec{\cdeg{\appbas}} \le \vsdim$, one
      can compute $\appbas \sys \bmod \mods$ using
      $\bigO{\polmatmultime{\nbun,\vsdim/\nbun}}$ operations in $\field$,
      assuming $\nbeq\le\nbun$.
  \end{itemize}
\end{lemma}
\begin{proof}
  For the first item, we use column partial linearization on $\sys$ to
  transform it into a matrix $\expand{\sys}$ with $\nbun$ rows,
  $\nbeq+\vsdim/(\order+1)$ columns, and degree at most $\order$. Then, we compute
  $\appbas \expand{\sys}$, and the columns of this product are compressed back
  to obtain $\appbas \sys$. More details can be found for example in the
  discussion preceding \citep[Prop.\,4.1]{JeNeScVi17}.

  For the second item, using column partial linearization on $\appbas$ we
  obtain $\expand{\appbas} \in \polMatSpace[\nbun][\expand{\nbun}]$ such that
  $\nbun\le\expand{\nbun} \le 2\nbun$, $\deg(\expand{\appbas}) \le
  \lceil\vsdim/\nbun\rceil$, and $\appbas = \expand{\appbas} \expandMat$ where
  the form of $\expandMat \in \polMatSpace[\expand{\nbun}][\nbun]$ is as in
  \cref{eqn:expandMat}. Then $\appbas \sys \bmod \mods =
  \expand{\appbas}\:\expand{\sys} \bmod \mods$, where $\expand{\sys} =
  \expandMat \sys \bmod \mods$ is obtained for free since each row of
  $\expandMat$ is of the form $[0 \cdots 0 \; \var^\alpha \; 0 \cdots 0]$
  for some $\alpha\in\NN$. Now, up to augmenting $\expand{\appbas}$ with
  $\expand{\nbun}-\nbun$ zero rows, we can apply the first item to compute
  $\expand{\appbas}\: \expand{\sys}$. Here we take
  $\order=\lceil\vsdim/\nbun\rceil$, implying $\vsdim/(\order+1) \le \nbun$ and
  thus $(\nbeq+\vsdim/(\order+1)) / \expand{\nbun} \le 2$, since
  $\expand{\nbun}\ge\nbun\ge\nbeq$. Hence, computing $\expand{\appbas}\:
  \expand{\sys}$ costs
  $\bigO{\polmatmultime{\expand{\nbun},\lceil\vsdim/\nbun\rceil}}$ operations,
  which is within the claimed bound since $\expand{\nbun}\le2\nbun$ and $\vsdim
  \ge \nbun$.
\end{proof}

\subsection{Computing matrix products via approximant bases}
\label{subsec:reductions}

Consider a constant matrix $\sys \in \matSpace[\nbun][\nbeq]$ and $\orders =
(1,\ldots,1)$; note that $\vsdim=\nbeq$. Then, as detailed in
\cref{sec:balanced_order}, finding the $\shifts$-Popov basis of $\modApp$ is
equivalent to computing a left nullspace basis in reduced row echelon form for
the matrix $\sys$ with rows permuted according to the entries of $\shifts$. The
multiplication of constant matrices can be embedded in such nullspace
computations. More generally, any algorithm for \cref{pbm:pab} can be used to
multiply polynomial matrices, following ideas from \citep{SarSto11}.

\begin{lemma}
  \label{lem:reduction_polmatmul}
  Let $\mathcal{P}$ be an algorithm which solves \cref{pbm:pab}. Then, for
  $\mat{A},\mat{B} \in \polMatSpace[\nbun]$ of degree at most $\order$, the
  product $\mat{A}\mat{B}$ can be read off from the output of
  $\mathcal{P}(\orders,\sys,\unishift)$, where
  \[
    \orders=(6\order+4,\ldots,6\order+4) \;\;\;\text{and}\;\;\;
    \sys =
    \begin{bmatrix}
      \var^{2\order+1}\idMat[\nbun] & \mat{B} \\
      -\var^{2\order+1} \mat{A} & \var^{2\order+1}\idMat[\nbun] \\
      -\idMat[\nbun] & \matz \\
      \matz & -\idMat[\nbun]
    \end{bmatrix}
    \in \polMatSpace[4\nbun][2\nbun].
  \]
\end{lemma}
\begin{proof}
  This follows from the results in \citep[Sec.\,4 and\,6]{SarSto11}, which
  imply that the $\unishift$-Popov left kernel basis of $\sys$ is
  \[
    \begin{bmatrix}
      \idMat[\nbun] & \matz & \var^{2\order+1}\idMat[\nbun] & \mat{B} \\
      \mat{A} & \idMat[\nbun] & \matz & \mat{A} \mat{B} + \var^{2\order+1}\idMat[\nbun]
    \end{bmatrix}
  \]
  and appears as the last $2\nbun$ rows of the $\unishift$-Popov
  basis of $\modApp$.
\end{proof}

\subsection{Stability of ordered weak Popov forms under some permutations}
\label{subsec:stability_owP}

When computing a basis of $\modApp$, it is sometimes useful to permute the rows
of $\sys$, that is, to consider $\modAppCustom{\orders}{\permMat\sys}$ for some
$\nbun\times\nbun$ permutation matrix $\permMat$.  Then, it is easily verified
that an $\shifts$-minimal basis $\appbas$ of
$\modAppCustom{\orders}{\permMat\sys}$ yields an $\shifts\permMat$-minimal
basis $\appbas\permMat$ of $\modApp$. However, the more specific weak Popov
forms are not preserved in this process: if $\appbas$ is in $\shifts$-weak Popov
form, then the column permuted basis $\appbas\permMat$ might for example have
all its $\shifts\permMat$-pivot entries in its last column.  Still, for
specific permutations and when considering a submatrix of $\appbas\permMat$, we
have the following result (we remark that it will only be used in
\cref{subsec:weakly_unbalanced_around_min}).

\begin{lemma} \label{lem:permuted_owpopov}
  Let $1 \le \cdim < \rdim$ and consider a partition $\{1,\ldots,\rdim\} =
  \{i_1,\ldots,i_\cdim\} \cup \{j_1,\ldots,j_{\rdim-\cdim}\}$
  with $(i_k)_k$ and $(j_k)_k$ both strictly increasing. Let further $\permMat
  = (\perm_{i,j})$ be the $\rdim\times\rdim$ permutation matrix such that
	$\perm_{k,i_k} = 1$ for $1\le k\le \cdim$ and $\perm_{k+\cdim,j_k} = 1$ for
	$1 \le k \le \rdim-\cdim$, and let $\shifts = (\shift{j}) \in \shiftSpace$.
	Then,
  \begin{itemize}
		\item if a matrix $\popov \in \polMatSpace[\rdim]$ is in $\shifts$-ordered
			weak Popov form, then the leading principal $\cdim\times\cdim$ submatrix
			of $\permMat \popov \permMat^{-1}$ is in
			$(\shift{i_1},\ldots,\shift{i_\cdim})$-ordered weak Popov form;
    \item for a tuple $\orders \in \NN^{\rdim-\cdim}$ and matrices $\popov \in
      \polMatSpace[\cdim]$ and $\mat{Q} \in
      \polMatSpace[\cdim][(\rdim-\cdim)]$, if the matrix
      \[
        \matt{P} =
        \begin{bmatrix}
          \popov & \mat{Q} \\
          \matz  & \xDiag{\orders}
        \end{bmatrix}
        \in \polMatSpace[\rdim]
      \]
      is in $\shifts$-ordered weak Popov, then $\permMat^{-1} \matt{P}
      \permMat$ is in $\shifts\permMat$-ordered weak Popov form.
  \end{itemize}
\end{lemma}
\begin{proof}
	Concerning the first item, let $\tuple{t} =
	(\shift{i_1},\ldots,\shift{i_\cdim})$ and write $[p_{i,j}]$ for the entries
	of $\popov$. Then, 
	the leading principal $\cdim\times\cdim$ submatrix of $\permMat \popov
	\permMat^{-1}$ is $[p_{i_k,i_\ell}]_{1\le k,\ell \le \cdim}$. Now,
	$\leadingMat[\tuple{t}]{[p_{i_k,i_\ell}]}$ is the submatrix of
	$\leadingMat[\shifts]{\popov}$ formed by its rows and columns indexed by
	$(i_1,\ldots,i_\cdim)$, and $\leadingMat[\shifts]{\popov}$ is unit lower
	triangular since $\popov$ is in $\shifts$-ordered weak Popov form.  Since
	$i_1 < \cdots < i_\cdim$, $\leadingMat[\tuple{t}]{[p_{i_k,i_\ell}]}$ is unit
	lower triangular as well, and therefore $[p_{i_k,i_\ell}]_{1\le k,\ell \le
	\cdim}$ is in $\shifts[t]$-ordered weak Popov form.

  For the second item, we prove that the $\shifts\permMat$-leading matrix of
  $\permMat^{-1} \matt{P} \permMat$ is unit lower triangular. For $1 \le k \le
  \rdim-\cdim$, the row $j_k$ of $\permMat^{-1} \matt{P} \permMat$ is $[0 \;
  \cdots \; 0 \; \var^{d_k} \; 0 \; \cdots \; 0]$ with $\var^{d_k}$ at index
  $j_k$; thus, the row $j_k$ of $\leadingMat[\shifts\permMat]{\permMat^{-1}
  \matt{P} \permMat}$ is $[0 \; \cdots \; 0 \; 1 \; 0 \cdots \; 0]$ with $1$ on
  the diagonal. It remains to show that, for $1 \le k \le \cdim$, the row $i_k$
  of $\leadingMat[\shifts\permMat]{\permMat^{-1} \matt{P} \permMat}$ has the
  form $[\any \; \cdots \; \any \; 1 \; 0 \; \cdots \; 0]$ with $1$ on the
  diagonal, that is, at index $i_k$. The row $i_k$ of
  $\leadingMat[\shifts\permMat]{\permMat^{-1} \matt{P} \permMat}$ is the row
  $k$ of $\leadingMat[\shifts]{\matt{P}}\permMat$; the latter has the desired
  form $[\any \; \cdots \; \any \; 1 \; 0 \; \cdots \; 0]$ with $1$ at index
  $i_k$, since the row $k$ of $\leadingMat[\shifts]{\matt{P}}$ has the form
  $[\any \; \cdots \; \any \; 1 \; 0 \; \cdots \; 0]$ with $1$ at index $k$ and
  since $i_1 < \cdots < i_k$. 
\end{proof}

\section{Algorithm \textsc{PM-Basis}: approximant bases via polynomial matrix multiplication}
\label{sec:balanced_order}

In this section, we focus on the case of a uniform order, that is, $\orders =
(\order,\ldots,\order) \in \orderSpace$ and $\vsdim=\nbeq\order$. For
simplicity, we write $\modAppCustom{\order}{\sys}$ to refer to
$\modAppCustom{(\order,\ldots,\order)}{\sys}$. Then, for any shift,
\citep[Algo.\,\algoname{PM-Basis}]{GiJeVi03} computes an $\shifts$-minimal
basis of $\modAppCustom{\order}{\sys}$ using
$\bigO{(1+\nbeq/\nbun)\appbastime{\nbun,\order}}$ operations; this is in
$\softO{\nbun^{\expmatmul-1}\vsdim}$ when $\nbeq\in\Omega(\nbun)$.

\algoname{PM-Basis} follows a divide and conquer approach, splitting the
instance at order $\order$ into two instances at order $\order/2$ and combining
the recursively obtained bases by polynomial matrix multiplication. The base
case ($\order=1$) is solved via fast dense linear algebra over the field
\(\field\).  Here, we describe \algoname{PM-Basis} with a modified base case,
ensuring that it returns the normalized basis. As a consequence, the whole
algorithm returns an $\shifts$-ordered weak Popov basis; this has the advantage
of directly revealing the $\shifts$-minimal degree of
$\modAppCustom{\order}{\sys}$, a fact used multiple times in this paper.

We now consider the base case: $\order=1$ and $\sys\in\matSpace[\nbun][\nbeq]$
is constant. Then, we will see that the $\shifts$-Popov basis of
$\modAppCustom{1}{\sys}$ has two sets of rows: rows corresponding to a
nullspace basis for $\sys$, and elementary rows of the form $[0 \;\; \cdots
\;\; 0 \;\; \var \;\; 0 \;\; \cdots \;\; 0]$. \cref{algo:lin_pab} is a modified
version of \citep[Algo.\,M-Basis with $d=1$]{GiJeVi03}, and also a
specialization of \citep[Algo.\,9]{JeNeScVi17} when the multiplication matrix
is zero.

\begin{algobox}
  \algoInfo
  {Popov basis at order {$(1,\ldots,1)$}}
  {M-Basis-1}
  {algo:lin_pab}

  \dataInfos{Input}{
    \item constant matrix $\sys\in\matSpace[\nbun][\nbeq]$,
    \item shift $\shifts \in \shiftSpace$.
  }

  \dataInfo{Output}{
    the $\shifts$-Popov basis of $\modAppCustom{1}{\sys}$.
  }

  \algoSteps{
    \item $\pi_{\shifts} \assign$ $\nbun\times\nbun$ permutation matrix such
      that $\pi_{\shifts} \, \trsp{[(\shift{1},1) ~ \cdots ~ (\shift{\nbun},\nbun)]}$ is
      lexicographically increasing
    \item $(\boldsymbol\rho,\mat{L}) \in \ZZp^\genRank \times
      \matSpace[\nbun][\nbun] \assign$ row rank profile of $\pi_{\shifts}
      \sys$, and L-factor in the LSP decomposition of $\pi_{\shifts} \sys$,
      where $\matcol{\mat{L}}{j}$ is an identity column for
      $j\not\in\boldsymbol\rho$
    \item $\mat{M} \in \matSpace[\nbun]$ $\assign$ matrix whose $i$th row is
      $\matrow{\mat{L}}{i}$ with negated off-diagonal entries if
      $i\not\in\boldsymbol{\rho}$, and is the identity row if
      $i\in\boldsymbol{\rho}$
    \item $\matt{P} \in \appbasSpace \assign$ the matrix $\xDiag{\boldsymbol{\mu}}\mat{M}$
    with $\boldsymbol{\mu} = (\mu_1,\ldots,\mu_m)$ such that $\mu_i=1$ if $i\in\boldsymbol{\rho}$, and
    $\mu_i=0$ otherwise
    \item \algoword{Return} $\pi_{\shifts}^{-1} \matt{P} \pi_{\shifts}$
    }
\end{algobox}

\begin{proposition}
  \label{prop:algo:lin_pab}
  \cref{algo:lin_pab} is correct and uses $\bigO{\genRank^{\expmatmul-2} \nbun
\nbeq}$ operations in $\field$, where $\genRank$ is the rank of $\sys$.
\end{proposition}
\begin{proof}
  Concerning the cost bound, the LSP decomposition at Step~\textbf{2} uses
  $\bigO{\genRank^{\expmatmul-2} \nbun \nbeq}$ operations
  \citep[Sec.\,2.2]{Storjohann00}, and reveals the row rank profile.

  For the correctness, we prove the following three properties: all the rows of
  the output $\appbas=\pi_{\shifts}^{-1} \mat{\hat{P}} \pi_{\shifts}$ are in
  $\modAppCustom{1}{\sys}$, the rows of $\appbas$ generate
  $\modAppCustom{1}{\sys}$, and $\appbas$ is in $\shifts$-Popov form.

  First, we have that $\appbas \sys = 0 \bmod \var$ since the rows of $\appbas$
  are either multiples of $\var$ or, by definition of $\mat{M}$, in the left
  nullspace of $\sys$. Indeed, by property of the LSP decomposition,
  the rows $\matrow{\mat{L}}{i}$ with negated off-diagonal entries for
  all $i\not\in\boldsymbol{\rho}$ form a basis of the left nullspace of
  $\pi_{\shifts} \sys$.

  Second, we show that any $\app\in\modAppCustom{1}{\sys}$ belongs to the
  row space of $\appbas$. Writing $\app = \row{q}\var + \row{r}$ with $\row{q}
  \in \polMatSpace[1][\nbun]$ and $\row{r}
  \in \matSpace[1][\nbun]$, we have the identity $\row{q} X = \row{q}
  \pi_{\shifts}^{-1} \mat{M}^{-1} \xDiag{1-\boldsymbol{\mu}} \pi_{\shifts}
  \appbas$. Furthermore, $\app\sys = \row{r}\sys = \row{r} \pi_{\shifts}^{-1}
  \pi_{\shifts} \sys = 0 \bmod \var$, and therefore $\row{r} \pi_{\shifts}^{-1} =
  \rowgrk{\lambda} \mat{M}$ for some $\rowgrk{\lambda} = [\lambda_i]_i
  \in\matSpace[1][\nbun]$ such that $\lambda_i = 0$ if $i\in\boldsymbol{\rho}$.
  Recalling that $\mu_i=0$ if $i\not\in\boldsymbol{\rho}$, we obtain $\row{r} =
  \rowgrk{\lambda} \xDiag{\boldsymbol{\mu}} \mat{M} \pi_{\shifts} =
  \rowgrk{\lambda} \pi_{\shifts} \appbas$.

  Finally, we prove that $\appbas$ is in $\shifts$-Popov form.
  By construction, $\matcol{\mat{\hat{P}}}{j}$ is the $j$th column of the
  identity if $j\not\in\boldsymbol{\rho}$, while for $j\in\boldsymbol{\rho}$,
  it has constants everywhere but at position $j$, where $\hat{p}_{jj} = X$.
  It follows that $\leadingMat[\unishift]{\trsp{\mat{\hat{P}}}} = \idMat$, and
  it is then easily checked that $\leadingMat[\unishift]{\trsp{\appbas}} =
  \idMat$.

  It remains to prove that $\leadingMat[\shifts]{\appbas}$ is unit lower
  triangular, or, equivalently, that
   \begin{equation}
    \label{eqn:lmat_low_tri}
    p_{ii} \text{ is monic and }
    \left\{
      \begin{array}{ll}
        \deg(p_{ij})+\shift{j} \le \deg(p_{ii})+\shift{i} & \,\, \text{if } j \le i, \\
        \deg(p_{ij})+\shift{j} < \deg(p_{ii})+\shift{i}   & \,\, \text{if } j > i.
      \end{array}
     \right. 
     \qquad\qquad\qquad\qquad\quad
  \end{equation} 
  where $p_{ij}$ is the entry of $\appbas$ at $(i,j)$.
  Writing $[1 \;\cdots\; \nbun] \pi_{\shifts} = [\pi_1 \;\cdots\; \pi_\nbun]$,
  we have $p_{ij} = \hat{p}_{\pi_i \pi_j}$ for all $i,j$. If
  $\matrow{\appbas}{i}$ is nonconstant, then so is
  $\matrow{\mat{\hat{P}}}{\pi_i}$ and thus, by construction, its only nonzero
  entry is $\hat{p}_{\pi_i\pi_i} = X$. Hence $\matrow{\mat{P}}{i} =
  [0 \;\cdots\; 0 \; X \; 0 \;\cdots \; 0]$ with $X$ at index $i$, so
  that \cref{eqn:lmat_low_tri} holds.

  Let now $\matrow{\mat{P}}{i}$ be a constant row. In this case,
  $\matrow{\mat{\hat{P}}}{\pi_i}$ is constant as well and $\hat{p}_{\pi_i\pi_i}
  = 1$. Consequently, $p_{ii} = 1$ and \cref{eqn:lmat_low_tri} is now
  equivalent to
  \[
    \text{if } (j\le i \text{ and } s_j>s_i)
    \text{ or }
    (j>i \text{ and } s_j\ge s_i),
    \text{ then } p_{ij} = 0.
  \]
  Now, by definition of $\pi_{\shifts}$, if $i$ and $j$ are such that
  $\shift{j} > \shift{i}$, or such that $\shift{j} \ge \shift{i}$ and $j>i$,
  then $\pi_j > \pi_i$. Since $\mat{\hat{P}}$ is lower triangular, this implies
  $\hat{p}_{\pi_i\pi_j} = 0$, that is, $p_{ij} = 0$.
\end{proof}

Now, we recall \algoname{PM-Basis} in \cref{algo:dac_ab}. Note that it computes
a basis of degree at most $\order$, although there often exist
$\shifts$-minimal bases with larger degree. As a result, the two bases obtained
recursively can be multiplied in $\polmatmultime{\nbun,\order}$ operations.

\begin{algobox}
  \algoInfo
  {Minimal basis for a uniform order}
  {PM-Basis}
  {algo:dac_ab}

  \dataInfos{Input}{
    \item order $\order \in \ZZp$,
    \item matrix $\sys\in\sysSpace$ of degree less than $\order$,
    \item shift $\shifts \in \shiftSpace$.
  }

  \dataInfos{Output}{
    \item an $\shifts$-ordered weak Popov basis of
      $\modAppCustom{\order}{\sys}$ of degree at most $\order$.
  }

  \algoSteps{
    \item \algoword{If} $\order=1$ \algoword{then return} $\algoname{M-Basis-1}(\sys,\shifts)$
    \item \algoword{Else:}
      \begin{enumerate}[{\bf a.}]
        \item $\appbas_1 \assign \algoname{PM-Basis}(\lceil\order/2\rceil,\sys\bmod \var^{\lceil\order/2\rceil},\shifts)$
        \item $\mat{G} \assign (\var^{-\lceil\order/2\rceil} \appbas_1 \sys) \bmod \var^{\lfloor\order/2\rfloor}$;~
         $\shifts[t] \assign \rdeg[\shifts]{\appbas_1}$
        \item $\appbas_2 \assign \algoname{PM-Basis}(\lfloor\order/2\rfloor,\mat{G},\shifts[t])$
        \item \algoword{Return} $\appbas_2 \appbas_1$
      \end{enumerate}
  }
\end{algobox}

\begin{proposition}
  \label{prop:algo:dac_ab}
  \cref{algo:dac_ab} is correct and uses $\bigO{(1+\frac{\nbeq}{\nbun})
  \appbastime{\nbun,\order}}$ operations in~$\field$.
\end{proposition}
\begin{proof}
  From \cref{prop:algo:lin_pab}, Step~\textbf{1} computes the $\shifts$-Popov
  basis of $\modAppCustom{1}{\sys}$, which has degree at most $1$.
  Then, it follows by induction that the output has degree at most $\order =
  \lceil\order/2\rceil + \lfloor\order/2\rfloor$, and items $(i)$ and $(iii)$
  of \cref{lem:recursiveness} prove the correctness.

  For the cost analysis, let us assume that $\order$ is a power of $2$. From
  \cref{prop:algo:lin_pab}, Step~\textbf{1} uses
  $\bigO{\nbun^{\expmatmul-1}\nbeq}$ operations. The tree of the recursion has
  $\order$ leaves, which altogether account for
  $\bigO{\nbun^{\expmatmul-1}\nbeq \order}$ field operations. Note that
  $\nbun^{\expmatmul-1}\nbeq \order \in
  \bigO{\frac{\nbeq}{\nbun}\appbastime{\nbun,\order}}$.

  Then, there are recursive calls at Steps~\textbf{2.a} and~\textbf{2.c}, in
  dimension $\nbun$ and at order $\order/2$. The residual $\mat{G}$ at
  Step~\textbf{2.b} is obtained from the product $\appbas_1 \sys$, where
  $\appbas_1$ is an $\nbun\times\nbun$ matrix of degree at most $\order/2$, and
  $\sys$ is an $\nbun\times\nbeq$ matrix of degree at most $\order$. This
  product is done in $\bigO{\polmatmultime{\nbun,\order}}$ operations if $\nbeq
  \le \nbun$, and in $\bigO{\frac{\nbeq}{\nbun} \polmatmultime{\nbun,\order}}$
  operations if $\nbun \le \nbeq$. The multiplication at Step~\textbf{2.d}
  involves two $\nbun\times\nbun$ matrices of degree at most $\order/2$, and
  hence is done in $\bigO{\polmatmultime{\nbun,\order/2}}$ operations in
  $\field$. The cost bound follows from the definition of the cost function
  \(\appbastime{\cdot,\cdot}\).
\end{proof}

Based on \cref{lem:mindeg_shift}, we show how to obtain the $\shifts$-Popov
approximant basis using two calls to \algoname{PM-Basis}
(\cref{algo:popov_pmbasis}). This yields an efficient solution to
\cref{pbm:pab} when $\nbeq \in \Omega(\nbun)$ and the order is balanced as in
$\hypobal$, and this proves the first item of \cref{thm:pm_basis}.  Note that
here we allow the order to be non-uniform, based on the following remark.

\begin{remark}
  \label{rmk:pmbasis_unbalancedorders}
  Let $\orders \in \orderSpace$ and $\sys \in \sysSpace$. Then, for any
  $\orders' \in \orderSpace$ such that $\orders' \ge \orders$, we have $\modApp
  = \modAppCustom{\orders'}{\sys \xDiag{\orders'-\orders}}$. In particular,
  algorithms for uniform orders can be used to solve the case of arbitrary
  orders: for $\order = \max(\orders)$ and $\mat{G} = \sys
  \xDiag{(\order,\ldots,\order)-\orders}$, we have $\modApp =
  \modAppCustom{\order}{\mat{G}}$. For example, for a balanced order ($\orders$
  such that $\hypobal$: $\order\in\bigO{\vsdim/\nbeq}$),
  \algoname{PM-Basis} uses
  $\bigO{(1+\nbeq/\nbun)\appbastime{\nbun,\vsdim/\nbeq}}$ operations, where
  $\vsdim = \sumVec{\orders}$.
\end{remark}

\begin{algobox}
  \algoInfo
  {Popov basis via PM-Basis}
  {Popov-PM-Basis}
  {algo:popov_pmbasis}

  \dataInfos{Input}{
    \item order $\orders \in \orderSpace$,
    \item matrix $\sys \in \sysSpace$ with $\cdeg{\sys} < \orders$,
    \item shift $\shifts \in \shiftSpace$.
  }

  \dataInfo{Output}{
    the $\shifts$-Popov basis of $\modApp$.
  }

  \algoSteps{
    \item $\order \assign \max(\orders)$; $\res \assign \sys
      \xDiag{(\order,\ldots,\order)-\orders}$
    \item $\appbas \assign \algoname{PM-Basis}(\order,\res,\shifts)$
    \item $\minDegs \assign$ the diagonal degrees of $\appbas$
    \item $\reduced \assign \algoname{PM-Basis}(\order,\res,-\minDegs)$
    \item \algoword{Return} $\leadingMat[-\minDegs]{\reduced}^{-1} \reduced$
  }
\end{algobox}

The correctness of \cref{algo:popov_pmbasis} follows from that of
\algoname{PM-Basis}, and from
\cref{lem:mindeg_shift,rmk:pmbasis_unbalancedorders}. Besides, the cost bound
$\bigO{(1 + \nbeq/\nbun) \appbastime{\nbun,\order}}$ follows from
\cref{prop:algo:dac_ab}, noting that Step~\textbf{5} uses
$\bigO{\nbun^{\expmatmul}\order} \subseteq \bigO{\appbastime{\nbun,\order}}$
operations since $\deg(\reduced)\le\order$.

\section{Reduction to the case \texorpdfstring{$\nbeq < \nbun$}{\nbeq less than \nbun}}
\label{sec:reduce_nbeq}

Let $\orders = (\order_1,\ldots,\order_\nbeq) \in \orderSpace$, $\sys \in
\sysSpace$ such that $\cdeg{\sys} < \orders$, and let $\shifts \in
\shiftSpace$. In this section we assume $\nbeq \ge \nbun$, which also implies
$\vsdim = \order_1+\cdots+\order_\nbeq \ge \nbun$, and we present an efficient
procedure relying on \algoname{PM-Basis} to reduce to the case $\nbeq < \nbun$.

Here is an overview of the reduction, assuming
$\order_1\ge\cdots\ge\order_\nbeq$ for simplicity. The idea is to efficiently
compute a basis $\appbas$ of a truncated instance, namely of
$\modAppCustom{\orders'}{\sys \bmod \xDiag{\orders'}}$ for the order
\[
  \orders'=(\order_\nbun,\ldots,\order_\nbun,\order_{\nbun+1},\ldots,\order_\nbeq)
\in \orderSpace.
\]
Then, the residual instance consists of the order $\ordersR = \orders -
\orders'$ and the matrix $\sysR = \appbas \sys \xDiag{-\orders'} \bmod
\xDiag{\ordersR}$: by \cref{lem:recursiveness}, for any basis $\appbasR$ of
$\modAppCustom{\ordersR}{\sysR}$, the product $\appbasR \appbas$ is a basis of
$\modApp$. By construction, the residual matrix \(\sysR\) has \(\nbun\) rows
and less than \(\nbun\) nonzero columns.

In \cref{algo:reduce_coldim}, we detail how to efficiently obtain $\appbas$ and
the residual instance $(\ordersR,\sysR)$. We now sketch this algorithm,
assuming that $\order_\nbun,\ldots,\order_\nbeq$ are powers of $2$ for ease of
presentation. How to reduce to this case follows from
\cref{rmk:pmbasis_unbalancedorders}.

Then, denoting by $\ell$ the integer such that $\order_\nbun = 2^\ell$, we
define
\[
  \nu_i = \card{\{ j \in \{1,\ldots,\nbeq\} \mid \order_j = 2^i \}}
\]
for $0\le i\le \ell$, as well as $\nu_{\ell+1} = \nbeq-\nu_0-\cdots-\nu_\ell$.
This can be illustrated as follows:
\begin{equation*}
  \orders \;=\;
  ( ~ 
  \lefteqn{\overbrace{\phantom{> 2^\ell, \ldots, > 2^\ell, ~ 2^{\ell},\ldots}}^{\nbun}}
  \underbrace{> 2^\ell, \ldots, > 2^\ell}_{\nu_{\ell+1}},
  ~ \underbrace{2^{\ell},\ldots,2^\ell}_{\nu_\ell},
  ~ \underbrace{2^{\ell-1},\ldots,2^{\ell-1}}_{\nu_{\ell-1}},
  ~ \ldots,
  ~ \underbrace{1,\ldots,1}_{\nu_{0}} ~).
\end{equation*}
Furthermore, we let $\mu_{i} = \nu_{\ell+1} + \nu_\ell + \cdots + \nu_{i} =
\max \{j \mid \order_j \ge 2^i\}$. Then, guided by this decomposition of
$\orders$, we obtain $\appbas$ in $\bigO{\appbastime{\nbun,\vsdim/\nbun}}$
operations via $\ell+1$ calls to \algoname{PM-Basis}. This is faster than the
straightforward approach consisting in a single call to \algoname{PM-Basis}
with order $\order=2^\ell\in \bigO{\vsdim/\nbun}$, which uses
$\bigO{\frac{\nbeq}{\nbun}\appbastime{\nbun,\vsdim/\nbun}}$ operations. 

The first call is with $\order=1$ and computes an approximant basis $\appbas_0$
for all $\mu_0 = \nbeq$ columns of $\sys \bmod \var$. After this, we are left
with the residual matrix $\res = \var^{-1} \appbas_0 \sys$ and the order
$(\order_1-1,\ldots,\order_\nbeq-1)$, whose last $\nu_0$ entries are zero.
Thus, the second call is with $\order=2^1-2^0 = 1$ and for the first
$\mu_1=\nbeq-\nu_0$ columns of $\res \bmod \var$, giving an approximant basis
$\appbas_1$. Then $\appbas_1\appbas_0$ is a basis of
$\modAppCustom{(2,\ldots,2)}{\sys}$. Considering the residual $\res = \var^{-2}
\appbas_1 \appbas_0 \sys$, the third call is with $\order=2^2-2^1=2$ and for
the first $\mu_2$ columns of $\res \bmod \var^2$, yielding an approximant basis
$\appbas_2$. Thus, $\appbas_2\appbas_1\appbas_0$ is a basis of
$\modAppCustom{(4,\ldots,4)}{\sys}$. Continuing this process until reaching the
order $(2^\ell,\ldots,2^\ell)$, we obtain $\appbas = \appbas_\ell \cdots
\appbas_0$ and we are left with a residual matrix having at most
$\nu_{\ell+1}=\mu_{\ell+1} < \nbun$ nonzero columns.

\begin{algobox}[htb]
  \algoInfo
  {Reduction to $\nbeq<\nbun$ via \algoname{PM-Basis}}
  {ReduceColDim}
  {algo:reduce_coldim}

  \dataInfos{Input}{
    \item order $\orders = (\order_1, \ldots, \order_\nbeq) \in
      \orderSpace$ with $\order_1 \ge \cdots \ge \order_\nbeq$,
    \item matrix $\sys \in \sysSpace$ with $\cdeg{\sys} <
      \orders$ and $\nbeq\ge\nbun$,
    \item shift $\shifts \in \shiftSpace$.
  }

  \dataInfos{Output}{
    \item $\ordersR = (\order_1-\order_\nbun,\ldots,\order_\nu-\order_\nbun)
      \in \ZZp^\nu$, where $\nu = \max \{ j \mid \order_j > \order_\nbun\}$,
    \item $\sysR = \var^{-\order_\nbun} 
      [ (\appbas \matcol{\mat{F}}{1}) \bmod \var^{\order_1} 
        | \cdots |
        (\appbas \matcol{\mat{F}}{\nu}) \bmod \var^{\order_{\nu}} ]
        \in \polMatSpace[\nbun][\nu]$, 
    \item $\shiftsR = \rdeg[\shifts]{\appbas} \in \shiftSpace[\nbun]$,
    \item $\appbas$ an $\shifts$-ordered weak Popov basis of
      $\modAppCustom{\orders-(\ordersR,\unishift)}{\sys}$.
  }

  \algoSteps{
    \item $\tilde\order_j \assign 2^{\lceil \log_2(\order_j)\rceil}$ for $\nbun\le j \le \nbeq$;
      and $\tilde\order_j \assign \order_j + \tilde\order_\nbun - \order_\nbun$ for $1 \le j < \nbun$
    \item $\mat{\tilde{F}} \assign \sys \xDiag{\tuple{\tilde\order}-\orders}$ where $\tuple{\tilde\order} = (\tilde{\order}_1,\ldots,\tilde{\order}_\nbeq)$
    \item $\ell \assign \log_2(\tilde\order_\nbun)$; $\mu_i \assign \max\{j
      \mid \tilde\order_j \ge 2^i\}$ for $1\le i\le \ell$; and $\nu \assign
      \max\{j \mid \tilde\order_j > 2^\ell\}$
    \item $\appbas \assign \algoname{M-Basis-1}(\mat{\tilde{F}} \bmod \var,\shifts)$
    \item \algoword{For} $i$ \algoword{from} $1$ \algoword{to} $\ell$:
      \begin{enumerate}[{\bf a.}]
        \item $\res \assign (\var^{-2^{i-1}} \appbas [\matcol{\mat{\tilde{F}}}{1} | \cdots | \matcol{\mat{\tilde{F}}}{\mu_{i}}]) \bmod \var^{2^{i-1}}$
        \item $\appbas_i \assign
          \algoname{PM-Basis}(2^{i-1},\res,\rdeg[\shifts]{\appbas})$
        \item $\appbas \assign \appbas_i \appbas$
      \end{enumerate}
    \item $\ordersR \assign (\order_1-\order_\nbun,\ldots,\order_{\nu}-\order_\nbun)$;
      and $\shiftsR \assign \rdeg[\shifts]{\appbas}$
    \item $\sysR \assign \var^{-\order_\nbun} [ (\appbas \matcol{\mat{F}}{1}) \bmod
        \var^{\order_1} | \cdots | (\appbas \matcol{\mat{F}}{\nu}) \bmod
      \var^{\order_{\nu}} ]$
    \item \algoword{Return} $(\ordersR,\sysR,\shiftsR,\appbas)$
  }
\end{algobox}

\begin{proposition}
  \label{prop:algo:reduce_coldim}
  \cref{algo:reduce_coldim} is correct and uses
  $\bigO{\appbastime{\nbun,\vsdim/\nbun}}$ operations in $\field$, where
  $\vsdim = \order_1+\cdots+\order_\nbeq$. Furthermore, the output is such that
  $\sysR$ has $\nbun$ rows and $\nu<\nbun$ columns, $\sumVec{\ordersR} \le
  \vsdim$, $\deg(\appbas) \le 2\vsdim/\nbun$, and for any basis
  $\mat{Q}$ of $\modAppCustom{\ordersR}{\sysR}$, then $\mat{Q}\appbas$ is a
  basis of $\modApp$.
\end{proposition}
\begin{proof}
  Steps~\textbf{1} and~\textbf{2} compute $\tuple{\tilde\order}$ and
  $\mat{\tilde{F}}$ such that $(\tilde\order_i)_{i\ge\nbun}$ are the smallest powers of
  two larger than or equal to $(\order_i)_{i\ge\nbun}$, and $\modApp =
  \modAppCustom{\tuple{\tilde\order}}{\mat{\tilde{F}}}$ (see
  \cref{rmk:pmbasis_unbalancedorders}). Step~\textbf{3} defines parameters, and
  Step~\textbf{4} computes the $\shifts$-Popov basis $\appbas$ of
  $\modAppCustom{(1,\ldots,1)}{\mat{\tilde{F}}}$.

  Then, \cref{lem:recursiveness} shows that we have the following invariant for
  the loop at Step~\textbf{5}: at the end of the iteration $i$, $\appbas$ is an
  $\shifts$-ordered weak Popov approximant basis for $\mat{\tilde{F}}$ at order
  $(2^i,\ldots,2^i,\tilde\order_{\mu_i+1},\ldots,\tilde\order_\nbeq)$. Thus,
  after exiting the loop, $\appbas$ is an $\shifts$-ordered weak Popov
  approximant basis for $\mat{\tilde{F}}$ at order
  \[
   (2^\ell,\ldots,2^\ell,\tilde\order_{\mu_{\ell}+1},\ldots,\tilde\order_\nbeq)
   = (\tilde\order_\nbun,\ldots,\tilde\order_\nbun,\tilde\order_{\nbun+1},\ldots,\tilde\order_\nbeq)
   = \tuple{\tilde\order} - (\ordersR,\unishift).
  \]
  By choice of $\mat{\tilde{F}}$, we obtain that $\appbas$ is an approximant
  basis for $\sys$ at order
  \[
    \orders - (\ordersR,\unishift) =
    (\order_\nbun,\ldots,\order_\nbun,\order_{\nbun+1},\ldots,\order_\nbeq).
  \]
  In particular, it follows from \cref{lem:recursiveness} that $\mat{Q}\appbas$
  is a basis of $\modApp$.

  Now, concerning the cost bound, \cref{prop:algo:lin_pab} states that
  Step~\textbf{4} costs $\bigO{\nbun^{\expmatmul-1}\nbeq}$ operations, since
  $\nbeq\ge\nbun$. This is within $\bigO{\appbastime{\nbun,\vsdim/\nbun}}$,
  since we have $\nbun^{\expmatmul-1}\nbeq \in
  \bigO{\appbastime{\nbun,\nbeq/\nbun}}$, with $\nbeq\le\vsdim$. The resulting
  basis $\appbas$ has degree at most $1$.

  To obtain the residual at Step~\textbf{5.a}, we compute $\appbas
  [\matcol{\mat{\tilde{F}}}{1}|\cdots|\matcol{\mat{\tilde{F}}}{\mu_i}] \bmod
  \var^{2^i}$; this is done in $\bigO{\frac{\mu_i}{\nbun}
  \polmatmultime{\nbun,2^i}}$ operations since $\mu_i \ge \nbun$. Then,
  according to \cref{prop:algo:dac_ab}, Step~\textbf{5.b} uses
  $\bigO{\frac{\mu_i}{\nbun} \appbastime{\nbun,2^{i-1}}}$ operations and
  $\deg(\appbas_i) \le 2^{i-1}$. Thus, at Step~\textbf{5.c} we multiply two
  $\nbun\times\nbun$ matrices of degree at most $2^{i-1}$, which uses
  $\bigO{\polmatmultime{\nbun,2^{i}}}$ operations.

  Altogether, the loop at Step~\textbf{5} uses
  $
    \bigO{\sum_{1\le i\le\ell} \frac{\mu_i}{\nbun} \appbastime{\nbun,2^{i-1}}}
    \;\subseteq\; \bigO{\appbastime{\nbun,\vsdim/\nbun}}
  $
  operations in $\field$, where we prove the inclusion as follows. By
  definition of $\appbastime{\cdot,\cdot}$,
  \begin{align*}
    \sum_{1\le i\le\ell} \frac{\mu_i}{\nbun} \appbastime{\nbun,2^{i-1}} & \;=\;
    \sum_{1\le i\le\ell,\;0\le k < i} \frac{\mu_i}{\nbun} 2^{i-1-k} \polmatmultime{\nbun,2^{k}} \\
    & \;\le\; 2 \sum_{0\le k<\ell} 2^{-k} \frac{\vsdim}{\nbun} \polmatmultime{\nbun,2^{k}}
    \;\le\; 4 \appbastime{\nbun,\vsdim/\nbun}.
  \end{align*}
  Both inequalities are consequences of the construction of
  $\tuple{\tilde\order}$: the first one follows from
  \[
    2\vsdim \ge \sumVec{\tuple{\tilde\order}} =
    \tilde\order_1 + \cdots + \tilde\order_\nu
    + (\mu_\ell-\mu_{\ell+1}) 2^\ell + \cdots + (\mu_1-\mu_2) 2 + (\nbeq-\mu_1)
    \ge \textstyle\sum_{1\le i\le\ell} \mu_i 2^{i-1},
  \]
  while the second one comes from the fact that we have $\ell-1 \le
  \log(\vsdim/\nbun)$, since
  \[
    \nbun 2^{\ell} = \nbun \tilde\order_\nbun \le \tilde\order_1 + \cdots +
    \tilde\order_\nbun \le \sumVec{\tuple{\tilde\order}} \le 2\vsdim.
  \]

  Finally, the matrix $\sysR$ at Step~\textbf{7} is directly obtained from the
  product $\appbas [\matcol{\sys}{1} | \cdots | \matcol{\sys}{\nu}]$. This is
  computed in $\bigO{\polmatmultime{\nbun,\vsdim/\nbun}}$ operations, according
  to the first item of \cref{lem:compute_residual} with $\order=2\vsdim/\nbun$,
  noting that $(\nu+\vsdim/(\order+1)) / \nbun < 2$ since $\nu<\nbun$.
\end{proof}

As a result, we obtain the second item in \cref{thm:pm_basis}; we only consider
the case $\nbeq\ge\nbun$, hence also $\vsdim\ge\nbun$, since otherwise the
claimed bound follows from that of the first item in the same theorem.  We
first apply \cref{algo:reduce_coldim} to reduce the column dimension in
$\bigO{\appbastime{\nbun,\vsdim/\nbun}}$ operations.  This gives a first basis,
in $\shifts$-ordered weak Popov form, and a new instance
$(\ordersR,\sysR,\shiftsR)$.  Then we compute a second basis, in
$\shiftsR$-ordered weak Popov form for $\modAppCustom{\ordersR}{\sysR}$, via
\cref{algo:dac_ab}; since $\sysR$ has fewer columns than rows by construction,
this uses $\bigO{\appbastime{\nbun,\max(\orders)}}$ operations.

Multiplying both bases costs
$\bigO{\polmatmultime{\nbun,\vsdim/\nbun+\max(\orders)}}$ and yields an
$\shifts$-ordered weak Popov basis of $\modApp$. To obtain the canonical basis,
one would rather deduce the $\shifts$-minimal degree $\minDegs$ from the two
bases (without computing the product), and then either restart the process with
the shift $-\minDegs$ (similarly to \cref{algo:popov_pmbasis}) or call the more
general algorithm in the next section.

\section{Computing approximant bases when the minimal degree is known}
\label{sec:knowndeg}

Let $(\orders,\sys,\shifts)$ be the input of \cref{pbm:pab}, and suppose that
the $\shifts$-minimal degree $\minDegs \in \NN^\rdim$ of $\modApp$ is known.
In this context, \cref{lem:mindeg_shift} suggests that we may focus on
computing a basis $\reduced$ of $\modApp$ which is $-\minDegs$-minimal; then,
the $\shifts$-Popov basis can be easily retrieved via the constant
transformation $\leadingMat[-\minDegs]{\reduced}^{-1} \reduced$.  An obstacle
towards computing $\reduced$ efficiently is the possible unbalancedness of
$\minDegs=\cdeg{\reduced}$, which also impacts the shift $-\minDegs$.  As
sketched in \cref{sec:intro} and in \cref{fig:linearizations} (bottom), we
handle this in \cref{algo:knowndeg_pab} by using the partial linearizations
from \citep{Storjohann06} which allow us to compute $\reduced$ using
essentially one call to \algoname{ReduceColDim} and then one call to
\algoname{PM-Basis}.  We defer the proof of \cref{prop:algo:knowndeg_pab} to
\cref{sec:knowndeg:proof}, and we first present the partial linearizations.

\begin{algobox}[htb]
  \algoInfo
  {Popov basis for known minimal degree}
  {KnownDegAppBasis}
  {algo:knowndeg_pab}

  \dataInfos{Input}{
    \item order $\orders \in \orderSpace$,
    \item matrix $\sys\in\polMatSpace[\nbun][\nbeq]$ with $\cdeg{\sys} <
      \orders$,
    \item shift $\shifts \in \shiftSpace$,
    \item the $\shifts$-minimal degree $\minDegs \in \NN^\nbun$ of $\modApp$.
  }

  \dataInfo{Output}{
    the $\shifts$-Popov basis of $\modApp$.
  }

  \algoSteps{
    \item \comment{Output column linearization $\Rightarrow$ balanced minimal degree} \\
      $\degExp \leftarrow \lceil \sumVec{\orders} / \nbun \rceil$ \\
      $(-\expand{\minDegs},\expandMat,(\quoExp_i)_{1\le i\le\nbun},\expand{\nbun}) \assign 
      \algoname{ColParLin}(-\minDegs,\degExp,\max(-\minDegs))$
      \eolcomment{see \cref{sec:knowndeg:output_parlin}}
    \item \comment{\algoname{ReduceColDim} $\Rightarrow$ fewer columns than rows} \\
      permute $\orders$ into nonincreasing order, and permute the columns of $\sys$ accordingly \\
      $(\ordersR,\sysR,-\minDegsR,\reduced_1) \assign
      \left\{
      \begin{array}{lc}
        \algoname{ReduceColDim}(\orders,\expandMat \sys \bmod \xDiag{\orders},-\expand{\minDegs}) & \algoword{if } \nbeq \ge \expand{\nbun} \\
        (\orders,\expandMat \sys \bmod \xDiag{\orders},-\expand{\minDegs},\idMat[\expand{\nbun}]) & \algoword{if } \nbeq < \expand{\nbun} 
      \end{array}\right.$ \\
      $\nu \assign$ the number of columns of $\sysR$ \eolcomment{$\sysR \in \polMatSpace[\expand{\nbun}][\nu]$ with $\nu < \expand{\nbun}$}
    \item \comment{Overlapping linearization $\Rightarrow$ balanced order and dimensions} \\  
      Construct $\ovlpLinOrd{\degExp}{\ordersR} \in
      \ZZp^{\expand{\nbun}+\expand{\nu}}$ and
      $\ovlpLin{\ordersR,\degExp}{\sysR} \in
      \polMatSpace[(\expand{\nbun}+\expand{\nu})][(\nu+\expand{\nu})]$ as in
      \cref{dfn:ovlplin} \\
      $\shifts[t] \assign (-\minDegsR,-\degExp,\ldots,-\degExp) \in \ZZ_{\le 0}^{\expand{\nbun}+\expand{\nu}}$
    \item \comment{Compute approximant basis for linearized instance} \\
      $\hat{d} \assign \max(\ovlpLinOrd{\degExp}{\ordersR})$;~ $\tuple{\Delta} \assign (\hat{d},\ldots,\hat{d}) - \ovlpLinOrd{\degExp}{\ordersR}$ \\
      $\expand{\appbas} \assign
      \algoname{PM-Basis}(\hat{d},\ovlpLin{\ordersR,\degExp}{\sysR}\xDiag{\tuple{\Delta}},\shifts[t])$
    \item \comment{Deduce basis for original instance and normalize} \\
      $\reduced_2 \assign $ leading principal $\expand{\nbun}\times\expand{\nbun}$ submatrix of $\expand{\appbas}$ \\
      $\reduced \assign$ submatrix of $\reduced_2\reduced_1\expandMat$ formed by its rows
      at indices $\quoExp_1+\cdots+\quoExp_i$ for $1\le i\le\nbun$ \\
      \algoword{Return} $\leadingMat[-\minDegs]{\reduced}^{-1} \reduced$
  }
\end{algobox}

\begin{proposition}
  \label{prop:algo:knowndeg_pab}
  \cref{algo:knowndeg_pab} is correct and uses
  $\bigO{\appbastime{\nbun,\vsdim/\nbun}}$ operations in $\field$, where we
  assume that $\vsdim= \sumVec{\orders} \in \Omega(\nbun)$.
\end{proposition}

\subsection{Output column linearization to balance the output degrees}
\label{sec:knowndeg:output_parlin}

Here, we detail the transformation used in Step~\textbf{1} of
\cref{algo:knowndeg_pab}, for which we closely follow ideas from
\citep[Sec.\,3]{Storjohann06} and \citep[Sec.\,6]{ZhoLab12}. Yet, there are a
few differences due to our goal of handling arbitrary orders \(\orders\) and
computing bases in Popov form.

This transformation corresponds to modifying the input matrix $\sys$ and the
input shift $\shifts$ so that the computed basis $\expand{\appbas}$ is a column
partial linearization of the sought approximant basis $\appbas$, the benefit
being that $\expand{\appbas}$ has uniformly small degrees.  Like all partial
linearizations, this increases the matrix dimensions, $\nbun$ in this case.
This transformation is thus mostly useful when we are able to predict which
columns of $\appbas$ may have large degree: then, we only perform partial
linearization for the columns that require it, and $\nbun$ is typically at most
doubled.  If the prediction was not completely accurate, this will only yield a
subset of the rows of $\appbas$ (see
\cref{subsec:weakly_unbalanced_around_max}).

When the shifted minimal degree is known, it directly gives the column degree
of the sought basis $\appbas$.  Thanks to this information, the original
transformation of \citet[Sec.\,3]{Storjohann06} allows us to reduce to the case
where the output has degree in $\bigO{\vsdim/\nbun}$, and yet to retrieve the
full Popov approximant basis $\appbas$. This has already been stated in
\citep[Lem.\,4.2]{JeNeScVi16} in a more general context; for the purpose of
this section, the latter result would be sufficient.

Still, in \cref{subsec:weakly_unbalanced_around_max} we will deal with
situations where the $\shifts$-minimal degree is not available a priori, but
where assumptions on the shift allow us to guess the locations of large degree
columns. Hence we present, in the next lemma, the details of a more
general transformation similar to that in \citep[Sec.\,6]{ZhoLab12} but for
arbitrary orders \(\orders\); in \cref{lem:parlin_appbas_knowndeg}, we apply it
to the specific case where the minimal degree is known. For more insight into
this transformation, we refer the reader to the latter reference as well as
to \citep[Sec.\,3]{Storjohann06}.

From the next lemma we derive a procedure \algoname{ColParLin} which, on input
$(\shifts,\degExp,\sdiff)$, returns the partial linearization objects
$(\expand{\shifts},\expandMat,(\quoExp_i)_{1\le i\le\nbun},\expand{\nbun})$.
It is used in \cref{algo:knowndeg_pab,algo:zhou_labahn_max}. The parameter
\(\degExp\) is a degree for partial linearization: roughly, columns of degree
more than \(\degExp\) will be split into several columns of degree less than
\(\degExp\), or shift entries that are more than \(\degExp\) will be split into
several shift entries that are less than \(\degExp\). On the other hand, the
parameter \(\sdiff\) has an impact on the degree threshold beyond which we can
recover the approximants for the original instance from those for the partially
linearized instance, as stated in \cref{lem:parlin_appbas}.

\begin{lemma}
  \label{lem:parlin_app}
  Let $\shifts \in \shiftSpace$ and consider two parameters $\degExp
  \in \ZZp$ and $\sdiff\in\ZZ$ for partial linearization.

  Define the shift $\shifts[t]  = (\shift[t]{1},\ldots,\shift[t]{\nbun}) =
  \shifts-\max(\shifts)+\sdiff \in \ZZ_{\le \sdiff}^\nbun$, and for each
  $i\in\{1,\ldots,\nbun\}$ write $-\shift[t]{i} = (\quoExp_i -1) \degExp +
  \remExp_i$ with $\quoExp_i = \lceil -\shift[t]{i}/ \degExp \rceil$ and $1 \le
  \remExp_i \le \degExp$ if $\shift[t]{i} < 0$, and with $\quoExp_i = 1$ and
  $\remExp_i = -\shift[t]{i}$ if $\shift[t]{i} \ge 0$.  Let $\expand{\nbun} =
  \quoExp_1 + \cdots + \quoExp_\nbun$, and define the shift
  $\expand{\shifts}\in \ZZ_{\le 0}^{\expand{\nbun}}$ as
  \begin{equation}
    \label{eqn:expandMinDegs_bis}
    \expand{\shifts} =
    ( \underbrace{-\degExp, \ldots, -\degExp, -\remExp_1}_{\quoExp_1}, \ldots,
    \underbrace{-\degExp, \ldots, -\degExp, -\remExp_\nbun}_{\quoExp_\nbun} ).
  \end{equation}
  We have  $-\degExp \le \expand{\shifts} \le \max(\sdiff,-1)$ and \(\nbun \le
  \expand{\nbun}\), and
  if $\sdiff\ge 0$ then $\expand{\nbun} \le \nbun +
  \sumVec{\!\max(\shifts)-\shifts} / \degExp$.

  Define also the compression-expansion matrix $\expandMat \in
  \polMatSpace[\expand{\nbun}][\nbun]$ as the transpose of
  \begin{equation}
    \label{eqn:expandMat}
    \trsp{\expandMat} = 
    \begin{bmatrix}
      1 & \var^\degExp & \cdots & \var^{(\quoExp_1-1)\degExp} \\
        &              &        &                             & \;\;\; \ddots \\
        &              &        &                             &                & 1 & \var^\degExp & \cdots & \var^{(\quoExp_\nbun-1)\degExp}
  \end{bmatrix} . 
  \end{equation}
  Then, for each $i\in\{1,\ldots,\nbun\}$,
  \begin{itemize}
    \item If a vector $\expand{\app} \in \polMatSpace[1][\expand{\nbun}]$ has
      $\expand{\shifts}$-pivot index $\quoExp_1+\cdots+\quoExp_i$ and
      $\expand{\shifts}$-pivot degree $\expand{\gamma}$, then
      $\expand{\app}\expandMat$ has $\shifts$-pivot index $i$ and
      $\shifts$-pivot degree $\expand{\gamma} + (\quoExp_i-1)\degExp =
      \expand{\gamma} - \shift[t]{i} - \remExp_i$.
    \item If a vector $\app\in\appSpace$ has $\shifts$-pivot index $i$ and
      $\shifts$-pivot degree $\gamma \ge -\shift[t]{i}$, then $\app =
      \expand{\app}\expandMat$ for some $\expand{\app} \in
      \polMatSpace[1][\expand{\nbun}]$ which has $\expand{\shifts}$-pivot index
      $\quoExp_1+\cdots+\quoExp_i$ and $\expand{\shifts}$-pivot degree
      $\gamma+\shift[t]{i} + \remExp_i$.
  \end{itemize}
\end{lemma}
\begin{proof}
  Since \(\quoExp_i\ge 1\) for \(1\le i\le \nbun\), we have \(\nbun \le
  \expand{\nbun}\). Besides, the bound on $\expand{\shifts}$ follows from
  $\min(-\sdiff,1) = \min(-\shifts[t],1) \le \remExp_i \le \degExp$, which
  holds by definition. Now, if $\sdiff\ge 0$, for all $i$ we have $\quoExp_i
  \le 1 + (\sdiff-\shift[t]{i})/\degExp$ since $\sdiff \ge \shift[t]{i}$,
  hence the upper bound on $\expand{\nbun}$.

  Let $\expand{\app}$ be as in the first item, and let $\app =
  \expand{\app}\expandMat$.  We write $\expand{\app} = [\expand{p}_j]_{1\le
  j\le \expand{\nbun}}$,  $\app = [p_j]_{1\le j\le \nbun}$, and
  $\expand{\shifts} = [\expand{s}_j]_{1\le j\le \expand{\nbun}}$.  Our
  assumption on the $\expand{\shifts}$-pivot of $\expand{\app}$ implies that
  $\deg(\expand{p}_j) \le \expand{\gamma} - \remExp_i - \shift[\expand{s}]{j}$
  holds for $1\le j\le\expand{\nbun}$, with equality if
  $j=\quoExp_1+\cdots+\quoExp_i$ (in which case
  \(\shift[\expand{s}]{j}=-\remExp_i\)) and strict inequality if
  $j>\quoExp_1+\cdots+\quoExp_i$.  By construction, $p_j =
  \sum_{1\le k\le \quoExp_j} \expand{p}_{\quoExp_1+\cdots+\quoExp_{j-1}+k}
  \var^{(k-1)\degExp}$ holds for $1\le j\le \nbun$, hence
  \begin{align*}
    \deg(p_j) & \le  \max_{1\le k \le \quoExp_j} \left(\expand{\gamma} - \remExp_i - \expand{s}_{\quoExp_1+\cdots+\quoExp_{j-1}+k} + (k-1)\degExp\right) \\
              & = \expand{\gamma} - \remExp_i + \remExp_j + (\quoExp_j-1)\degExp = \expand{\gamma} - \remExp_i -\shift[t]{j},
  \end{align*}
  with equality if $j=i$ and strict inequality if $j>i$.  Thus, $\app$ has
  $\shifts[t]$-pivot index $i$ and $\shifts[t]$-pivot degree $\expand{\gamma} -
  \remExp_i - \shift[t]{i}$; its $\shifts$-pivot index and degree are the same
  since $\shifts$ and $\shifts[t]$ only differ by a constant.

  Let $\app$ be as in the second item, and write $\app = [p_j]_{1\le j\le
  \nbun}$.  We define $\expand{\app} = [\expand{p}_k]_{1\le k\le\expand{\nbun}}
  \in \polMatSpace[1][\expand{\nbun}]$ as the (unique) vector such that $\app =
  \expand{\app}\expandMat$ and $\deg(\expand{p}_{k}) < \degExp$ if $k \not\in
  \{\quoExp_1+\cdots+\quoExp_j, 1\le j\le \nbun\}$.  Thus, the entry
  $\expand{p}_{\quoExp_1+\cdots+\quoExp_j}$ is the nonnegative degree part of
  $\var^{-(\quoExp_j-1)\degExp} p_j$.  In particular, for $j=i$, since by
  assumption $\deg(p_i) = \gamma \ge \max(-\shift[t]{i},0) \ge
  (\quoExp_i-1)\degExp$, we obtain that
  $\expand{p}_{\quoExp_1+\cdots+\quoExp_i}$ has degree exactly $\deg(p_i) -
  (\quoExp_i-1)\degExp = \gamma + \shift[t]{i} + \remExp_i$, which we denote by
  $\expand{\gamma}$.  Then, our assumption on the $\shifts$-pivot index and
  degree of $\app$, which are the same as its $\shifts[t]$-pivot index and
  degree, implies that
  \begin{align*}
    \deg(\expand{p}_{\quoExp_1+\cdots+\quoExp_j})
    \le \deg(p_j) - (\quoExp_j -1)\degExp
    & \le \gamma + \shift[t]{i} - \shift[t]{j} - (\quoExp_j -1)\degExp \\
    & = \expand{\gamma} - \remExp_i + \remExp_j
    = \expand{\gamma} + \expand{s}_{\quoExp_1+\cdots+\quoExp_i} - \expand{s}_{\quoExp_1+\cdots+\quoExp_j},
  \end{align*}
  where the second inequality is strict if $j>i$.  Furthermore, for $k \not\in
  \{\quoExp_1+\cdots+\quoExp_j, 1\le j\le \nbun\}$, the requirement
  $\deg(\expand{p}_k) < \degExp = -\expand{s}_k$ implies that
  $\deg(\expand{p}_k) + \expand{s}_k < 0 \le \gamma + \shift[t]{i} =
  \expand{\gamma} - \remExp_i =
  \expand{\gamma} + \expand{s}_{\quoExp_1+\cdots+\quoExp_i}$. Thus,
  $\expand{\app}$ has $\expand{\shifts}$-pivot index
  $\quoExp_1+\cdots+\quoExp_i$ and $\expand{\shifts}$-pivot degree
  $\expand{\gamma}$.
\end{proof}

\begin{lemma}
  \label{lem:parlin_appbas}
  Let $\orders \in \orderSpace$, let $\sys \in \sysSpace$ with
  $\cdeg{\sys}<\orders$, and let $\shifts \in \shiftSpace$. Let $\degExp \in
  \ZZp$ and $\sdiff\in\ZZ$. Below, we use notation from the construction in
  \cref{lem:parlin_app} on input $(\shifts,\degExp,t)$, and in particular,
  $(\expand{\shifts},\expandMat,(\quoExp_i)_{1\le i\le\nbun},\expand{\nbun}) =
  \algoname{ColParLin}(\shifts,\degExp,t)$. Then, we have $\modApp =
  \modAppCustom{\orders}{\expandMat \sys \bmod \mods}\expandMat$.

  Let $\minDegs = (\minDeg_1,\ldots,\minDeg_\nbun)\in\NN^\nbun$ be the
  $\shifts$-minimal degree of $\modApp$, let $\expand{\appbas} \in
  \polMatSpace[\expand{\nbun}]$ be an $\expand{\shifts}$-ordered weak Popov
  basis of $\modAppCustom{\orders}{\expandMat \sys \bmod \mods}$, and let
  $i\in\{1,\ldots,\nbun\}$.  If $\minDeg_i \ge -\shift[t]{i}$, the approximant
  $\matrow{\expand{\appbas}}{\quoExp_1+\cdots+\quoExp_i}\expandMat \in \modApp$
  has $\shifts$-pivot index $i$ and $\shifts$-pivot degree $\minDeg_i$.
  Furthermore, if $\matrow{\expand{\appbas}}{\quoExp_1+\cdots+\quoExp_i}$ has
  $\expand{\shifts}$-pivot degree larger than $\remExp_i$ (or, equivalently,
  $\rdeg[\expand{\shifts}]{\matrow{\expand{\appbas}}{\quoExp_1+\cdots+\quoExp_i}}
  > 0$), then $\minDeg_i > -\shift[t]{i}$.
\end{lemma}
\begin{proof}
  The inclusion $\modApp \supseteq \modAppCustom{\orders}{\expandMat \sys \bmod
  \mods}\expandMat$ is obvious: any $\expand{\app} \in
  \modAppCustom{\orders}{\expandMat \sys \bmod \mods}$ satisfies $\expand{\app}
  \expandMat \sys = 0 \bmod \mods$ by definition, hence $\expand{\app}
  \expandMat \in \modApp$.  Conversely, from any $\app \in \modApp$ one can
  construct $\expand{\app}$ such that $\app = \expand{\app}\expandMat$, since
  $\expandMat$ contains $\idMat[\nbun]$ as a submatrix; then
  $\expand{\app}\expandMat \sys = \app\sys = 0 \bmod \mods$, hence
  $\expand{\app} \in \modAppCustom{\orders}{\expandMat \sys \bmod \mods}$ and
  therefore $\app \in \modAppCustom{\orders}{\expandMat \sys \bmod
  \mods}\expandMat$.

  Now, let $\expand{\app} =
  \matrow{\expand{\appbas}}{\quoExp_1+\cdots+\quoExp_i}$. The above paragraph
  shows $\expand{\app}\expandMat \in \modApp$. Since $\expand{\appbas}$ is in
  $\expand{\shifts}$-ordered weak Popov form, $\expand{\app}$ has
  $\expand{\shifts}$-pivot index $\quoExp_1+\cdots+\quoExp_i$; let
  $\expand{\gamma}$ be the $\expand{\shifts}$-pivot degree of $\expand{\app}$.

  From the first item in \cref{lem:parlin_app}, we obtain that
  $\expand{\app}\expandMat$ has $\shifts$-pivot index $i$ and $\shifts$-pivot
  degree $\expand{\gamma} - \shift[t]{i} - \remExp_i$; this must be at least
  $\minDeg_i$ by minimality of $\minDegs$.  On the other hand, the second item
  implies that there exists an approximant in
  $\modAppCustom{\orders}{\expandMat \sys \bmod \mods}$ which has
  $\expand{\shifts}$-pivot index $\quoExp_1+\cdots+\quoExp_i$ and
  $\expand{\shifts}$-pivot degree $\minDeg_i + \shift[t]{i} + \remExp_i$; this
  must be at least $\expand{\gamma}$ by minimality of $\expand{\appbas}$.
  Thus, we have $\expand{\gamma} - \shift[t]{i} - \remExp_i = \minDeg_i$.

  To prove our last claim, we assume that $\expand{\gamma}>\remExp_i$, and we
  show that $\minDeg_i \le -\shift[t]{i}$ leads to a contradiction.  Indeed, in
  this case there exists $\app \in \modApp$ with $\shifts$-pivot index $i$ and
  $\shifts$-pivot degree $\gamma=-\shift[t]{i}$.  Then, the second item in
  \cref{lem:parlin_app} proves the existence of an approximant in
  $\modAppCustom{\orders}{\expandMat \sys \bmod \mods}$ with
  $\expand{\shifts}$-pivot index $\quoExp_1+\cdots+\quoExp_i$ and
  $\expand{\shifts}$-pivot degree $\gamma + \shift[t]{i} + \remExp_i =
  \remExp_i < \expand{\gamma}$, which is impossible by minimality of
  $\expand{\gamma}$.
\end{proof}

We now specialize this result to the case where the \(\shifts\)-minimal degree
of \(\modApp\) is known.

\begin{lemma}
  \label{lem:parlin_appbas_knowndeg}
  Let $\orders \in \orderSpace$, let $\sys \in \sysSpace$ with
  $\cdeg{\sys}<\orders$, let $\shifts \in \shiftSpace$, and let $\minDegs \in
  \NN^\nbun$ be the $\shifts$-minimal degree of $\modApp$.  Choose parameters
  $\degExp \ge \lceil \sumVec{\minDegs}/\nbun \rceil$ and
  $\sdiff=\max(-\minDegs)$. We use notation from \cref{lem:parlin_app} on input
  $(-\minDegs,\degExp,t)$; in particular,
  $(\expand{\shifts},\expandMat,(\quoExp_i)_{1\le i\le\nbun},\expand{\nbun}) =
  \algoname{ColParLin}(-\minDegs,\degExp,t)$.

  Then, we have $\nbun \le \expand{\nbun} \le 2\nbun$, $-\degExp \le
  \expand{\shifts} \le 0$, and $\expand{\shifts} = -\expand{\minDegs}$ where
  $\expand{\minDegs}$ is the $\expand{\shifts}$-minimal degree of
  $\modAppCustom{\orders}{\expandMat \sys \bmod \mods}$.  Let $\expand{\appbas}
  \in \polMatSpace[\expand{\nbun}]$ be an $\expand{\shifts}$-ordered weak Popov
  basis of $\modAppCustom{\orders}{\expandMat \sys \bmod \mods}$ and $\reduced
  \in \appbasSpace$ be the submatrix of $\expand{\appbas}\expandMat$ formed by
  its rows at indices $\{\quoExp_1+\cdots+\quoExp_i, 1\le i\le\nbun\}$.  Then,
  $\reduced$ is a $-\minDegs$-ordered weak Popov basis of $\modApp$ and
  therefore, as a consequence of \cref{lem:mindeg_shift},
  $\leadingMat[-\minDegs]{\reduced}^{-1} \reduced$ is the $\shifts$-Popov basis
  of $\modApp$.
\end{lemma}
\begin{proof}
  The lower bound on \(\expand{\nbun}\) follows directly from
  \cref{lem:parlin_app}, and so do the bounds on $\expand{\shifts}$ since
  \(\max(\sdiff,-1)\le 0\). By choice of \(\sdiff\), we have $\shifts[t] =
  -\minDegs$, whose entries are nonpositive. Thus, for each \(i \in \{
  1,\ldots,\nbun \}\), we have \(\quoExp_i = 1\) if \(\minDeg_i=t_i=0\) and
  \(\quoExp_i = \lceil \minDeg_i / \degExp \rceil\) otherwise; in both cases,
  \(\quoExp_i \le 1 +  \minDeg_i / \degExp \). Summing these inequalities, we
  obtain \(\expand{\nbun} \le \nbun + \sumVec{\minDegs} / \degExp \le 2\nbun\)
  by choice of \(\degExp\).  Furthermore, since $-\shifts[t] \le \minDegs$
  entry-wise, \cref{lem:parlin_appbas} shows that $\reduced$ is a
  $-\minDegs$-ordered weak Popov basis of $\modApp$.
  
  Our claim on $\expand{\minDegs}$ can be showed using the minimality of
  $\minDegs$ and the arguments used for proving the two items of
  \cref{lem:parlin_app}; details can be found in the proof of
  \citep[Lem.\,4.2]{JeNeScVi16} which contains an explicit description of the
  $\expand{\shifts}$-Popov basis of $\modAppCustom{\orders}{\expandMat \sys
  \bmod \mods}$.
\end{proof}

\subsection{Overlapping linearization to balance orders and dimensions}
\label{sec:knowndeg:overlapping_parlin}

Now, we study Step~\textbf{3} of \cref{algo:knowndeg_pab}: assuming that the
shifted minimal degree is known, balanced (Step~\textbf{1}), and that $\nbeq <
\nbun$ (Step~\textbf{2}), we reduce to an instance which is solved efficiently
by \algoname{PM-Basis}. Namely, we use the \emph{overlapping linearization} of
\citet[Sec.\,2]{Storjohann06} to further transform the instance of
\cref{pbm:pab} into one with a balanced order and $\nbeq\in\Theta(\nbun)$. In
the latter reference, as well as in \citep[Sec.\,3]{ZhoLab12}, this
linearization has been considered in the case of a uniform order
$\orders=(\order,\ldots,\order)$. Here, we extend the construction to arbitrary
orders, and we show how it can be used in our specific situation where the
$\shifts$-minimal degree is known.

We first give an overview of the construction and of its properties. Let
$\orders\in\orderSpace$ and $\sys \in \sysSpace$ with $\cdeg{\sys} < \orders$,
and choose a positive integer $\degExp$. Then, we build an order
$\ovlpLinOrd{\degExp}{\orders} \in \ZZp^{\nbeq+\expand{\nbeq}}$ and a matrix
$\ovlpLin{\orders,\degExp}{\sys} \in
\polMatSpace[(\nbun+\expand{\nbeq})][(\nbeq+\expand{\nbeq})]$ such that
\begin{itemize}
  \item the largest entry of the order $\ovlpLinOrd{\degExp}{\orders}$ is
    at most $2\degExp$,
  \item the increase in dimension is $\expand{\nbeq} < \vsdim/\degExp$,
    where $\vsdim = \sumVec{\orders}$,
  \item approximants $\app \in \modApp$ of degree at most $\degExp$ correspond
    to approximants $[\app \;\; \row{q}] \in
    \modAppOvlpLin$
    for some $\row{q}$ of degree less than $\rdeg{\app}$.
\end{itemize}

The last item, detailed in \cref{lem:ovlplin_app}, gives a link between the
original approximation instance and the one obtained after linearization. This
implies that a minimal basis for the original instance can be retrieved from a
minimal basis for the transformed instance, assuming we choose $\degExp$ as an
upper bound on the degree of the former basis; this approach is detailed in
\cref{lem:ovlplin_appbas}.

The first two items are direct consequences of the construction, given in
\cref{dfn:ovlplin}. They specify the dimensions of the transformed instance. In
the context of \cref{algo:knowndeg_pab}, the output column linearization of
\cref{sec:knowndeg:output_parlin} has already been applied, which ensures that
we are seeking a basis of degree about \(\vsdim/\nbun\), and hence that
$\degExp$ can be chosen to be about $\vsdim/\nbun$. Then, the new order is
balanced and the dimension increase is only about $\nbun$: the transformed
instance can be solved efficiently using a single call of \algoname{PM-Basis}.
More details about Step~\textbf{3} of \cref{algo:knowndeg_pab} can be found in
\cref{sec:knowndeg:proof}.

Note that, in general, the $\shifts$-Popov approximant basis may have degree up
to $\vsdim$, in which case one would choose $\degExp \ge \vsdim$ in the above
approach: this would not lead to any improvement since the entries of the order
have not been decreased by the linearization.
Still, in some contexts it is known that the sought basis has rows of small
degree: using a small parameter $\degExp$ will not yield the whole basis but
does give the small degree part of the basis (see \cref{lem:ovlplin_app}). This
was one of the key properties mentioned in the original design of this
linearization in \citep{Storjohann06}, and used in \citep{ZhoLab12} to handle
shifts that are weakly unbalanced around their minimum value (see also
\cref{subsec:weakly_unbalanced_around_min}).

Let us now present the construction of $\ovlpLinOrd{\degExp}{\orders}$ and
$\ovlpLin{\orders,\degExp}{\sys}$.

\begin{definition}
  \label{dfn:ovlplin}
  Let $\orders=(\order_1,\ldots,\order_\nbeq) \in \orderSpace$, let $\sys \in
  \sysSpace$ with $\cdeg{\sys} < \orders$, and let $\degExp\in\ZZp$. Then, for
  $1\le i\le \nbeq$, let $\order_i = \quoExp_i \degExp + \remExp_i$ with
  $\quoExp_i = \left\lceil \frac{\order_i}{\degExp} - 1 \right\rceil$ and $1\le
  \remExp_i\le \degExp$.
  
  Let also $\expand{\nbeq} = \max(\quoExp_1-1,0) + \cdots +
  \max(\quoExp_\nbeq-1,0)$, and define
  \[
    \ovlpLinOrd{\degExp}{\orders} =
    (\expand{\order}_1,\ldots,\expand{\order}_\nbeq) \in \ZZp^{\nbeq+\expand{\nbeq}},
  \]
  where $\expand{\order}_i = (2\degExp,\ldots,2\degExp,\degExp+\remExp_i) \in
  \ZZp^{\quoExp_i}$ if $\quoExp_i > 1$ and $\expand{\order}_i = \order_i$
  otherwise.
  Considering the $i$th column of $\sys$, we write its
  $\var^\degExp$-adic representation as
  \begin{align*}
    \matcol{\sys}{i} & = \matcol{\sys}{i}^{(0)} + \matcol{\sys}{i}^{(1)}
    \var^{\degExp} + \cdots + \matcol{\sys}{i}^{(\quoExp_i)} \var^{\quoExp_i\degExp}
     \\
    & \quad\text{where }\, \cdeg{[\matcol{\sys}{i}^{(0)} \;\; \matcol{\sys}{i}^{(1)} \;\; \cdots \;\; \matcol{\sys}{i}^{(\quoExp_i)}]} < (\degExp,\ldots,\degExp,\remExp_i) .
  \end{align*}
  If $\quoExp_i> 1$, we define
  \[
    \matcol{\expand{\sys}}{i} = 
    \begin{bmatrix}
      \matcol{\sys}{i}^{(0)} + \matcol{\sys}{i}^{(1)}\var^{\degExp} \,\;&\,\; \matcol{\sys}{i}^{(1)} + \matcol{\sys}{i}^{(2)}\var^{\degExp} \,\;&\,\; \cdots \,\;&\,\; \matcol{\sys}{i}^{(\quoExp_i-1)} + \matcol{\sys}{i}^{(\quoExp_i)}\var^{\degExp} \\
    \end{bmatrix}
    \in \polMatSpace[\nbun][\quoExp_i]
  \]
  and $\emat_i = [\matz \;\; \idMat[\quoExp_i-1]] \in
  \polMatSpace[(\quoExp_i-1)][\quoExp_i]$, and otherwise we let
  $\matcol{\expand{\sys}}{i} = \matcol{\sys}{i}$ and $\emat_i \in
  \polMatSpace[0][1]$. Then,
  \[
    \ovlpLin{\orders,\degExp}{\sys} =
    \begin{bmatrix}
      \matcol{\expand{\sys}}{1} & \matcol{\expand{\sys}}{2} & \cdots & \matcol{\expand{\sys}}{\nbeq} \\
      \emat_1 &  & &  \\
       & \emat_2 & &  \\
       &  & \ddots &  \\
       &  &  & \emat_\nbeq
    \end{bmatrix}
    \in \polMatSpace[(\nbun+\expand{\nbeq})][(\nbeq+\expand{\nbeq})]
  \]
  is called the \emph{overlapping linearization} of $\sys$ with respect to
  $\orders$ and $\degExp$.
\end{definition}

The next lemma gives a correspondence between the approximants of degree
bounded by $\degExp$ in $\modApp$ and in
$\modAppOvlpLin$.
It uses notation from \cref{dfn:ovlplin}.

\begin{lemma}
  \label{lem:ovlplin_app}
  Let $\orders \in \orderSpace$, let $\sys \in \sysSpace$ with $\cdeg{\sys} <
  \orders$, and let $\degExp\in\ZZp$.  Then,
  \begin{itemize}
    \item If $\app\in\appSpace$ is in $\modApp$, then there exists a
      unique $\row{q} \in \polMatSpace[1][\expand{\nbeq}]$ such that $[\app
      \;\; \row{q}] \in
      \modAppOvlpLin$,
      $\rdeg{\row{q}} < \rdeg{\app}$, and $\cdeg{\row{q}} <
      \ovlpLinOrd{\degExp}{\orders} \trsp{\emat}$ where $\emat =
      \diag{\emat_1,\ldots,\emat_\nbeq}$.  Explicitly, it is defined as
      $\row{q} = - \app [\matcol{\expand{\sys}}{1} \;\; \cdots \;\;
      \matcol{\expand{\sys}}{\nbeq}] \trsp{\emat} \bmod
      \xDiag{\ovlpLinOrd{\degExp}{\orders}\trsp{\emat}}$.
    \item If $[\app \;\; \row{q}] \in \polMatSpace[1][(\nbun+\expand{\nbeq})]$
      is in
      $\modAppOvlpLin$
      and such that $\rdeg{\row{q}} < \degExp$ and $\rdeg{\app} \le \degExp$,
      then $\app \in \modApp$; in particular, $\rdeg{\row{q}} < \rdeg{\app}$.
  \end{itemize}
\end{lemma}
\begin{proof}
  Concerning the first item, we first consider $i \in \{1,\ldots,\nbeq\}$ such
  that $\quoExp_i \in \{0,1\}$. Then, we have $\matcol{\expand{\sys}}{i} =
  \matcol{\sys}{i}$, $\expand{\order}_i = \order_i$, and $\emat_i \in
  \polMatSpace[0][1]$.  Defining $\row{q}_i$ as an empty matrix in
  $\polMatSpace[1][0]$, the identity $\app \matcol{\sys}{i} = 0 \bmod
  \var^{\order_i}$ can be rewritten as $\app \matcol{\expand{\sys}}{i} +
  \row{q}_i \emat_i = 0 \bmod \xDiag{\expand{\order}_i}$.


  Now, for $i$ such that $\quoExp_i > 1$, we define $\row{q}_i = [q_{1,i} \;\;
  \cdots \;\; q_{\quoExp_i-1,i}] \in \polMatSpace[1][(\quoExp_i-1)]$ as
  \begin{equation}
    \label{eqn:quotient_ab_parlin}
    \left\{
    \begin{array}{l}
      q_{j,i} = \var^{-j\degExp} \row{p} (\matcol{\sys}{i}^{(0)} + \cdots +
      \matcol{\sys}{i}^{(j-1)} \var^{(j-1)\degExp}) \bmod X^{2\degExp},
      \;\,\text{for}\;\, 1\le j<\quoExp_i-1,\\
      q_{\quoExp_i-1,i} = \var^{-(\quoExp_i-1)\degExp} \row{p}
      (\matcol{\sys}{i}^{(0)} + \cdots + \matcol{\sys}{i}^{(\quoExp_i-2)}
      \var^{(\quoExp_i-2)\degExp}) \bmod X^{\degExp+\remExp_i}.
    \end{array}
    \right.
  \end{equation}
 These are polynomials since $\app \matcol{\sys}{i} = 0 \bmod \var^{\order_i}$,
 and $\rdeg{\row{q}_i} < \rdeg{\app}$ holds since by construction
 $\cdeg{\matcol{\sys}{i}^{(k)}} < \degExp$ for all $k$.  For $j<\quoExp_i-1$,
 $\app (\matcol{\sys}{i}^{(0)} + \cdots + \matcol{\sys}{i}^{(j+1)}
 \var^{(j+1)\degExp}) = 0 \bmod \var^{(j+2)\degExp}$ becomes $q_{j,i}
 \var^{j\degExp} + \app(\matcol{\sys}{i}^{(j)} \var^{j\degExp} +
 \matcol{\sys}{i}^{(j+1)} \var^{(j+1)\degExp}) = 0 \bmod \var^{(j+2)\degExp}$,
 hence $\app(\matcol{\sys}{i}^{(j)} + \matcol{\sys}{i}^{(j+1)}\var^\degExp) +
 q_{j,i} = 0 \bmod \var^{2\degExp}$.  Similarly, we obtain
 $\app(\matcol{\sys}{i}^{(\quoExp_i-1)} +
 \matcol{\sys}{i}^{(\quoExp_i)}\var^\degExp) + q_{\quoExp_i-1,i} = 0 \bmod
 \var^{\degExp+\remExp_i}$.  In short, we have
  \begin{equation}
    \label{eqn:local_approx_ab_parlin}
    \begin{bmatrix} \app \: & \: \row{q}_i \end{bmatrix} \begin{bmatrix}
      \matcol{\expand{\sys}}{i} \\ \emat_i \end{bmatrix} = 0 \bmod
    \xDiag{\expand{\order_i}},
    \;\;\; \text{where } \expand{\order_i} =
    (2\degExp,\ldots,2\degExp,\degExp+\remExp_i).
  \end{equation}
  Thus, by construction of $\ovlpLin{\orders,\degExp}{\sys}$ and
  $\ovlpLinOrd{\degExp}{\orders}$, we have $[\app \;\, \row{q}_1 \;\, \cdots
  \;\, \row{q}_\nbeq] \in
  \modAppOvlpLin$.
  Besides, we have proved the degree bound for $[\row{q}_1 \;\; \cdots \;\;
  \row{q}_\nbeq]$; the explicit formula follows from
  \cref{eqn:local_approx_ab_parlin}, since the latter gives $\row{q}_i =
  \row{q}_i \emat_i \trsp{\emat}_i = - \app \matcol{\expand{\sys}}{i}
  \trsp{\emat}_i \bmod \xDiag{\expand{\order}_i \trsp{\emat}_i}$.

  Now, we prove the second item.  We write $\row{q} = [\row{q}_1 \;\; \cdots
  \;\; \row{q}_\nbeq]$ with $\row{q}_i \in \polMatSpace[1][0]$ if
  $\quoExp_i\in\{0,1\}$ and $\row{q}_i = [q_{1,i},\ldots,q_{\quoExp_i-1,i}] \in
  \polMatSpace[1][(\quoExp_i-1)]$ if $\quoExp_i>1$.  Let $i \in
  \{1,\ldots,\nbeq\}$. If $\quoExp_i \in \{0,1\}$, then we have $\row{p}
  \matcol{\sys}{i} = 0 \bmod \var^{\order_i}$.  If $\quoExp_i > 1$, then the
  identity in \cref{eqn:local_approx_ab_parlin} holds and yields
  \begin{align*}
    & \app (\matcol{\sys}{i}^{(0)} + \matcol{\sys}{i}^{(1)} \var^\degExp) = 0 \bmod
      \var^{2\degExp}, \\
    & \app(\matcol{\sys}{i}^{(j)} + \matcol{\sys}{i}^{(j+1)} \var^\degExp) =
      -q_{j,i} \bmod \var^{2\degExp} \quad \text{for }
      1\le j\le \quoExp_i-2, \\
    & \app(\matcol{\sys}{i}^{(\quoExp_i-1)} + \matcol{\sys}{i}^{(\quoExp_i)}
      \var^\degExp) = -q_{\quoExp_i-1,i} \bmod \var^{\degExp+\remExp_i},
  \end{align*}
  where $\row{q}_i = [q_{1,i},\ldots,q_{\quoExp_i-1,i}]$. The first identity
  and the second one for $j=1$ imply that
  \[
    \app(\matcol{\sys}{i}^{(0)} + \matcol{\sys}{i}^{(1)} \var^\degExp +
    \matcol{\sys}{i}^{(2)} \var^{2\degExp}) = \app \matcol{\sys}{i}^{(0)} -
    q_{1,i} \var^\degExp = 0 \bmod X^{2\degExp};
  \]
  using the bounds $\rdeg{\row{q}} < \degExp$ and $\rdeg{\app} \le \degExp$ we
  obtain $q_{1,i} = \var^{-\degExp} \app \matcol{\sys}{i}^{(0)}$ and $\app
  \matcol{\sys}{i} = 0 \bmod \var^{3\degExp}$. Then the same arguments with the
  above identity for $j=2$, we obtain $q_{2,i} = \var^{-2\degExp} \app
  (\matcol{\sys}{i}^{(0)} + \matcol{\sys}{i}^{(1)}\var^\degExp)$ and $\app
  \matcol{\sys}{i} = 0 \bmod \var^{4\degExp}$. Continuing this process, we
  eventually obtain $\app \matcol{\sys}{i} = 0 \bmod \var^{\order_i}$.
\end{proof}

We now show that the $\shifts$-Popov basis $\appbas$ of $\modApp$
can be deduced from one for the transformed problem, as long as $\degExp$ is
chosen to be at least $\deg(\appbas)$.

\begin{lemma}
  \label{lem:ovlplin_appbas}
  Let $\orders \in \orderSpace$, let $\sys \in \sysSpace$ with $\cdeg{\sys} <
  \orders$, let $\shifts\in\shiftSpace$, let $\minDegs\in\NN^\nbun$ be the
  $\shifts$-minimal degree of $\modApp$, and let $\degExp\in\ZZp$ be such that
  $\degExp\ge \max(\minDegs)$. Let  $\expand{\appbas}$ be a
  $(-\minDegs,-\degExp,\ldots,-\degExp)$-ordered weak Popov basis
  of
  $\modAppOvlpLin$.
  Then, the leading principal submatrix $\reduced \in \polMatSpace[\nbun]$ of
  $\expand{\appbas}$ is a $-\minDegs$-ordered weak Popov basis of
  $\modApp$ and therefore, as a consequence of \cref{lem:mindeg_shift},
  $\leadingMat[-\minDegs]{\reduced}^{-1} \reduced$ is the $\shifts$-Popov basis
  of $\modApp$.
\end{lemma}
\begin{proof}
  In this proof, we use the notation $\shifts[t] =
  (-\minDegs,-\degExp,\ldots,-\degExp)\in \shiftSpace[\nbun+\expand{\nbeq}]$.
  
  Let $\appbas \in \polMatSpace[\nbun]$ be a $-\minDegs$-ordered weak Popov
  basis of $\modApp$. Then, we have $\rdeg[-\minDegs]{\appbas} = \unishift$
  according to \cref{lem:mindeg_shift}, hence in particular all the rows of
  $\appbas$ have degree at most $\degExp$.  The first item of
  \cref{lem:ovlplin_app} implies that there exists a matrix $\mat{Q} \in
  \polMatSpace[\nbun][\expand{\nbeq}]$ such that all the rows of $[\appbas \;\;
  \mat{Q}]$ are in $\modAppOvlpLin$ and $\rdeg{\mat{Q}}<\rdeg{\appbas}$. Then,
  by choice of $\shifts[t]$, we have $\leadingMat[{\shifts[t]}]{[\appbas \;\;
  \mat{Q}]} = [\leadingMat[-\minDegs]{\appbas} \;\; \matz]$, with
  $\leadingMat[-\minDegs]{\mat{P}}$ lower triangular by assumption. Thus
  $[\appbas \;\; \mat{Q}]$ is in $\shifts[t]$-ordered weak Popov form with all
  $\shifts[t]$-pivots in $\appbas$.

  Now, let us write
  \[
    \expand{\appbas} =
    \begin{bmatrix}
      \reduced & \expand{\appbas}_{12} \\
      \expand{\appbas}_{21} & \expand{\appbas}_{22}
    \end{bmatrix}
    \text{ with } \reduced \in \polMatSpace[\nbun] \text{ and }
    \expand{\appbas}_{22} \in \polMatSpace[\expand{\nbeq}].
  \]
  Since the $\shifts[t]$-pivots of $[\reduced \;\; \expand{\appbas}_{12}]$ are
  on the diagonal of $\reduced$, by minimality of $\expand{\appbas}$ we obtain
  $\rdeg[-\minDegs]{\reduced} = \rdeg[{\shifts[t]}]{[\reduced \;\;
  \expand{\appbas}_{12}]} \le \rdeg[{\shifts[t]}]{[\appbas \;\; \mat{Q}]} =
  \unishift$. Thus $\deg(\reduced)\le\max(\minDegs)\le\degExp$ and
  $\deg(\expand{\appbas}_{12}) < \degExp$, and the second item of
  \cref{lem:ovlplin_app} applied to the rows of $[\reduced \;\;
  \expand{\appbas}_{12}]$ shows that each row of $\reduced$ is in $\modApp$.
  Since $\reduced$ is in $-\minDegs$-ordered weak Popov form, this gives
  $\rdeg[-\minDegs]{\reduced} \ge \rdeg[-\minDegs]{\appbas} = \unishift$ by
  minimality of $\appbas$. Thus, we have $\rdeg[-\minDegs]{\reduced}=\unishift$
  and $\reduced$ is a $-\minDegs$-ordered weak Popov basis of $\modApp$.
\end{proof}

\subsection{Proof of \texorpdfstring{\cref{prop:algo:knowndeg_pab}}{Proposition}}
\label{sec:knowndeg:proof}

We first give some properties of the manipulated quantities to verify that the
assumptions of the lemmas and corollary referred to in the next paragraph are
indeed satisfied.  In what follows, we let $\expand{\sys} = \expandMat \sys
\bmod \xDiag{\orders}$.  First, we have $\sumVec{\minDegs} \le \vsdim =
\sumVec{\orders}$ by \cref{lem:degdet_appbasis}, hence $\degExp =
\lceil\vsdim/\nbun\rceil \ge \lceil\sumVec{\minDegs}/\nbun\rceil$ and thus we
can apply \cref{lem:parlin_appbas_knowndeg}; it ensures that the tuple
$\expand{\minDegs}$ computed at Step~\textbf{1} is the
$-\expand{\minDegs}$-minimal degree of $\modAppCustom{\orders}{\expand{\sys}}$
and satisfies $-\expand{\minDegs} \ge -\degExp$, that is,
$\max(\expand{\minDegs}) \le \degExp$.  Besides, since $\reduced_1$ is in
$-\expand{\minDegs}$-ordered weak Popov form, it has $-\expand{\minDegs}$-pivot
degree $\rdeg[-\expand{\minDegs}]{\reduced_1} + \expand{\minDegs} = -\minDegsR
+\expand{\minDegs}$, by definition of $\minDegsR$ at Step~\textbf{2}.  Thus, by
the fourth item of \cref{lem:recursiveness} and by
\cref{prop:algo:reduce_coldim}, $\minDegsR$ is the $-\minDegsR$-minimal degree
of $\modAppCustom{\ordersR}{\sysR}$.  This further implies $\minDegsR \le
\expand{\minDegs}$, and therefore $\max(\minDegsR) \le \max(\expand{\minDegs})
\le \degExp$.

By \cref{rmk:pmbasis_unbalancedorders}, Step~\textbf{4} computes a
$\shifts[t]$-ordered weak Popov basis $\expand{\appbas}$ of
$\modAppCustom{\ovlpLinOrd{\degExp}{\ordersR}}{\ovlpLin{\ordersR,\degExp}{\sysR}}$.
Then, \cref{lem:ovlplin_appbas} applied to
$(\ordersR,\sysR,-\minDegsR,\minDegsR,\degExp)$ shows that $\reduced_2$ is a
$-\minDegsR$-ordered weak Popov basis of $\modAppCustom{\ordersR}{\sysR}$.
Then, \cref{prop:algo:reduce_coldim} implies that $\reduced_2\reduced_1$ is a
basis of $\modAppCustom{\orders}{\expand{\sys}}$ and the third item of
\cref{lem:recursiveness} shows that it is in $-\expand{\minDegs}$-ordered weak
Popov form, since $-\minDegsR =\rdeg[-\expand{\minDegs}]{\reduced_1}$.  It then
follows from \cref{lem:parlin_appbas_knowndeg} applied to
$(\orders,\sys,\shifts,\minDegs)$ that $\leadingMat[-\minDegs]{\reduced}^{-1}
\reduced$ is the $\shifts$-Popov basis of $\modApp$.

Concerning the cost, Steps~\textbf{1} and~\textbf{3} use no field operation.
At Step~\textbf{2}, obtaining the matrix $\expandMat\sys \bmod \xDiag{\orders}$
involves no field operation given the form of $\expandMat$, but only at most
$\expand{\nbun}\vsdim$ read/write of field elements, where $\expand{\nbun} \le
2\nbun$ according to \cref{lem:parlin_appbas_knowndeg}.  Then
\Cref{prop:algo:reduce_coldim} indicates that Step~\textbf{2} uses
$\bigO{\appbastime{\nbun,\vsdim/\nbun}}$ operations, which is within the
announced bound.

From $\expand{\nu} \le \sumVec{\ordersR} / \degExp$ by \cref{dfn:ovlplin} and
$\sumVec{\ordersR} \le \vsdim$ by \cref{prop:algo:reduce_coldim}, we get
$\expand{\nu} \le \vsdim / \lceil \vsdim/\nbun \rceil \le \nbun$.  Thus,
$\ovlpLin{\ordersR,\degExp}{\expand{\sys}}$ has $\expand{\nbun} + \expand{\nu}
\le 3\nbun$ rows and $\nu + \expand{\nu} < \expand{\nbun} + \expand{\nu} \le
3\nbun$ columns.  Besides, by construction of $\ovlpLinOrd{\degExp}{\ordersR}$
we have $\hat{d} \le 2\degExp = 2 \lceil \vsdim/\nbun \rceil$, hence $\hat{d}
\in \bigO{\vsdim/\nbun}$.  Note that we can discard the ceiling since we have
assumed $\vsdim \in \Omega(\nbun)$.  Then, according to
\cref{prop:algo:dac_ab}, the call to \algoname{PM-Basis} at Step~\textbf{4}
uses $\bigO{\appbastime{\expand{\nbun}+\expand{\nu},\hat{d}}} \subseteq
\bigO{\appbastime{\nbun,\vsdim/\nbun}}$ operations.

Now, $\deg(\reduced_1) \le 2 \vsdim/\nbun$ by \cref{prop:algo:reduce_coldim}.
We have seen that $\reduced_2$ has $-\tuple{\hat\minDegs}$-pivot degree
$\tuple{\hat\minDegs}$, which implies $\cdeg{\reduced_2} =
\tuple{\hat\minDegs}$ by \cref{lem:mindeg_shift}.  Thus $\deg(\reduced_2) =
\max(\tuple{\hat\minDegs}) \le \lceil\vsdim/\nbun\rceil$, which gives
$\deg(\reduced_2)\in\bigO{\vsdim/\nbun}$ (remark that here only the case
$\vsdim \ge \nbun$ is relevant, since otherwise $\nbeq \le \vsdim < \nbun \le
\expand{\nbun}$ and then $\reduced_1=\idMat[\expand{\nbun}]$).  Thus, computing
$\reduced_2\reduced_1$ uses $\bigO{\polmatmultime{\nbun,\vsdim/\nbun}}$
operations.  Then, given the shape of $\expandMat$, obtaining $\reduced$ from
$\reduced_2\reduced_1$ uses $\bigO{\nbun \expand{\nbun} \vsdim/\nbun} \subseteq
\bigO{\nbun \vsdim}$ additions in $\field$.

Finally, the computation of $\leadingMat[-\minDegs]{\reduced}^{-1}$ at
Step~\textbf{5} uses $\bigO{\nbun^\expmatmul}$ operations.  Since
$\cdeg{\reduced} = \minDegs$ by \cref{lem:mindeg_shift} and $\sumVec{\minDegs}
\le \vsdim$ by \cref{lem:degdet_appbasis}, applying the first item of
\cref{lem:compute_residual} with $\order=0$ shows that the product
$\leadingMat[-\minDegs]{\reduced}^{-1} \reduced$ costs $\bigO{\lceil
(\nbun+\vsdim)/\nbun\rceil \nbun^{\expmatmul}}$ operations.  Since $\vsdim\in
\Omega(\nbun)$ this bound is in $\bigO{\nbun^{\expmatmul-1}\vsdim}$, which
itself is in $\bigO{\appbastime{\nbun,\vsdim/\nbun}}$.

\section{Computing approximant bases for arbitrary shifts}
\label{sec:fastpab}

We now describe our algorithm for solving the general case of \cref{pbm:pab}
(\cref{algo:fast_pab}), and we prove that it is correct and admits the cost
bound announced in \cref{thm:fast_pab}.

\begin{algobox}
  \algoInfo
  {Shifted Popov approximant basis}
  {PopovAppBasis}
  {algo:fast_pab}

  \dataInfos{Input}{
    \item order $\orders = (\order_1,\ldots,\order_\nbeq) \in \orderSpace$,
    \item matrix $\sys\in\polMatSpace[\nbun][\nbeq]$ with $\cdeg{\sys} <
      \orders$,
    \item shift $\shifts \in \shiftSpace$.
  }

  \dataInfo{Output}{the $\shifts$-Popov basis of $\modApp$.}

  \algoSteps{
    \item \algoword{If} $\vsdim=\order_1+\cdots+\order_\nbeq \le \nbun$ : \eolcomment{Base case}
      \begin{enumerate}[{\bf a.}]
        \item \algoword{For} $i$ \algoword{from} $1$ \algoword{to} $\nbeq$:
        \begin{enumerate}[(i)]
        \item $\mat{E}_i \assign \begin{bmatrix} \col{f}_i^{(0)} & \col{f}_i^{(1)} & \cdots & \col{f}_i^{(\order_i-1)} \end{bmatrix} \in \matSpace[\nbun][\order_i]$ where $\matcol{\sys}{i} = \sum_{0\le k<\order_i} \col{f}_i^{(k)} \var^k$
          \item $\mulshift_i \assign 
                \left[\begin{smallmatrix}
                  0 & 1                   \\
                    & \sddots & \sddots   \\
                    &        &   0    & 1 \\
                    &        &        & 0
                \end{smallmatrix}\right] \in \matSpace[\order_i]$
        \end{enumerate}
        \item $\mat{E} \assign \begin{bmatrix} \mat{E}_1 & \cdots & \mat{E}_\nbeq \end{bmatrix} \in \matSpace[\nbun][\vsdim]$;~
              $\mulshift \assign \diag{\mulshift_1,\ldots,\mulshift_\nbeq} \in \matSpace[\vsdim]$
        \item \algoword{Return} $\algoname{LinearizationInterpolationBasis}(\mat{E},\mulshift,\shifts,\max(\orders))$ \\
          \strut\eolcomment{\citep[Algo.\,9]{JeNeScVi17}}
      \end{enumerate}
    \item \algoword{Else if} $\nbeq\ge\nbun$: \eolcomment{Entered at most once at initial call}
      \begin{enumerate}[{\bf a.}]
        \item permute $\orders$ into nonincreasing order, and the columns of
          $\sys$ accordingly
        \item
          $(\ordersR,\sysR,\shiftsR,\appbas_1)
          \assign \algoname{ReduceColDim}(\orders,\sys,\shifts)$
        \item $\popov_2 \assign \algoname{PopovAppBasis}(\ordersR,\sysR,\shiftsR)$
        \item $\minDegs_1 \assign$ diagonal degrees of $\appbas_1$;~
        $\minDegs_2 \assign$ diagonal degrees of $\popov_2$
        \item \algoword{Return} $\algoname{KnownDegAppBasis}(\orders,\sys,\shifts,\minDegs_1+\minDegs_2)$
      \end{enumerate}
    \item \algoword{Else:} \eolcomment{Divide and conquer}
      \begin{enumerate}[{\bf a.}]
        \item $1\le i_0 \le \nbeq$ and $1\le d \le \order_{i_0}$ such that
          $\order_1 + \cdots + \order_{i_0-1}+d = \lfloor \vsdim/2 \rfloor$
        \item $\col{f}_{i_0,1} \assign \matcol{\sys}{i_0} \bmod \var^d$;~
              $\col{f}_{i_0,2} \assign \var^{-d} (\matcol{\sys}{i_0} - \col{f}_{i_0,1})$
        \item $\orders_1 \assign (\order_1,\ldots,\order_{i_0-1},d)$;~
              $\sys_1 \assign [ \matcol{\sys}{1} | \cdots | \matcol{\sys}{i_0-1} | \col{f}_{i_0,1} ]$
        \item $\orders_2 \assign (\order_{i_0}-d,\order_{i_0+1},\ldots,\order_{\nbeq})$;~
              $\sys_2 \assign [ \col{f}_{i_0,2} | \matcol{\sys}{i_0+1} | \cdots | \matcol{\sys}{\nbeq} ]$
        \item $\popov_1 \assign \algoname{PopovAppBasis}(\orders_1,\sys_1,\shifts)$;
              $\minDegs_1 \assign$ diagonal degrees of $\popov_1$
        \item $\res \assign \popov_1 \sys_2 \bmod \xDiag{\orders_2}$ \eolcomment{using partial linearization}
        \item $\popov_2 \assign \algoname{PopovAppBasis}(\orders_2,\mat{G},\shifts+\minDegs_1)$;
              $\minDegs_2 \assign$ diagonal degrees of $\popov_2$
        \item \algoword{Return} $\algoname{KnownDegAppBasis}(\orders,\sys,\shifts,\minDegs_1+\minDegs_2)$
      \end{enumerate}
  }
\end{algobox}

\begin{proof}[Proof of \cref{thm:fast_pab}]
  Concerning the base case of the recursion at Step~\textbf{1},
  \citep[Prop.\,7.1]{JeNeScVi17} shows that it correctly computes the
  $\shifts$-Popov basis of $\modApp$ using $\bigO{\nbun^\expmatmul
  \log(\nbun)}$ operations. When the algorithm is called on an instance with
  $\vsdim>\nbun$, Step~\textbf{1} is performed less than $2\vsdim/\nbun$ times
  in the whole computation, thus leading to a total contribution of
  $\bigO{\nbun^{\expmatmul-1} \vsdim \log(\nbun)}$ operations in the cost
  bound.

  Let us now study Step~\textbf{3}, where $\vsdim>\nbun$ and $\nbeq<\nbun$. The
  instance $(\orders,\sys)$ is first split into two instances
  $(\orders_1,\sys_1)$ and $(\orders_2,\sys_2)$ such that $\sumVec{\orders_1} =
  \lfloor\vsdim/2\rfloor$ and $\sumVec{\orders_2} = \lceil\vsdim/2\rceil$, and
  with $\cdeg{\sys_1} < \orders_1$ and $\cdeg{\sys_2} < \orders_2$.
  Furthermore, since $\nbeq<\nbun$, the column dimensions of both $\sys_1$ and
  $\sys_2$ are less than their row dimension, so that the recursive calls at
  Steps~\textbf{3.e} and~\textbf{3.g} will not lead to entering
  Step~\textbf{2}. We note that when $d=\order_{i_0}$ the first entry of
  $\orders_2$ is zero; then, one can discard this entry and the corresponding
  zero column of $\sys_2$.

  At Step~\textbf{3.f}, the residual $\res$ is computed in
  $\bigO{\polmatmultime{\nbun,\vsdim/\nbun}}$ operations according to the
  second item of \cref{lem:compute_residual}. Indeed, we have
  $\vsdim>\nbun>\nbeq$, $\sumVec{\cdeg{\appbas_1}} \le \lfloor\vsdim/2\rfloor
  \le \vsdim$ by \cref{lem:degdet_appbasis}, and $\sumVec{\orders_2} =
  \lceil\vsdim/2\rceil \le \vsdim$ by construction.

  Let us define the shift $\shifts[t] \in \shiftSpace$ as $\shifts[t] =
  \rdeg[\shifts]{\appbas_1} = \shifts+\minDegs_1$. Suppose that the recursive
  calls correctly compute the $\shifts$- and $\shifts[t]$-Popov
  bases $\appbas_1$ and $\appbas_2$ of $\modAppCustom{\orders_1}{\sys_1}$ and
  $\modAppCustom{\orders_2}{\res}$. Then, the $\shifts$-minimal degree of
  $\modApp$ is $\minDegs_1 + \minDegs_2$ according to the fourth item of
  \cref{lem:recursiveness}. Thus, by \cref{prop:algo:knowndeg_pab},
  Step~\textbf{3.h} computes the sought approximant basis in
  $\bigO{\appbastime{\nbun,\vsdim/\nbun}}$ operations.

  The recursive calls (Steps~\textbf{3.e} and \textbf{3.g}) are with the same
  dimension $\nbun$ and half the total order $\vsdim/2$, hence the cost bound
  in the case $\nbeq<\nbun$.

  Step~\textbf{2} deals with the case $\nbeq\ge\nbun$, and starts by calling
  \cref{algo:reduce_coldim} to efficiently reduce to $\nbeq<\nbun$. According
  to the above discussion, Step~\textbf{2} may only be entered once, at the
  initial call to the algorithm. The correctness and cost bound in the case
  $\nbeq\ge\nbun$ then follow from \cref{prop:algo:reduce_coldim} and from the
  arguments used above concerning Step~\textbf{3}.
\end{proof}

\section{Computing approximant bases for weakly unbalanced shifts}
\label{sec:weakly_unbalanced_shift}

In this section, we describe approximant basis algorithms which are efficient
when the shift is weakly unbalanced around its minimum value
(\cref{subsec:weakly_unbalanced_around_min}) or around its maximum value
(\cref{subsec:weakly_unbalanced_around_max}).  We recall these notions from
\cref{sec:intro}. In the first case, this means that $\shifts$ satisfies the
assumption $\hypsmin$, that is, $\sumVec{\shifts-\min(\shifts)} \in
\bigO{\vsdim}$ with $\vsdim = \sumVec{\orders}$. In particular, a balanced
shift (that is, one which satisfies $\hypsbal$: $\max(\shifts) - \min(\shifts)
\in \bigO{\vsdim/\nbun}$) also satisfies $\hypsmin$.  In the second case, this
means that $\shifts$ satisfies $\hypsmax$: $\sumVec{\!\max(\shifts)-\shifts}
\in \bigO{\vsdim}$.

For shifts satisfying $\hypsmin$, any $\shifts$-minimal approximant basis
$\appbas$ has small average row degree $\degExp$, which means that the
overlapping linearization of \cref{sec:knowndeg:overlapping_parlin} at degree
$\degExp$ will efficiently recover a large number of the rows of $\appbas$ (all
those of degree $\le \degExp$).  Then, \citet{ZhoLab12} show how the computed
rows allow us to discard a corresponding large number of rows and columns in
the overlapping linearization at degree $2\degExp$, making it efficient to
recover the rows of $\appbas$ of degree $\le 2\degExp$.  This process is
continued until all rows are obtained.

In \cref{subsec:weakly_unbalanced_around_min}, we present a generalization of
\citep[Algo.\,1]{ZhoLab12} which supports arbitrary orders and returns the
basis in $\shifts$-ordered weak Popov form. We do not assume that $\shifts$
satisfies $\hypsmin$, but we describe the algorithm and a complexity analysis
using the parameter $\sumVec{\shifts-\min(\shifts)}$ (see
\cref{prop:algo:zhou_labahn_min}). Besides, we observe that this generalized
algorithm remains efficient: it has the same cost bound as in
\citepalias[Thm.\,5.3]{ZhoLab12} if we assume $\hypsmin$.

For shifts satisfying $\hypsmax$, an $\shifts$-minimal approximant basis
$\appbas$ may have both large average row degree and large average column
degree. Nevertheless, under this assumption, the size of $\appbas$ remains in
$\bigO{\nbun\vsdim}$, and we can guess the location of the columns of $\appbas$
which may have uniformly large degrees: they correspond to the smallest entries
of the shift. For example, with $\shifts = (-\vsdim,0,\ldots,0)$, only the
first column of $\appbas$ may have all its entries of degree close to $\vsdim$.
Based on this, \citepalias[Algo.\,2]{ZhoLab12} uses output column linearization
to balance the degrees according to this guessed column degree profile of
$\appbas$.  This is similar to the output column linearization of
\cref{algo:knowndeg_pab}, except that here we have no guarantee that the
guessed column degree is the actual column degree of $\appbas$.  As a result,
the linearization will be called a logarithmic number of times, until all rows
of $\appbas$ are revealed.  The efficiency of each step depends on the quantity
$\sumVec{\!\max(\shifts)-\shifts}$, which is assumed small in $\hypsmax$.

In \cref{subsec:weakly_unbalanced_around_max}, we present a generalization of
\citepalias[Algo.\,2]{ZhoLab12} which supports arbitrary orders and returns a
basis in $\shifts$-ordered weak Popov form. We do not assume that $\shifts$
satisfies $\hypsmax$ but the cost bound is parametrized by
$\sumVec{\!\max(\shifts)-\shifts}$ (see \cref{prop:algo:zhou_labahn_max}).  As
above, this generalized algorithm is efficient: it has the same cost bound as
in \citepalias[Thm.\,6.14]{ZhoLab12} if we assume \(\hypsmax\).

Before going into detail, we remark that the first item (resp.~second item) of
\cref{thm:zhou_labahn} follows as a corollary of
\cref{prop:algo:zhou_labahn_min} (resp.~\cref{prop:algo:zhou_labahn_max}),
although these propositions only prove that we can compute an $\shifts$-ordered
weak Popov basis of $\modApp$ within the claimed cost bound. Indeed, such a
basis reveals the $\shifts$-minimal degree of $\modApp$ and therefore it only
remains to call \cref{algo:knowndeg_pab}, which also fits within the claimed
cost bound, to obtain the $\shifts$-Popov basis.

\subsection{Weakly unbalanced shift around its minimum value}
\label{subsec:weakly_unbalanced_around_min}

Here we consider \(\shifts\)-minimal approximant bases for shifts such that
$\sumVec{\shifts-\min(\shifts)}$ is small. We extend the approach of
\citep[Sec.\,3~to\,5]{ZhoLab12} to work with an arbitrary order, and we seek a
basis in $\shifts$-ordered weak Popov form. In this approach, one computes
approximants for overlapping linearizations of $(\orders,\sys)$, for a
linearization degree parameter $\degExp$ which is doubled iteratively until a
basis of $\modApp$ is obtained. The correctness is based on the next result,
which shows how to use the knowledge of a basis of $\modAppOvlpLin$ to find a
basis of $\modAppOvlpLin[2]$ (see \cref{dfn:ovlplin} for the overlapping
linearization giving the matrix \(\ovlpLin{\orders,\degExp}{\sys}\) and the order
\(\ovlpLinOrd{\degExp}{\orders}\)).

Hereafter, for $\rdim\in\ZZp$, we write $\idMatEven{\rdim}$ for the
$\rdim\times(\lceil\rdim/2\rceil-1)$ matrix whose column $k$ is the column $2k$
of $\idMat[\rdim]$, and $\idMatOdd{\rdim}$ for the
$\rdim\times(\lfloor\rdim/2\rfloor+1)$ submatrix of $\idMat[\rdim]$ formed by
the remaining columns.  We stress that if $\rdim$ is even, the last column of
$\idMat[\rdim]$ does not appear in $\idMatEven{\rdim}$ but in
$\idMatOdd{\rdim}$.  In particular, $\idMatEven{1}$ and $\idMatEven{2}$ are the
empty $1\times 0$ and $2\times 0$ matrices, while $\idMatOdd{2} = \idMat[2]$.
Besides, in what follows $\idMatEven{m}$ and $\idMatOdd{m}$ refer to the
$0\times 0$ matrix when $\rdim \in \{-1,0\}$, and we use the notation
$\matz_{\rdim\times?}$ or $\matz_{?\times\cdim}$ for the zero matrix when the
row dimension $\rdim$ or the column dimension $\cdim$ is not clear from the
context.

\begin{lemma}
  \label{lem:correctness_zholab_min}
  Let $\orders=(\order_1,\ldots,\order_\nbeq) \in \orderSpace$, let $\sys \in
  \sysSpace$ with $\cdeg{\sys} < \orders$, let $\shifts\in\shiftSpace$, and let
  $\degExp\in\ZZp$.  As in \cref{dfn:ovlplin}, let $\quoExp_i =
  \lceil\frac{\order_i}{\degExp}-1\rceil$ for $1\le i\le \nbeq$ and
  $\expand{\nbeq} = \sum_{1\le i\le\nbeq} \max(\quoExp_i-1,0)$.  Then, consider
  the overlapping linearization $\ovlpLin{\orders,2\degExp}{\sys} \in
  \polMatSpace[(\nbun+\expand{\nbeq}_2)][(\nbeq+\expand{\nbeq}_2)]$, with
  \[
    \expand{\nbeq}_2 = \sum_{1\le i\le\nbeq}
  \max\left(\left\lceil\frac{\order_i}{2\degExp}-1\right\rceil-1,0\right) = \sum_{1\le i\le\nbeq}
  \max(\lfloor\quoExp_i/2\rfloor-1,0).
  \]
  We augment this matrix with \(\expand{\nbeq}-\expand{\nbeq}_2\) zero rows in
  order to define
  \[
    \sysL_2 = 
    \diag{\idMat[\nbun],\idMatEven{\quoExp_1-1},\ldots,\idMatEven{\quoExp_\nbeq-1}}
    \ovlpLin{\orders,2\degExp}{\sys}
    =
    \permMat^{-1}
    \begin{bmatrix}
      \ovlpLin{\orders,2\degExp}{\sys} \\
      \matz
    \end{bmatrix}
    \;\;
    \in \polMatSpace[(\nbun+\expand{\nbeq})][(\nbeq+\expand{\nbeq}_2)],
  \]
  where $\permMat$ is the inverse of the permutation matrix
  \[
    \permMat^{-1} =
    \begin{bmatrix}
      \idMat[\nbun] \\
      & \idMatEven{\quoExp_1-1} & & & \idMatOdd{\quoExp_1-1} \\
      & & \ddots & & & \ddots \\
      & & &  \idMatEven{\quoExp_\nbeq-1} & & & \idMatOdd{\quoExp_\nbeq-1}
    \end{bmatrix}
    \; \in \matSpace[(\nbun+\expand{\nbeq})].
  \]
  Now define a matrix $\selCol$ which, through right-multiplication, selects
  a given set of $\nbeq+\expand{\nbeq}_2$ columns from any matrix with
  $\nbeq+\expand{\nbeq}$ columns, and a matrix $\selColComp$ which selects
  the $\expand{\nbeq}-\expand{\nbeq}_2$ remaining columns:
  $\selCol = \diag{\selCol_1,\ldots,\selCol_\nbeq} \in
  \matSpace[(\nbeq+\expand{\nbeq})][(\nbeq+\expand{\nbeq}_2)]$
  and
  $\selColComp = \diag{\selColComp_1,\ldots,\selColComp_\nbeq} \in
  \matSpace[(\nbeq+\expand{\nbeq})][(\expand{\nbeq}-\expand{\nbeq}_2)]$
  with, for $1\le i\le \nbeq$,
  \[
    \selCol_i =
    \begin{bmatrix}
      1 \\
      & \idMatEven{\quoExp_i-1} \\
    \end{bmatrix}
    \in \matSpace[\max(\quoExp_i,1)][\max(\lfloor\quoExp_i/2\rfloor,1)]
    \quad\text{and}\quad
    \selColComp_i =
    \begin{bmatrix}
      \matz_{1\times ?} \\
      \idMatOdd{\quoExp_i-1}
    \end{bmatrix}
    \in \matSpace[\max(\quoExp_i,1)][(\max(\quoExp_i,1)-\max(\lfloor\quoExp_i/2\rfloor,1))]
    .
  \]
  By construction, we have $\ovlpLin{\orders,\degExp}{\sys} \selCol = \sysL_2
  \bmod \xDiag{\ovlpLinOrd{\degExp}{\orders}\selCol}$ and $0 \le
  \ovlpLinOrd{2\degExp}{\orders} - \ovlpLinOrd{\degExp}{\orders}\selCol \le
  2\degExp$.

  Let us define the order $\ordersL =
  (\ovlpLinOrd{\degExp}{\orders},\ovlpLinOrd{2\degExp}{\orders}) \in
  \ZZp^{2\nbeq+\expand{\nbeq}+\expand{\nbeq}_2}$, the shifts $\shiftsL =
  (\shifts-\min(\shifts),\unishift) \in \ZZ^{\nbun+\expand{\nbeq}}$ and
  $\expand{\shifts} = (\shifts-\min(\shifts),\unishift) \in
  \ZZ^{\nbun+\expand{\nbeq}_2}$, and the matrix $\sysL =
  [\ovlpLin{\orders,\degExp}{\sys} \;\; \sysL_2] \in
  \polMatSpace[(\nbun+\expand{\nbeq})][(2\nbeq+\expand{\nbeq}+\expand{\nbeq}_2)]$.
  Then,
  \begin{itemize}
    \item For any $\expand{\shifts}$-ordered weak
      Popov basis $\appbas \in \polMatSpace[(\nbun+\expand{\nbeq}_2)]$ of
      $\modAppOvlpLin[2]$,
      the matrix
      \begin{equation}
        \label{eqn:build_appbasL}
        \permMat^{-1}
        \begin{bmatrix}
          \appbas & -\appbas_\ell \, \expand{\sys} \, \selColComp \bmod
                        \xDiag{\ovlpLinOrd{\degExp}{\orders}\selColComp} \\
          \matz & \xDiag{\ovlpLinOrd{\degExp}{\orders}\selColComp}
        \end{bmatrix}
        \permMat
        \;\in \polMatSpace[(\nbun+\expand{\nbeq})]
      \end{equation}
      is an $\shiftsL$-ordered weak Popov basis of $\modAppL$,
      where $\appbas_\ell \in \polMatSpace[(\nbun+\expand{\nbeq}_2)][\nbun]$ is
      the submatrix of $\appbas$ formed by its leftmost $\nbun$ columns and
      $\expand{\sys} \in \polMatSpace[\nbun][(\nbeq+\expand{\nbeq})]$ is the
      submatrix of $\ovlpLin{\orders,\degExp}{\sys}$ formed by its top $\nbun$
      rows.
    \item For any $\shiftsL$-ordered weak Popov basis $\appbasL \in
      \polMatSpace[(\nbun+\expand{\nbeq})]$ of $\modAppL$, the leading
      principal $(\nbun+\expand{\nbeq}_2)\times(\nbun+\expand{\nbeq}_2)$
      submatrix of $\permMat\appbasL\permMat^{-1}$ is an
      $\expand{\shifts}$-ordered weak Popov basis of $\modAppOvlpLin[2]$.
    \item For any vectors $\app\in\polMatSpace[1][\nbun]$ and
      $\row{q} \in \polMatSpace[1][\expand{\nbeq}]$ such that $\rdeg{\row{q}} <
      \rdeg{\app} \le \degExp$ and $[\app \;\; \row{q}] \in \modAppOvlpLin$, we
      have $[\app \;\; \row{q}] \in \modAppL$.
  \end{itemize}
\end{lemma}
\begin{proof}
  \emph{(First item.)}  We define $\mat{Q} = -\appbas_\ell \, \expand{\sys} \,
  \selColComp \bmod \xDiag{\ovlpLinOrd{\degExp}{\orders}\selColComp} \in
  \polMatSpace[(\nbun+\expand{\nbeq}_2)][(\expand{\nbeq}-\expand{\nbeq}_2)]$
  and we denote by $\mat{B}$ the matrix in \cref{eqn:build_appbasL}.  Then, we
  start by showing that all rows of $\mat{B}$ are in $\modAppL$, that is,
  $\mat{B} \sysL_2 = 0 \bmod \xDiag{\ovlpLinOrd{2\degExp}{\orders}}$ and
  $\mat{B} \ovlpLin{\orders,\degExp}{\sys} = 0 \bmod
  \xDiag{\ovlpLinOrd{\degExp}{\orders}}$.  First, we have
  \[
    \mat{B} \sysL_2 =
    \permMat^{-1}
    \begin{bmatrix}
      \appbas & \mat{Q} \\
      \matz & \xDiag{\ovlpLinOrd{\degExp}{\orders}\selColComp}
    \end{bmatrix}
    \permMat
    \permMat^{-1}
    \begin{bmatrix}
      \ovlpLin{\orders,2\degExp}{\sys} \\
      \matz
    \end{bmatrix}
    =
    \permMat^{-1}
    \begin{bmatrix}
      \appbas \ovlpLin{\orders,2\degExp}{\sys} \\
      \matz
    \end{bmatrix}
    =
    0 \bmod \xDiag{\ovlpLin{2\degExp}{\orders}}
  \]
  by assumption on $\appbas$. Since $\ovlpLinOrd{2\degExp}{\orders} \ge
  \ovlpLinOrd{\degExp}{\orders}\selCol$, this also gives $\mat{B}
  \ovlpLin{\orders,\degExp}{\sys} \selCol = \mat{B} \sysL_2 = 0 \bmod
  \xDiag{\ovlpLinOrd{\degExp}{\orders}\selCol}$ and thus it remains to show
  that $\mat{B} \ovlpLin{\orders,\degExp}{\sys} \selColComp = 0 \bmod
  \xDiag{\ovlpLinOrd{\degExp}{\orders}\selColComp}$.  By construction, the last
  $\expand{\nbeq}$ rows of $\permMat\ovlpLin{\orders,\degExp}{\sys}\selColComp$
  are formed by $\expand{\nbeq}_2$ zero rows followed by the identity matrix:
  \begin{equation}
    \label{eqn:zholab_formula_zero_id}
    \begin{bmatrix}
      \matz_{?\times\nbun} & \idMat[\expand{\nbeq}]
    \end{bmatrix}
    \permMat\ovlpLin{\orders,\degExp}{\sys}\selColComp
    =
    \left[ \begin{smallmatrix}
      \trsp{(\idMatEven{\quoExp_1-1})} \\
      & \ddots \\
      & & \trsp{(\idMatEven{\quoExp_\nbeq-1})} \\
      \trsp{(\idMatOdd{\quoExp_1-1})} \\
      & \ddots \\
      & & \trsp{(\idMatOdd{\quoExp_\nbeq-1})} \\
    \end{smallmatrix} \right]
    \left[ \begin{smallmatrix}
      \matz_{?\times 1} & \idMat[\quoExp_1-1] \\
      & & \ddots \\
      & & & \matz_{?\times 1} & \idMat[\quoExp_\nbeq-1]
    \end{smallmatrix}
    \right]
    \left[ \begin{smallmatrix}
      \matz_{1\times ?} \\
      \idMatOdd{\quoExp_1-1} \\
      & \ddots \\
      & & \matz_{1\times?} \\
      & & \idMatOdd{\quoExp_\nbeq-1}
    \end{smallmatrix} \right]
    =
    \begin{bmatrix}
      \matz_{\expand{\nbeq}_2\times?} \\
      \idMat[\expand{\nbeq}-\expand{\nbeq}_2]
    \end{bmatrix} .
  \end{equation}
  As a consequence, we have
  \[
    \mat{B} \ovlpLin{\orders,\degExp}{\sys} \selColComp
    =
    \permMat^{-1}
    \begin{bmatrix}
      \appbas & \mat{Q} \\
      \matz & \xDiag{\ovlpLinOrd{\degExp}{\orders}\selColComp}
    \end{bmatrix}
    \permMat \ovlpLin{\orders,\degExp}{\sys} \selColComp
    =
    \permMat^{-1}
    \begin{bmatrix}
      \appbas_\ell \, \expand{\sys} \, \selColComp + \mat{Q} \\
      \xDiag{\ovlpLinOrd{\degExp}{\orders}\selColComp}
    \end{bmatrix}
    = 0 \bmod \xDiag{\ovlpLinOrd{\degExp}{\orders}\selColComp}.
  \]

  Now, we prove that any $\appL \in \modAppL$ is a combination of the rows of
  $\mat{B}$.  Write $\appL = [\app \;\; \row{q}]\permMat$ with $\app \in
  \polMatSpace[1][(\nbun+\expand{\nbeq}_2)]$ and $\row{q} \in
  \polMatSpace[1][(\expand{\nbeq}-\expand{\nbeq}_2)]$.  Then, $\appL \in
  \modAppL$ implies first $\app \in
  \modAppOvlpLin[2]$,
  hence $\app = \rowgrk{\lambda} \appbas$ for some $\rowgrk{\lambda} \in
  \polMatSpace[1][(\nbun+\expand{\nbeq}_2)]$, and second $\rowgrk{\lambda}
  \appbas_\ell \expand{\sys} \selColComp + \row{q} = [\app \;\; \row{q}]
  \permMat \ovlpLin{\orders,\degExp}{\sys}\selColComp = 0 \bmod
  \xDiag{\ovlpLinOrd{\degExp}{\orders}\selColComp}$, hence $\row{q} =
  \rowgrk{\lambda} \mat{Q} + \rowgrk{\mu}
  \xDiag{\ovlpLinOrd{\degExp}{\orders}\selColComp}$ for some $\rowgrk{\mu} \in
  \polMatSpace[1][(\expand{\nbeq}-\expand{\nbeq}_2)]$. Thus, $\appL =
  [\rowgrk{\lambda} \;\; \rowgrk{\mu}] \permMat \mat{B}$.

  It remains to prove that $\permMat\mat{B}\permMat^{-1}$ is in
  $\shiftsL$-ordered weak Popov form; then, the second item of
  \cref{lem:permuted_owpopov} shows that $\mat{B}$ is also in
  $\shiftsL$-ordered weak Popov form (note that $\shiftsL \permMat =
  \shiftsL$). Since the bottom-right block of $\permMat\mat{B}\permMat^{-1}$ is
  a diagonal matrix and the top-left block is already in
  $\expand{\shifts}$-ordered weak Popov form, where
  $\shiftsL=(\expand{\shifts},\unishift)$, it is enough to show that
  $\rdeg{\mat{Q}} < \rdeg[\expand{\shifts}]{\appbas}$.  Since
  $\expand{\shifts}\ge0$, we have $\rdeg{\appbas} \le
  \rdeg[\expand{\shifts}]{\appbas}$ and thus it is enough to show that
  $\rdeg{\mat{Q}} < \rdeg{\appbas}$.  Consider a row $[\app \;\; \row{q}]$ of
  $[\appbas \;\; \mat{Q}]$.  If $\rdeg{\app} \ge 2\degExp$, then
  $\rdeg{\row{q}} < \rdeg[\expand{\shifts}]{\app}$ follows since by
  construction we have $\rdeg{\row{q}} < \max(\ovlpLinOrd{\degExp}{\orders})
  \le 2\degExp$.  If $\rdeg{\app} < 2\degExp$, since $\app$ is in
  $\modAppOvlpLin[2]$,
  the second item of \cref{lem:ovlplin_app} (with parameter $2\degExp$) shows
  that the $\nbun$ leftmost entries of $\app$ are in $\modApp$; then, the first
  item of the same lemma (with parameter $\degExp$) gives in particular
  $\rdeg{\row{q}} < \rdeg{\app}$.

  \emph{(Second item.)}  The first item implies that $\appbasL = \mat{U} \mat{B}$ for
  some unimodular matrix $\mat{U}$.  Let $\mat{U}_0$ and $\appbas_0$ denote the
  leading principal $(\nbun+\expand{\nbeq}_2)\times(\nbun+\expand{\nbeq}_2)$
  submatrices of $\permMat \mat{U} \permMat^{-1}$ and $\permMat \appbasL
  \permMat^{-1}$.  The first item of \cref{lem:permuted_owpopov} shows that
  $\appbas_0$ is in $\expand{\shifts}$-ordered weak Popov form.  Besides, the
  identity $\permMat\appbasL\permMat^{-1} = \permMat\mat{U}\permMat^{-1}
  \permMat \mat{B}\permMat^{-1}$ and the triangular shape of
  $\permMat\mat{B}\permMat^{-1}$ yield $\appbas_0 = \mat{U}_0 \appbas$.
  Furthermore, $\permMat\appbasL\permMat^{-1}$ and
  $\permMat\mat{B}\permMat^{-1}$ being $\shiftsL$-ordered weak Popov bases of
  the same module, they have the same $\shiftsL$-minimal degree (see
  \cref{subsec:forms}), and thus the same $\shiftsL$-row degree.  This implies
  that their leading principal submatrices $\appbas_0$ and $\appbas$ have the
  same $\expand{\shifts}$-row degree, hence
  \[
    \deg(\det(\mat{U}_0)) = \deg(\det(\appbas_0)) -\deg(\det(\appbas))
    = \sumVec{\rdeg[\expand{\shifts}]{\appbas_0}} -  \sumVec{\rdeg[\expand{\shifts}]{\appbas}}
    = 0.
  \]
  This means that $\mat{U}_0$ is unimodular, and therefore $\appbas_0$ is a
  basis of
  $\modAppOvlpLin[2]$.

  \emph{(Third item.)}  We want to prove that $[\app \;\; \row{q}] \in
  \modAppCustom{\ovlpLinOrd{2\degExp}{\orders}}{\sysL_2}$.  The second item of
  \cref{lem:ovlplin_app} implies that $\app\in\modApp$, while its first item
  gives the uniqueness of $\row{q}$: if $\row{r} \in
  \polMatSpace[1][\expand{\nbeq}]$ is such that $\rdeg{\row{r}}<\rdeg{\app}$
  and $[\app \;\; \row{r}] \in \modAppOvlpLin$, then $\row{r} = \row{q}$.
  (Note that here the constraint $\cdeg{\col{r}} <
  \ovlpLinOrd{\degExp}{\orders}\trsp{\emat}$ from \cref{lem:ovlplin_app} is
  implied by $\rdeg{\row{r}}<\degExp <
  \min(\ovlpLinOrd{\degExp}{\orders}\trsp{\emat})$.)

  \cref{lem:ovlplin_app} gives $\row{q}_2 \in
  \polMatSpace[1][\expand{\nbeq}_2]$ such that $\rdeg{\row{q}_2} < \rdeg{\app}$
  and $[\app \;\; \row{q}_2] \in \modAppOvlpLin[2]$.  Then, define $\row{q}_3 =
  -\app\expand{\sys}\selColComp \bmod
  \xDiag{\ovlpLinOrd{\degExp}{\orders}\selColComp}$, which is a subvector of
  $\row{q} = -\app\expand{\sys}\trsp{\emat} \bmod
  \xDiag{\ovlpLinOrd{\degExp}{\orders}\trsp{\emat}}$ since $\selColComp$
  selects a subset of the columns selected by $\trsp{\emat}$.  Let further
  $\row{r} = [\row{q}_2 \;\; \row{q}_3] [\matz \;\;
  \idMat[\expand{\nbeq}]]\permMat \in \polMatSpace[1][\expand{\nbeq}]$; by
  construction, we have $\rdeg{\row{r}} < \rdeg{\app}$.  We are going to show
  that $[\app \;\; \row{r}] \in
  \modAppCustom{\ovlpLinOrd{2\degExp}{\orders}}{\sysL_2}$ and $[\app \;\;
  \row{r}] \in \modAppOvlpLin$: the latter point implies $\row{r} = \row{q}$ by
  the mentioned uniqueness, and then the former point gives $[\app \;\;
  \row{q}] \in \modAppCustom{\ovlpLinOrd{2\degExp}{\orders}}{\sysL_2}$, thus
  concluding the proof.

  Noticing that $[\app \;\; \row{r}] = [\app \;\; \row{q}_2 \;\; \row{q}_3]
  \permMat$, the first point follows by construction of $\sysL_2$:
  \[
    [\app \;\; \row{r}] \sysL_2
    =
    [\app \;\; \row{q}_2 \;\; \row{q}_3] \permMat \permMat^{-1}
      \begin{bmatrix}
        \ovlpLin{\orders,2\degExp}{\sys} \\ \matz
      \end{bmatrix}
    =
    [\app \;\; \row{q}_2] \ovlpLin{\orders,2\degExp}{\sys} 
    = 0 \bmod \xDiag{\ovlpLinOrd{2\degExp}{\orders}}.
  \]
  Furthermore, since $\ovlpLinOrd{2\degExp}{\orders} \ge
  \ovlpLinOrd{\degExp}{\orders}\selCol$ we can consider the same identity
  modulo $\xDiag{\ovlpLinOrd{\degExp}{\orders}\selCol}$.  Using
  $\ovlpLin{\orders,\degExp}{\sys} \selCol = \sysL_2 \bmod
  \xDiag{\ovlpLinOrd{\degExp}{\orders}\selCol}$, this directly yields $[\app
  \;\; \row{r}] \ovlpLin{\orders,\degExp}{\sys} \selCol = 0 \bmod
  \xDiag{\ovlpLinOrd{\degExp}{\orders}\selCol}$.  For the second point, it
  remains to show $[\app \;\; \row{r}] \ovlpLin{\orders,\degExp}{\sys}
  \selColComp = 0 \bmod \xDiag{\ovlpLinOrd{\degExp}{\orders}\selColComp}$.
  This follows from the definition of $\row{q}_3$ since
  \cref{eqn:zholab_formula_zero_id} gives
  \(
    [\app \;\; \row{r}] \ovlpLin{\orders,\degExp}{\sys} \selColComp
    =
    [\app \;\; \row{q}_2 \;\; \row{q}_3] \permMat \ovlpLin{\orders,\degExp}{\sys} \selColComp
    = \app \expand{\sys} \selColComp + \row{q}_3
  \).
\end{proof}

We remark that working with matrices in ordered weak Popov form allows us to
directly locate the submatrix that contains the sought basis, and thus to avoid
resorting to computations of row rank profiles as was done for example in
\citep[Thm.\,3.15~and~Algo.\,1]{ZhoLab12}.

The second item in this lemma implies that, knowing a basis of
$\modAppOvlpLin$, we can obtain a basis of $\modAppOvlpLin[2]$ via the
classical approach of computing a residual, a second approximant basis, and the
product of the two bases.  Furthermore, the third item shows that rows of degree
less than $\degExp$ in the first basis are already in $\modAppOvlpLin[2]$.
Thus, they can be discarded when computing the second basis (see
\cref{lem:ok_rows}); this is a key property for the efficiency of
\cref{algo:zhou_labahn_min}.  The next result formalizes these remarks, using
notation from \cref{lem:correctness_zholab_min}.

\begin{corollary}
  \label{cor:correctness_zholab_min}
  Let $\appbas \in \polMatSpace[(\nbun+\expand{\nbeq})]$ be an
  $\shiftsL$-ordered weak Popov basis of $\modAppOvlpLin$, let $\okRows
  \subseteq \{1,\ldots,\nbun+\expand{\nbeq}\}$ be the set of indices $i$ of the
  rows $\matrow{\appbas}{i} = [\app \;\; \row{q}]$ such that $\rdeg{\row{q}} <
  \rdeg{\app} \le \degExp$, where $\app \in \polMatSpace[1][\nbun]$ and
  $\row{q} \in \polMatSpace[1][\expand{\nbeq}]$.  Let further $\okRowsComp =
  \{1,\ldots,\nbun+\expand{\nbeq}\}\setminus\okRows$ be the complement of
  $\okRows$ and let $\okRowsNb$ denote the cardinality of $\okRows$.  We have
  $\okRows \subseteq \{1,\ldots,\nbun\}$.

  Now, consider the tuples $\tuplegrk{\mu} =
  \ovlpLinOrd{\degExp}{\orders}\selCol$ and $\tuplegrk{\nu} =
  \ovlpLinOrd{2\degExp}{\orders}$ both in $\ZZp^{\nbeq+\expand{\nbeq}_2}$, as
  well as the residual $\res = \matrows{\appbas}{\okRowsComp} \sysL_2
  \xDiag{-\tuplegrk{\mu}} \bmod \xDiag{\tuplegrk{\nu}-\tuplegrk{\mu}} \in
  \polMatSpace[(\nbun+\expand{\nbeq}-\okRowsNb)][(\nbeq+\expand{\nbeq}_2)]$ and
  a basis $\appbas_2 \in \polMatSpace[(\nbun+\expand{\nbeq}-\okRowsNb)]$ of
  $\modAppCustom{\tuplegrk{\nu}-\tuplegrk{\mu}}{\res}$ in
  $\rdeg[\shiftsL]{\matrows{\appbas}{\okRowsComp}}$-ordered weak Popov form.
  Modify $\appbas$ by left-multiplying its submatrix
  $\matrows{\appbas}{\okRowsComp}$ by $\appbas_2$, that is, perform
  the operation $\matrows{\appbas}{\okRowsComp} \assign
  \appbas_2\matrows{\appbas}{\okRowsComp}$.  Then, the leading principal
  $(\nbun+\expand{\nbeq}_2)\times(\nbun+\expand{\nbeq}_2)$ submatrix of
  $\permMat\appbas\permMat^{-1}$ is an $\expand{\shifts}$-ordered weak Popov
  basis of $\modAppOvlpLin[2]$.
\end{corollary}
\begin{proof}
  The fact that $\okRows \subseteq \{1,\ldots,\nbun\}$ follows by definition of
  the $\shiftsL$-ordered weak Popov form.  Indeed, since
  $\shiftsL=(\shifts-\min(\shifts),\unishift)$, such a row $[\app \;\;
  \row{q}]$ with $\rdeg{\row{q}} < \rdeg{\app} \le
  \rdeg[\shifts-\min(\shifts)]{\app}$ must have its $\shiftsL$-pivot entry in
  $\app$, or in other words, its $\shiftsL$-pivot index in
  $\{1,\ldots,\nbun\}$.  Since the $\shiftsL$-pivot entries are on the
  diagonal, $[\app \;\; \row{q}]$ must be one of the first $\nbun$ rows of
  $\appbas$.

  The other claims follow directly from
  \cref{lem:correctness_zholab_min,lem:ok_rows}. 
\end{proof}

This suggests an algorithm which computes approximant bases iteratively for the
overlapping linearized problems with a linearization parameter $\degExp$ which
is doubled at each step.  When the parameter reaches $\degExp > \max(\orders)$,
we actually have $\ovlpLinOrd{\degExp}{\orders} = \orders$ and
$\ovlpLin{\orders,\degExp}{\sys} = \sys$, and therefore the computed basis is a
basis of $\modApp = \modAppOvlpLin$.  In what follows,
let $\vsdim=\sumVec{\orders}$.

In this process, the number of columns of the approximant instances steadily
decreases.  On the first hand, the number of columns $\expand{\nbeq}$ added by
the overlapping linearization is roughly halved when $\degExp$ is doubled.  On
the other hand, only the $\le 2\vsdim/\degExp$ columns of $\sys$ with
corresponding order $\order_i \ge \degExp/2$ need to be considered in the
iteration with linearization parameter $\degExp$, since all the others have
been fully processed already (see the proof of
\cref{prop:algo:zhou_labahn_min} for more details).

Furthermore, the corollary above indicates that if at some iteration one of the
computed approximants in $\modAppOvlpLin$ has degree less than $\degExp$, then
it can be stored as a row of the sought basis and can be discarded in the
computation of the residual and of the second basis.  In the process outlined
above, this allows us to decrease the row dimension each time such a small
degree approximant has been found.

Yet, there remains an obstacle towards efficiency: if the output basis has no
row of small degree, there will be no such row dimension decrease before the
very last few iterations.  In this case, some iterations may ask us to solve
instances with roughly the same dimensions and degrees as the original instance
$(\orders,\sys)$; then, this approach is not faster than a direct call to
\algoname{PM-Basis}.

Nevertheless, there are many shifts for which this worst-case scenario cannot
occur, since the sum of the row degree of an $\shifts$-minimal basis of
$\modApp$ is at most $\sumRdeg = \vsdim + \sumVec{\shifts-\min(\shifts)}$
\cite[Thm.\,4.1]{BarBul92}.  Thus, this $\shifts$-minimal basis has at most
$\sumRdeg/\degExp$ rows of degree $\ge \degExp$; this is especially beneficial
when $\sumRdeg$ is small, that is, for shifts that are weakly unbalanced around
their minimum value (assumption $\hypsmin$).  For example, for the uniform
shift, a $\unishift$-minimal basis has at most $\nbun/2^i$ rows of degree $\ge
2^i \lceil \vsdim/\nbun \rceil$, which means that in our process at least
$\nbun-\nbun/2^i$ rows can be discarded when $\degExp$ has reached $2^i \lceil
\vsdim/\nbun \rceil$.

\begin{algobox}
  \algoInfo
  {Minimal basis for small $\sumVec{\shifts-\min(\shifts)}$}
  {ShiftAroundMinAppBasis}
  {algo:zhou_labahn_min}

  \dataInfos{Input}{
    \item order $\orders \in \orderSpace$,
    \item matrix $\sys \in \sysSpace$ with $\cdeg{\sys} < \orders$,
    \item shift $\shifts \in \shiftSpace$.
  }

  \dataInfo{Output}{an $\shifts$-ordered weak Popov basis of $\modApp$.}

  \algoSteps{
    \item \algoword{If} $\nbeq \ge \nbun$:
      \begin{enumerate}[{\bf a.}]
        \item permute $\orders$ into nonincreasing order,
          and the columns of $\sys$ accordingly
        \item $(\ordersR,\sysR,\shiftsR,\appbas_1) \assign
                \algoname{ReduceColDim}(\orders,\sys,\shifts)$
        \item $\popov_2 \assign \algoname{ShiftAroundMinAppBasis}(\ordersR,\sysR,\shiftsR)$
        \item $\minDegs_1 \assign$ diagonal degrees of $\appbas_1$;~
        $\minDegs_2 \assign$ diagonal degrees of $\popov_2$
        \item \algoword{Return} $\algoname{KnownDegAppBasis}(\orders,\sys,\shifts,\minDegs_1+\minDegs_2)$
      \end{enumerate}
    \item \algoword{Else}:
      \begin{enumerate}[{\bf a.}]
        \item $\degExp \assign \lceil (\sumVec{\orders}+\sumVec{\shifts-\min(\shifts)}) / \nbun\rceil$ \\
          Construct $\ovlpLinOrd{\degExp}{\orders} \in
          \ZZp^{\nbun+\expand{\nbeq}}$ and
          $\ovlpLin{\orders,\degExp}{\sys} \in
          \polMatSpace[(\nbun+\expand{\nbeq})][(\nbeq+\expand{\nbeq})]$ as in
          \cref{dfn:ovlplin} \\
          $\appbas \assign \algoname{PM-Basis}(
                2\degExp,
                \ovlpLin{\orders,\degExp}{\sys}\xDiag{2\degExp-\ovlpLinOrd{\degExp}{\orders}},
                (\shifts-\min(\shifts),\unishift)
                )$ \\
          $\okRows \assign \{i\in\{1,\ldots,\nbun+\expand{\nbeq}\} \mid
            \matrow{\appbas}{i}=[\app\;\;\row{q}] \;\text{is such that}\;
            \rdeg{\row{q}} < \rdeg{\app} \le \degExp \}$, \\
           \phantom{$\okRows \assign$}
           where $\app \in \polMatSpace[1][\nbun]$ and $\row{q} \in
           \polMatSpace[1][\expand{\nbeq}]$
           \eolcomment{for these rows, $\app \in \modApp$}
        \item \algoword{While} $\card{\okRows} < \nbun$:
          \eolcomment{$\okRows \subseteq \{1,\ldots,\nbun\}$ holds, cf.~\cref{cor:correctness_zholab_min}}
          \begin{enumerate}[(i)]
            \item Construct matrices
              $\permMat \in
              \matSpace[(\nbun+\expand{\nbeq})]$ and $\selCol \in
              \matSpace[(\nbeq+\expand{\nbeq})][(\nbeq+\expand{\nbeq}_2)]$ as
              in \cref{lem:correctness_zholab_min}, \\
              tuples $\tuplegrk{\mu} \assign \ovlpLinOrd{\degExp}{\orders}\selCol$
              and $\tuplegrk{\nu} \assign \ovlpLinOrd{2\degExp}{\orders}$
              both in $\ZZp^{\nbeq+\expand{\nbeq}_2}$, and sets\\
              $\remCols \assign \{ j \in \{1,\ldots,\nbeq+\expand{\nbeq}_2\} \mid \nu_j-\mu_j > 0 \}$ and
              $\okRowsComp \assign \{1,\ldots,\nbun+\expand{\nbeq}\} \setminus \okRows$
            \item $\res \assign \matrows{\appbas}{\okRowsComp} \, \permMat^{-1}
                    \left[\begin{smallmatrix} \matcols{\ovlpLin{\orders,2\degExp}{\sys}}{\remCols} \\ \matz \end{smallmatrix}\right]
                      \xDiag{\subTuple{-\tuplegrk{\mu}}{\remCols}} 
                      \bmod \xDiag{\subTuple{\tuplegrk{\nu}}{\remCols}-\subTuple{\tuplegrk{\mu}}{\remCols}}$
            \item $\appbas_2 \assign \algoname{PM-Basis}(
                      2\degExp,
                      \res\xDiag{2\degExp-\subTuple{\tuplegrk{\nu}}{\remCols}+\subTuple{\tuplegrk{\mu}}{\remCols}},
                    \rdeg[(\shifts-\min(\shifts),\unishift)]{\matrows{\appbas}{\okRowsComp}}
                      )$
            \item $\matrows{\appbas}{\okRowsComp} \assign \appbas_2 \matrows{\appbas}{\okRowsComp}$
              \eolcomment{this modifies $\appbas$}
            \item $\appbas \assign$ leading principal $(\nbeq+\expand{\nbeq}_2)
              \times (\nbeq+\expand{\nbeq}_2)$ submatrix of $\permMat \appbas \permMat^{-1}$ \\
              $\degExp \assign 2\degExp$;
              $\expand{\nbeq} \assign \expand{\nbeq}_2$;
              $\okRows \assign \okRows \cup \{i\in\okRowsComp \mid 
                \matrow{\appbas}{i}=[\app\;\;\row{q}] \;\text{is such that}\;
                \rdeg{\row{q}} < \rdeg{\app} \le \degExp \}$,
                 where $\app \in \polMatSpace[1][\nbun]$ and $\row{q} \in
                \polMatSpace[1][\expand{\nbeq}]$
          \end{enumerate}
        \item \algoword{Return} $\appbas$
      \end{enumerate}
  }
\end{algobox}

\begin{proposition}
  \label{prop:algo:zhou_labahn_min}
  \cref{algo:zhou_labahn_min} is correct.  Let $\vsdim = \sumVec{\orders}$, let
  $\sumRdeg = \vsdim + \sumVec{\shifts-\min(\shifts)}$, and let $\order =
  \max(\orders)$.  If $\sumRdeg \le \nbun\order$, then
  \cref{algo:zhou_labahn_min} uses $\costZhoLabMin{\sumRdeg,\nbun,\order}$
  operations in $\field$, where $\costZhoLabMin{\cdot}$ is defined as in
  \cref{eqn:cost_zholab_min}.  If $\sumRdeg > \nbun\order$, it uses
  $\bigO{\appbastime{\nbun,\lceil\vsdim/\nbun\rceil} +
  \appbastime{\nbun,\order}}$ operations in $\field$.
\end{proposition}
\begin{proof}
  The correctness of Step~\textbf{1} follows from
  \cref{lem:recursiveness,prop:algo:reduce_coldim}.  Concerning
  Step~\textbf{2}, we first note that if $\lceil\sumRdeg/\nbun\rceil >
  \order$, then $\ovlpLin{\orders,\degExp}{\sys} = \sys$ and
  $\ovlpLinOrd{\degExp}{\orders} = \orders$ and therefore the call to
  \algoname{PM-Basis} at Step~\textbf{2.a} computes a whole $\shifts$-ordered
  weak Popov basis of $\modApp$.  Then, the loop at Step~\textbf{2.b} is not
  entered, and Step~\textbf{2} uses $\bigO{\appbastime{\nbun,\order}}$ operations
  according to \cref{prop:algo:dac_ab}.

  On the other hand, if $\lceil\sumRdeg/\nbun\rceil \le \order$, the
  correctness of Step~\textbf{2} follows from
  \cref{cor:correctness_zholab_min}, noticing that the loop terminates after at
  most $1+\lfloor\log_2(\order/\lceil\sumRdeg/\nbun\rceil)\rfloor$ iterations
  since $\degExp$ is doubled at each iteration, and as mentioned above
  $\ovlpLin{\orders,\degExp}{\sys} = \sys$ and $\ovlpLinOrd{\degExp}{\orders} =
  \orders$ for $\degExp > \order$.  Furthermore, in this algorithm we use the
  set $\remCols$ to explicitly filter out columns for which the correct order
  has already been reached, thus for which the residual columns are zero.  This
  was not done in \cref{cor:correctness_zholab_min} which focused on
  correctness, yet here it makes it easier to describe column dimensions in the
  following cost analysis.

  Concerning Step~\textbf{2}, we place ourselves at the beginning of an
  iteration, and we start by describing the dimensions and the degrees of the
  matrices involved in the computations.  Then,
  \begin{itemize}
    \item $\matrows{\appbas}{\okRowsComp}$ has dimensions $\card{\okRowsComp}
      \times (\nbun+\expand{\nbeq})$ and degree $< 2\degExp$;
    \item $\permMat^{-1}
      \left[\begin{smallmatrix}
      \matcols{\ovlpLin{\orders,2\degExp}{\sys}}{\remCols} \\
  \matz\end{smallmatrix}\right]$ has dimensions $(\nbun+\expand{\nbeq}) \times
  \card{\remCols}$ and degree $< \max(\ovlpLinOrd{2\degExp}{\orders})
  \le 4\degExp$;
    \item $\res$ has dimensions $\card{\okRowsComp} \times \card{\remCols}$ and
      degree $< \max(\boldsymbol{\nu}-\boldsymbol{\mu}) \le 2\degExp$;
    \item $\appbas_2$ has dimensions $\card{\okRowsComp} \times
      \card{\okRowsComp}$ and degree $< 2\degExp$.
  \end{itemize}
  As above, $\expand{\nbeq}$ is such that $\ovlpLin{\orders,\degExp}{\sys}$ has
  dimensions $(\nbun+\expand{\nbeq}) \times (\nbeq+\expand{\nbeq})$, and
  $\expand{\nbeq} < \vsdim / \degExp$ where $\vsdim = \sumVec{\orders}$.

  Besides, as a consequence of \cite[Thm.\,4.1]{BarBul92}, the sum of the
  degrees of the rows of the sought basis is at most $\sumRdeg$, and thus this
  basis has more than $\nbun - \sumRdeg/\degExp$ rows of degree $\le\degExp$;
  \cref{lem:ovlplin_app} shows that the set
  $\okRows\subseteq\{1,\ldots,\nbun\}$ precisely contains the indices of the
  latter rows. Thus, $\card{\okRows} > \nbun-\sumRdeg/\degExp$, and
  $\card{\okRowsComp} = \nbun+\expand{\nbeq} - \card{\okRows} < \expand{\nbeq}
  + \sumRdeg/\degExp \le 2\sumRdeg/\degExp$.

  Furthermore, note that the entries of $\ovlpLinOrd{\degExp}{\orders}\selCol$
  and $\ovlpLinOrd{2\degExp}{\orders}$ which coincide are exactly those
  corresponding to columns with order $\order_i \le 2\degExp$ (or,
  equivalently, $\quoExp_i=1$): these are columns $\matcol{\sys}{i}$ which
  appear as such in $\ovlpLin{\orders,\degExp}{\sys}$ and also in
  $\ovlpLin{\orders,2^k\degExp}{\sys}$ for all subsequent iterations.  Indeed,
  if $\order_i>2\degExp$, the corresponding entries in
  $\ovlpLinOrd{\degExp}{\orders}\selCol$ are at most $2\degExp$ and cannot
  coincide with those in $\ovlpLinOrd{2\degExp}{\orders}$ which are at least
  $2\degExp+1$.  As a result, $\card{\remCols}$ is the sum of the number
  $\expand{\nbeq}_2$ of columns added by the overlapping linearization with
  degree parameter $2\degExp$, and of the number of indices $i\in
  \{1,\ldots,\nbeq\}$ such that $\order_i>2\degExp$; both numbers are less than
  $\vsdim/(2\degExp)$.  Thus, $\card{\remCols} <
  \vsdim/\degExp$.

  Now, let $\degExp_0 = \lceil \sumRdeg / \nbun\rceil$ be the initial value of
  $\degExp$.  Then, at the beginning of the $k$-th iteration of the loop (the
  first one being for $k=1$), we have $\degExp = 2^{k-1} \degExp_0$ and the
  dimensions satisfy $\nbun+\expand{\nbeq} < 2\nbun$, $\card{\okRowsComp} <
  2\sumRdeg/\degExp = 2^{2-k} \sumRdeg/\degExp_0 \le 2^{2-k} \nbun$, and
  $\card{\remCols} < \vsdim/\degExp = 2^{1-k} \vsdim/\degExp_0 \le
  2^{1-k}\sumRdeg/\degExp_0 \le 2^{1-k} \nbun$.

  Then, both matrix multiplications at Steps~\textbf{2.b.}(ii)
  and~\textbf{2.b.}(iv) use $\bigO{2^{k-1} \polmatmultime{2^{1-k}\nbun,2^{k-1}
      \degExp_0}}$ operations.  Besides, the call to \algoname{PM-Basis} at
      Step~\textbf{2.b.}(iii) uses $\bigO{\appbastime{2^{1-k}\nbun,2^{k-1}
      \degExp_0}}$ operations according to \cref{prop:algo:dac_ab}, while the
      call at Step~\textbf{2.a} uses $\bigO{\appbastime{\nbun,\degExp_0}}$
      operations.  Summing these terms over all iterations gives the cost bound
      announced in the statement, since as explained above the loop terminates
      before or when $k$ reaches $1+\lfloor
      \log_2(\order/\lceil\sumRdeg/\nbun\rceil)\rfloor$.

  Now, independently from assumptions on $\lceil\sumRdeg/\nbun\rceil$,
  Steps~\textbf{1.b} and\,\textbf{1.e} both use
  $\bigO{\appbastime{\nbun,\vsdim/\nbun}}$ operations according to
  \cref{prop:algo:reduce_coldim,prop:algo:knowndeg_pab}; here
  $\lceil\vsdim/\nbun\rceil\in\Theta(\vsdim/\nbun)$ since
  $\vsdim\ge\nbeq\ge\nbun$.  Besides, the former proposition and the
  specification of \algoname{ReduceColDim} ensure that:
  \begin{itemize}
    \item $\deg(\appbas_1) \le 2\vsdim/\nbun$, hence
      $\shifts \le \shiftsR \le \shifts + 2\vsdim/\nbun$ since
      $\shiftsR = \rdeg[\shifts]{\appbas_1}$;
    \item $\sumVec{\ordersR} \le \vsdim$, hence $\vsdim +
      \sumVec{\shiftsR-\min(\shiftsR)} \le \sumRdeg +2\vsdim \le
      3\sumRdeg$;
    \item $\sysR$ has fewer columns than rows, hence the call at
      Step~\textbf{1.c} will enter Step~\textbf{2}.
  \end{itemize}
  Then, the cost bounds given above hold for Step~\textbf{1.c}: if
  $\lceil\sumRdeg/\nbun\rceil \le \order$ this step is thus the bottleneck of
  Step~\textbf{1}, and if $\lceil\sumRdeg/\nbun\rceil > \order$ we obtain the
  claimed bound $\bigO{\appbastime{\nbun,\vsdim/\nbun} +
  \appbastime{\nbun,\order}}$.
\end{proof}

We remark that it would also be correct, instead of Steps~\textbf{1.d}
and\,\textbf{1.e}, to directly compute and return the product
$\appbas_2\appbas_1$; this uses
$\bigO{\polmatmultime{\nbun,\lceil\sumRdeg/\nbun\rceil}}$ operations and thus
does not impact the cost bound if $\sumRdeg\in\bigO{\vsdim}$.  In addition, for
input instances with $\vsdim \ll \nbun$, one may rather rely on linear algebra
over $\field$ instead of the above algorithm (see Steps~\textbf{1.a},
\textbf{1.b}, and \textbf{1.c} of \cref{algo:fast_pab}).

We now show the upper bound on $\costZhoLabMin{\sumRdeg,\nbun,\order}$ given in
\cref{thm:zhou_labahn}, for the case $\sumRdeg\le \nbun\order$.  Under the
assumption $\hyppolmul$, we obtain
\begin{align*}
  & \appbastime{2^{-k}\nbun,2^k\lceil\sumRdeg/\nbun\rceil} +
  2^k \polmatmultime{2^{-k}\nbun,2^k\lceil\sumRdeg/\nbun\rceil} \\
  & \quad \in \bigOPar{(2^{-k}\nbun)^{\expmatmul} \polmultime{2^k\lceil\sumRdeg/\nbun\rceil}  \log(2^k\lceil\sumRdeg/\nbun\rceil)
    + 2^k (2^{-k}\nbun)^{\expmatmul} \polmultime{2^k\lceil\sumRdeg/\nbun\rceil}} \\
  & \quad \subseteq \bigOPar{\nbun^\expmatmul \polmultime{\lceil\sumRdeg/\nbun\rceil} (2^{-k}(k+\log(\lceil\sumRdeg/\nbun\rceil)) + 1)},
\end{align*}
since $\hyppolmul$ implies in particular $\polmultime{2^k
\lceil\sumRdeg/\nbun\rceil} \in \bigO{2^{(\expmatmul-1)k}
\polmultime{\lceil\sumRdeg/\nbun\rceil}}$.  Since $\sum_{k\ge 0} k2^{-k}$ is
the constant $2$, summing over $0 \le k \le 1+
\log(\order/\lceil\sumRdeg/\nbun\rceil)$ gives the sought bound
\[
  \costZhoLabMin{\sumRdeg,\nbun,\order} \in
  \bigO{\nbun^\expmatmul \polmultime{\lceil\sumRdeg/\nbun\rceil}
  (\log(\lceil\sumRdeg/\nbun\rceil) + \log(\order/\lceil\sumRdeg/\nbun\rceil))}
  = \bigO{\nbun^\expmatmul \polmultime{\lceil\sumRdeg/\nbun\rceil}
  \log(\order)},
\]
valid under $\hyppolmul$ and for an arbitrary order and shift.

We remark that the latter bound is precisely the one which was obtained
\citep[Thm.\,5.3]{ZhoLab12}, under the additional assumptions that
$\sumRdeg\in\bigO{\vsdim}$ and that $\orders = (\order,\ldots,\order) \in
\orderSpace$ with $\nbeq \le \nbun \le \vsdim = \nbeq \order$; in that case the
bound can be written $\bigO{\nbun^\expmatmul \polmultime{\nbeq\order/\nbun}
\log(\order)}$.

\subsection{Weakly unbalanced shift around its maximum value}
\label{subsec:weakly_unbalanced_around_max}

Here, we will only sketch the correctness and cost bound of the algorithm, and
refer to \citep[Sec.\,6]{ZhoLab12} for more details and examples.  Indeed, it
can be noticed that the output column linearization does not modify the order
$\orders$ and does not depend on it.  As a result, generalizing
\citepalias[Algo.\,2]{ZhoLab12} to the case of arbitrary orders was mostly done
in \cref{sec:knowndeg:output_parlin} where the definition and
properties of the output column linearization were presented.

In \cref{algo:zhou_labahn_max}, we interrupt the iterative use of output column
linearization as soon as it becomes more efficient to directly resort to
\algoname{PM-Basis} (Step~\textbf{4}).  We remark that, while this may seem to
differ from \citepalias[Algo.\,2]{ZhoLab12}, it is in fact mentioned in the
proof of \citepalias[Thm.\,6.14]{ZhoLab12} that the algorithm should behave
like this to avoid weakening its efficiency.

\begin{algobox}[htb]
  \algoInfo
  {Minimal basis for small $\sumVec{\!\max(\shifts)-\shifts}$}
  {ShiftAroundMaxAppBasis}
  {algo:zhou_labahn_max}

  \dataInfos{Input}{
    \item order $\orders \in \orderSpace$,
    \item matrix $\sys \in \sysSpace$ with $\cdeg{\sys} < \orders$,
    \item shift $\shifts \in \shiftSpace$.
  }

  \dataInfo{Output}{an $\shifts$-ordered weak Popov basis of $\modApp$.}

  \algoSteps{
    \item $\appbas \assign$ empty matrix in $\polMatSpace[0][\nbun]$
    \item $I \assign \{1,\ldots,\nbun\}$ \eolcomment{indices of rows still to be found}
    \item \algoword{While} $\vsdim + \sumVec{\!\max(\shifts)-\shifts} \le \card{I} \order$:
      \begin{enumerate}[{\bf a.}]
        \item $\degExp \assign 1 + 2 \lfloor \sumVec{\!\max(\shifts)-\shifts} / \card{I} \rfloor$
        \item $(\expand{\shifts},\expandMat, (\quoExp_i)_{1\le i\le\nbun},\expand{\nbun}) \assign
          \algoname{ColParLin}(\subTuple{\shifts}{I},\degExp,\degExp)$
          \eolcomment{see \cref{sec:knowndeg:output_parlin}}
        \item $\expand{\appbas} \assign \algoname{ShiftAroundMinAppBasis}(\orders,\expandMat\matrows{\sys}{I}\bmod\mods,\expand{\shifts})$
        \item $\mat{E} \in \matSpace[\nbun] \assign \diag{e_1,\ldots,e_\nbun}$ with $e_i = 1$ if $i\in I$ and $e_i = 0$ otherwise
        \item \algoword{For} $i\in I$ such that $\shift{i} \ge \max(\shifts)-\degExp$
              or $\rdeg[\expand{\shifts}]{\matrow{\expand{\appbas}}{\quoExp_1+\cdots+\quoExp_i}} > 0$: \\
              \phantom{space} $\matrow{\appbas}{i} \assign
                \matrow{\expand{\appbas}}{\quoExp_1+\cdots+\quoExp_i} \expandMat \mat{E}$;
                $I \assign I \setminus \{i\}$
      \end{enumerate}
    \item \algoword{If} $I \neq \emptyset$: \eolcomment{compute remaining rows via \algoname{PM-Basis}}
      \begin{enumerate}[{\bf a.}]
        \item permute $\orders$ into nonincreasing order, and the columns of
          $\matrows{\sys}{I}$ accordingly
        \item $(\ordersR,\sysR,\shiftsR,\appbas_1) \assign
            \algoname{ReduceColDim}(\orders,\matrows{\sys}{I},\subTuple{\shifts}{I})$
        \item $\appbas_2 \assign \algoname{PM-Basis}(\ordersR,\sysR,\shiftsR)$
        \item $\minDegs_1 \assign$ diagonal degrees of $\appbas_1$;~
              $\minDegs_2 \assign$ diagonal degrees of $\appbas_2$
        \item  $\expand{\appbas}\assign\algoname{KnownDegAppBasis}(\orders,\matrows{\sys}{I},\subTuple{\shifts}{I},\minDegs_1+\minDegs_2)$
        \item $\matrows{\appbas}{I} \assign \expand{\appbas} \,\, \diag{e_1,\ldots,e_\nbun}$, where $e_i = 1$ if $i\in I$ and $e_i = 0$ otherwise
      \end{enumerate}
    \item \algoword{Return} $\appbas$
}
\end{algobox}

We recall that $\costZhoLabMin{\cdot}$ was defined in
\cref{eqn:cost_zholab_min}.

\begin{proposition}
  \label{prop:algo:zhou_labahn_max}
  \cref{algo:zhou_labahn_max} is correct.  Let $\vsdim = \sumVec{\orders}$, let
  $\sumDistMax = \vsdim + \sumVec{\!\max(\shifts)-\shifts}$, and let $\order =
  \max(\orders)$.  If $\sumDistMax > \nbun\order$, \cref{algo:zhou_labahn_max}
  uses $\bigO{\appbastime{\nbun,\lceil\vsdim/\nbun\rceil} +
  \appbastime{\nbun,\order}}$ operations in $\field$.  If $\sumDistMax \le
  \nbun\order$, it uses
  \[
    \bigOPar{
      \appbastime{\mu,\lceil\vsdim/\mu\rceil}
      + \appbastime{\mu,\order}
      + \sum_{k=0}^{\lfloor\log_2(\nbun\order/\sumDistMax)\rfloor}
      \costZhoLabMin{\sumDistMax,2^{-k}\nbun,\order} 
      }
  \]
  operations in $\field$, where $\costZhoLabMin{\cdot}$ is defined as in
  \cref{eqn:cost_zholab_min} and $\mu$ is the cardinality of the set $I$
  after Step~\textbf{3} has been performed; it is such that $\mu <
  \sumDistMax/\order$.
\end{proposition}
\begin{proof}
  First, if $\sumDistMax > \nbun\order$, the loop at Step~\textbf{3} is not
  entered, and at Step~\textbf{4} we have $I= \{1,\ldots,\nbun\}$; in
  particular, $\matrows{\appbas}{I} = \appbas$ and Step~\textbf{4.f} simply
  amounts to $\appbas \assign \expand{\appbas}$.  In this case, the correctness
  and cost bound follow from \cref{prop:algo:reduce_coldim,prop:algo:dac_ab},
  the fourth item of \cref{lem:recursiveness}, and
  \cref{prop:algo:knowndeg_pab}.

  From now on, suppose $\sumDistMax \le \nbun\order$.  The same results prove
  the correctness of Step~\textbf{4} while \cref{lem:parlin_appbas} proves that
  of Step~\textbf{3}, using in addition \citep[Thm.\,6.11]{ZhoLab12} to show
  that we may discard the rows of $\sys$ with index not in $I$
  (Steps~\textbf{3.b} and~\textbf{4.b}) and fill corresponding columns of
  $\appbas$ with zeroes (multiplication by $\mat{E}$ in Step~\textbf{3.e} and
  by the diagonal in~\textbf{4.f}).

  Furthermore, the above propositions show that Step~\textbf{4} uses
  $\bigO{\appbastime{\mu,\lceil\vsdim/\mu\rceil} +
  \appbastime{\mu,\order}}$ operations; since the loop at Step~\textbf{3}
  has exited, we have $\sumDistMax > \mu\order$.

  Concerning Step~\textbf{3}, the main point is that the cardinality of $I$ is
  at least halved at the end of each iteration of the \algoword{While} loop.
  Indeed, let $c>0$ be the cardinality of $I$ at the beginning of an iteration;
  hence $\degExp > 2 \sumVec{\!\max(\shifts)-\shifts} / c$.  Then, at the end
  of the iteration, we have that $I$ is contained in $\{i\in \{1,\ldots,\nbun\}
  \mid \shift{i} < \max(\shifts)-\degExp\}$ which has cardinality at most
  $\sumVec{\!\max(\shifts)-\shifts} / \degExp$.  Thus, we obtain $\card{I} \le
  \sumVec{\!\max(\shifts)-\shifts} / \degExp < c/2$.

  As a consequence, the worst case in terms of cost occurs when $\card{I}$ is
  divided by only slightly more than $2$ at each iteration.  Then, this
  cardinality is about $2^{-k}\nbun$ at the end of the $k$th iteration of the
  \algoword{While} loop.  This iteration then uses
  $\costZhoLabMin{\sumDistMax,2^{-k}\nbun,\order}$ operations in $\field$; this
  follows from the bounds on $\expand{\nbun}$ and $\expand{\shifts}$ in
  \cref{lem:parlin_app} and from the cost of Step~\textbf{3.c} given in
  \cref{prop:algo:zhou_labahn_min}.  We remark that the condition $\sumDistMax
  \le \card{I} \order$ of the loop precisely ensures that we are in the case
  ``$\sumRdeg \le \nbun\order$'' of the latter proposition.
\end{proof}

%
%
%

To conclude this section, we derive the upper bound given in the second item of
\cref{thm:zhou_labahn} under the assumption $\hyppolmul$.  We first remark that
we have $\lceil2^k\sumDistMax/\nbun\rceil \le
2^k\lceil\sumDistMax/\nbun\rceil$, since $\lceil \lceil\alpha r\rceil /
\alpha\rceil = \lceil r \rceil$ holds for any real number $r$ and any positive
integer $\alpha$.  Besides, the assumption $\hyppolmul$ implies that
$\polmultime{2^k\lceil\sumDistMax/\nbun\rceil} \in \bigO{2^{(\expmatmul-1)k}
\polmultime{\lceil\sumDistMax/\nbun\rceil}}$.  Then, the first item in
\cref{thm:zhou_labahn} yields
\[
  \costZhoLabMin{\sumDistMax,2^{-k}\nbun,\order}
  \in \bigOPar{ (2^{-k}\nbun)^\expmatmul
    \polmultime{\lceil2^k\sumDistMax/\nbun\rceil} \log(\order) }
  \subseteq \bigOPar{ 2^{-k} \nbun^\expmatmul
    \polmultime{\lceil\sumDistMax/\nbun\rceil} \log(\order) },
\]
from which we obtain
\[
  \sum_{k=0}^{\lfloor\log_2(\nbun\order/\sumDistMax)\rfloor}
  \costZhoLabMin{\sumDistMax,2^{-k}\nbun,\order}
  \;\in \bigOPar{\nbun^\expmatmul \polmultime{\lceil\sumDistMax/\nbun\rceil}
  \log(\order)}.
\]
Now, using $\order \le \frac{\nbun}{\mu} \lceil \frac{\mu}{\nbun}\order\rceil
\le \frac{\nbun}{\mu} \lceil \frac{\sumDistMax}{\nbun}\rceil$ and the
assumption $\hyppolmul$ leads to $\polmultime{\order} \in
\bigO{(\nbun/\mu)^{\expmatmul-1} \polmultime{\lceil \sumDistMax/\nbun\rceil}}$,
and therefore we also have
\[
  \appbastime{\mu,\order} \in
  \bigOPar{\nbun^{\expmatmul-1} \mu \polmultime{\lceil\sumDistMax/\nbun\rceil} \log(\order) }
  \subseteq
  \bigOPar{\nbun^{\expmatmul} \polmultime{\lceil\sumDistMax/\nbun\rceil} \log(\order) }.
\]
This completes the proof of the upper bound in the second item of
\cref{thm:zhou_labahn}, since we have $\appbastime{\mu,\lceil\vsdim/\mu\rceil}
\in \bigO{\mu^\expmatmul \polmultime{\lceil\vsdim/\mu\rceil}
\log(\lceil\vsdim/\mu\rceil)}$.

One can simplify the latter bound slightly further, in order to facilitate the
comparison with \citep[Thm.\,6.14]{ZhoLab12}. Indeed, we have $\lceil
\frac{\vsdim}{\mu} \rceil \in \bigO{\frac{\nbun}{\mu}
\lceil\frac{\vsdim}{\nbun}\rceil}$ since $\nbun \ge \mu$.
Then, using
\[
  \polmultime{\lceil\vsdim/\mu\rceil} \log(\lceil\vsdim/\mu\rceil)
  \in \bigO{(\nbun/\mu)^{\expmatmul-1} \polmultime{\lceil\vsdim/\nbun\rceil} \log(\lceil\vsdim/\nbun\rceil)},
\]
which is a minor strengthening of the assumption $\hyppolmul$,
the last bound in \cref{thm:zhou_labahn} becomes:
\begin{align*}
  & \bigO{\nbun^{\expmatmul} \polmultime{\lceil\sumDistMax/\nbun\rceil}
  \log(\order) + \mu^\expmatmul \polmultime{\lceil\vsdim/\mu\rceil} \log(\lceil\vsdim/\mu\rceil)} \\
  \subseteq\; & \bigO{\nbun^{\expmatmul}
  \polmultime{\lceil\sumDistMax/\nbun\rceil} \log(\order) + \nbun^\expmatmul
\polmultime{\lceil\vsdim/\nbun\rceil} \log(\lceil\vsdim/\nbun\rceil)} \\
    \subseteq\; & \bigO{\nbun^{\expmatmul} \polmultime{\lceil\sumDistMax/\nbun\rceil} \log(\order\lceil\vsdim/\nbun\rceil)}.
\end{align*}
  
Finally, we remark that if $\nbeq \le \nbun$, then we have $\vsdim \le
\nbun\order$ and therefore this upper bound becomes $\bigO{\nbun^{\expmatmul}
\polmultime{\lceil\sumDistMax/\nbun\rceil} \log(\order) }$.  This matches the
bound in \citep[Thm.\,6.14]{ZhoLab12}, where $\nbeq \le \nbun$ is assumed. We
further note that in the specific case considered in this reference (the order
$\orders$ is uniform and $\nbeq \le \nbun$), the algorithm stops as soon as the
row and column dimensions become roughly equal, and therefore it does not need
to rely on column dimension reduction; thus, in this case, the term
$\appbastime{\mu,\lceil\vsdim/\mu\rceil}$ can be removed from the above cost
bounds.

\section*{Acknowledgement}

The authors want to thank \'Eric Schost for his useful comments. The research
leading to these results was partly done while Vincent Neiger was affiliated
with the Department of Applied Mathematics and Computer Science of the
Technical University of Denmark, with funding from the People Programme (Marie
Curie Actions) of the European Union's Seventh Framework Programme
(FP7/2007-2013) under REA grant agreement no. 609405 (COFUNDPostdocDTU).

\bibliographystyle{elsarticle-harv} 

\end{document}